\documentclass[11pt]{article}
\usepackage[utf8]{inputenc}
\usepackage{gal_basic}
\usepackage[numbers]{natbib} 
\usepackage{booktabs}
\usepackage{format}
\usepackage{amsfonts,amsmath,amssymb,amsthm,boxedminipage,color,url,fullpage}
\usepackage{enumerate,paralist}
\usepackage{hyperref}
\usepackage[labelfont=bf]{caption}
\usepackage{aliascnt,cleveref}
\usepackage{graphicx}
\usepackage{amsmath}
\usepackage{amssymb}
\graphicspath{ {./images/} }
\usepackage{algorithm}
\usepackage{algorithmic}

\begin{document}
\pagenumbering{gobble}

\title{On Fairness and Stability in Two-Sided Matchings}
\author{Gili Karni\thanks{Weizmann Institute of Science. Email: \url{gili.karni1210@gmail.com}. Research supported by the
Israel Science Foundation (grant number 5219/17), and by the Simons Foundation Collaboration on the Theory of
Algorithmic Fairness.} \and 
Guy N. Rothblum\thanks{Weizmann Institute of Science. Email: \url{rothblum@alum.mit.edu}. This project has received funding from the European Research
Council (ERC) under the European Union’s Horizon 2020 research and innovation programme (grant agreement No. 819702), from the
Israel Science Foundation (grant number 5219/17), and from the Simons Foundation Collaboration on the Theory of
Algorithmic Fairness.}  \and Gal Yona\thanks{Weizmann Institute of Science. Email: \url{gal.yona@weizmann.ac.il} Supported by the European Research Council (ERC) (grant agreement No. 819702), by the Israel Science Foundation (grant number 5219/17), by the Simons Foundation Collaboration on the Theory of Algorithmic Fairness, by the Israeli Council for Higher Education via the Weizmann Data Science Research Center, by a Google PhD fellowship, and by a grant from the Estate of Tully and Michele Plesser.} }
\date{}

\maketitle

\begin{abstract}

There are growing concerns that algorithms, which increasingly make or influence important decisions pertaining to individuals, might produce outcomes that discriminate against protected groups. We study such fairness concerns in the context of a two-sided market, where there are two sets of agents, and each agent has preferences over the other set. The goal is producing a matching between the sets. Throughout this work, we use the example of matching medical residents (who we call ``doctors'') to hospitals.  This setting has been the focus of a rich body of work. The seminal work of Gale and Shapley  formulated a {\em stability} desideratum, and showed that a stable matching always exists and can be found in polynomial time. 

With fairness concerns in mind, it is natural to ask: might a stable matching be discriminatory towards some of the doctors? How can we obtain a {\em fair} matching? The question is interesting both when  hospital preferences might be discriminatory, and also when each hospital's preferences are fair.

We study this question through the lens of metric-based fairness notions (Dwork {\em et al.} [ITCS 2012] and Kim {\em et al.} [ITCS 2020]). We formulate appropriate definitions of fairness and stability in the presence of a similarity metric, and ask: does a fair and stable matching always exist? Can such a matching be found in polynomial time? Can classical Gale-Shapley algorithms find such a matching? Our contribution are as follows:

\begin{itemize}

\item {\bf Composition failures for classical algorithms.} We show that composing the Gale-Shapley algorithm with fair hospital preferences can produce blatantly unfair outcomes. 

\item {\bf New algorithms for finding fair and stable matchings.} Our main technical contributions are efficient new algorithms for finding fair and stable matchings when: $(i)$ the hospitals' preferences are fair, and $(ii)$ the fairness metric satisfies a strong ``proto-metric'' condition: the distance between every two doctors is either zero or one. In particular, these algorithms also show that, in this setting, fairness and stability are compatible.
 
\item {\bf Barriers for finding fair and stable matchings in the general case.} We show that if the hospital preferences can be unfair, or if the metric fails to satisfy the proto-metric condition, then no algorithm in a natural class can find a fair and stable matching. The natural class includes the classical Gale-Shapley algorithms and our new algorithms.


\end{itemize}

\end{abstract}

\newpage
\tableofcontents
\newpage

\pagenumbering{arabic}

\section{Introduction}

Algorithms are increasingly influencing or replacing human decision-makers in several sensitive domains, including the criminal justice system, online advertising, and medical risk prediction. Along with the benefits of automated decision-making, there is also a potential risk of discrimination towards groups of individuals, which might be illegal or unethical \cite{o2016weapons, bogen2019all}. Examples for unintended but harmful behavior have been shown to happen in algorithms that allocate resources in domains such as online advertising \cite{ali2019discrimination, tobin2019hud}, healthcare systems \cite{obermeyer2019dissecting, rajkomar2018ensuring}, and more.



Many resource allocation problems, such as online advertising and assigning hospitals to medical students for residency, can be viewed as \emph{two-sided markets}. In this setting, there are two sets of agents, that both have preferences over the other set. We seek a matching: a symmetric allocation in which every member from the first set is matched to a member of the other set (ads to users, drivers to customers, hospitals to students). A well-studied desideratum for such two-sided matchings is \emph{stability} \cite{gale1962college}. A matching is stable only if it leaves no pair of agents on opposite sides of the market who are not matched to each other, but would both prefer to be. In their seminal work, Gale and Shapley \cite{gale1962college} proved that in two-sided markets a stable matching always exists and can be found efficiently using a simple procedure known as the \emph{Gale-Shapley} algorithm. 
Stability represents the incentive to participate in the matching, i.e., no unmatched pair can {\em both} improve their situation by being matched to each other. 
This problem was introduced by \cite{gale1962college} and has been studied broadly, see e.g. the books by Roth and Sotomayor \cite{roth1992two}, by Knuth \cite{knuth1997stable}, and by Gusfield and Irving  \cite{gusfield1989stable}. In particular, Roth \cite{roth1984evolution, roth1986allocation} studied the real-world setting of matching medical residents to hospitals. We use this as a running example throughout our work.

There are many instances, however, where we may seek to look beyond utility-based desiderata such as stability. Suppose, for example, in the problem of assigning hospitals for residency, that given two equally-qualified residents, the preferences of some hospitals are discriminatory, e.g.,  they display a preference towards residents who don't live in certain neighborhoods, or residents without children. In this case, a stable matching may be undesirable.
Our goal is preventing discrimination in the two-sided matching setting, where we match individuals to resources such as hospitals, ads, or schools.
The emerging literature about algorithmic fairness typically focuses on one-sided allocation problems (e.g., supervised learning).
The two-sided setting is unique in that we have two sets of agents, both of which have preferences over the other side.
Studying fairness in the context of two-sided markets is well-motivated since many of the 
the resource allocation problems captured by two-sided markets have high stakes for the individuals involved.

\subsection{This Work: Fair and Stable Matchings}
 
We would like the matching to be both fair and stable.
This requires appropriate definitions of fairness and of stability.
In this work, we embark on a study of this question;
we set out to present such definitions and study when and how fairness and stability can be compatible.
 
A fundamental prerequisite is a notion of what makes a given allocation fair. We build on the approach of \emph{individual fairness} (IF) \cite{dwork2012fairness}, which assumes the existence of a task-specific similarity metric that measures how ``similar'' two individuals are (say, how similarly qualified they are). We would like to ensure that the eventual allocation satisfies \emph{preference-informed individual fairness} (PIIF) \cite{kim2020preference}, which roughly means that there is no envy between similar individuals (i.e., a doctor $i_1$ will not prefer the outcome that a similarly qualified doctor $i_2$ receives).
In the case where similar doctors have the same preferences, a fair deterministic solution does not exist. Thus, we focus on finding a fair {\em distribution} over matchings.

In this work, we begin to chart the landscape of \emph{fair and stable} matchings in the context of two-sided markets. 
Specifically,  our contributions are:

\begin{itemize}
    \item \textbf{Fair hospitals might arrive at unfair allocations}. Our exploration builds on an important but counter-intuitive result. We show that even when the hospitals' preferences satisfy a very strong notion of fairness (strong indifference between equally-qualified candidates), running the classic Gale-Shapley algorithm  is \emph{not} guaranteed to result in a fair allocation. This can be seen as a failure of \emph{composition}: even when the inputs to an algorithm are ``fair'', the result may not be. Composition has previously been highlighted as a challenge in the context of designing fair algorithms \cite{dwork2018fairness}. This stands in stark contrast to e.g. the landscape around privacy (where \emph{differential privacy} enjoys graceful composition properties, facilitating the design of complex private algorithms using private ``building blocks''). 

    
    \item \textbf{New algorithms assuming fair preferences and ``simple'' metrics}: Our main technical contribution is a strong positive  result establishing that fairness and stability are compatible, when (i) the preferences of the hospitals are fair, and (ii) we place strong assumptions on the structure of the similarity metric defining fairness. Specifically, our results hold for a class of similarity metrics that we refer to as ``proto-metrics'', in which the distances between every pair of individuals must be either zero or one. Importantly, we show that not only do fair and stable solutions exist, but they can be found efficiently: we provide new algorithms, inspired by the Gale-Shapley algorithm, which obtain fair and stable solutions.
    
    \item \textbf{Barriers for compatibility in the general case}. A natural question is whether the assumptions that the hospital preferences must be fair and that the metric must be ``simple'' are necessary. We provide a partial answer by demonstrating a rich class of natural algorithms (extending our algorithms and the original Gale-Shapley algorithm), and proving no algorithm in this class can guarantee a fair and stable solution if either assumption is removed.

    \item \textbf{New notions of stability}.  The above barrier results require formalizing new notions of stability for this more general setting. We aim to formalize notions of stability that provide a reasonable guarantee for the hospitals without being trivially incompatible with fairness. For example, if the hospital that is preferred by all the doctors has discriminatory preferences, it might be the case that every stable matching is blatantly unfair. Similar situations can arise even when the hospital preferences are fair, if the metric is not a proto-metric. Thus we aim for a relaxed goal, which we find suitable for this setting: obtaining a {\em fair matching} where no hospital prefers any {\em fair alternative} to the allocation it received in the matching.
    Formalizing this intuition presents several subtleties. Thus, we present two definitions, one for the case of unfair preferences and a proto-metric, and a weaker definition for the case of general metrics (we use this weaker definition in our negative results).
    
    
\end{itemize}

\subsection{Related Work}
Our work bridges the established literature on matchings, originating in the seminal work of Gale and Shapley \cite{gale1962college}, and the emerging literature on algorithmic fairness, which seeks to formalize and mitigate discrimination in algorithmic decision-making systems \cite{barocas2016big}. The latter has focused on formalizing and studying different notions of fairness, and understanding the tensions between them and possible ``accuracy'' based desiderata \cite{chouldechova2017fair, kleinberg2016inherent, hardt2016equality}. One popular approach to quantifying unfairness are group-based definitions, in which the objective is equalizing some statistic across a fixed collection of ``protected'' groups. In this work we build on a different approach, using the notion of preference-informed individual fairness \cite{kim2020preference}, which combines the individual-based fairness notion proposed in \cite{dwork2012fairness} with the notion of envy freeness from game theory \cite{foley1967resource, varian1974efficiency}. Recent work \cite{dwork2018fairness} has considered the question of composition in the context of fairness, showing that fairness may fail to compose in a variety of settings. Our results complement their findings by demonstrating such failures may also occur for a natural and widely popular algorithm. 

Fairness has also been studied previously in the context of two-sided matchings. \cite{gusfield1989stable} introduced the \emph{equitable stable marriage problem}, where the objective is to minimize (across stable matchings) the difference between the sum of rankings of each set over their matches. This implies avoiding unequal degrees of happiness among the two sides, and \cite{giannakopoulos2015equitable} showed an algorithm that finds such a solution. This is motivated by the fact that stable matchings are in general not unique (in fact, their number could even be exponential in the size of the sets \cite{knuth1997stable, thurber2002concerning}), and that the algorithm introduced by \cite{gale1962college} finds the optimal stable matching for one of the sets and the worst stable matching for the other. Our approach is inherently different in that we seek a matching that is fair within one of the sets rather than across sets. \cite{freeman2021two} studied the problem of finding an allocation that satisfies \emph{envy-freeness up to one good} (EF1), for both sets simultaneously. This is a relaxation of the classic notion of envy-freeness, that requires that any pairwise envy can be eliminated by removing a single item from the envied agent’s allocation. This formulation is similar to ours in that it tries to guarantee fairness for each set separately. However, our work is conceptually different in that we consider metric-based fairness requirements as well as our focus on the compatibility between fairness and stability, and technically different because in the one-to-one setting that we study the EF1 objective becomes vacuous. Finally, \cite{suhr2019two} study the problem of matching drivers to passengers in ride-hailing platforms, trying to achieve equal income for the drivers. We focus on a different fairness notion, that is appropriate when the preferences of the individuals are potentially diverse.

\subsection{Defining Fairness and Stability}
We focus on an asymmetric setting, such as assigning hospitals to medical students for residency.
Usually in this problem each hospital can be assigned to multiple residents. 
For simplicity, we focus on the setting where each hospital is assigned to a single resident.
We refer to the medical students as doctors. We want to guarantee fairness for the doctors and stability, to be defined, for the hospitals. 
Since the resources are limited and sometimes multiple individuals want a single resource, a deterministic fair solution does not always exist. Thus, we focus on finding probabilistic solutions where the allocation is a distribution over matchings.

\subsubsection{Fairness Requirements}
\label{sec:intro:fairness}

To enforce fairness, we assume that we are given an unbiased similarity metric for the doctors. 
We wish to ensure a metric-based fairness guarantee such as \emph{individual fairness} (IF) \cite{dwork2012fairness}, that is, similar individuals should have a similar outcome.
For instance, the metric can represent differences in the GPA of medical students. Then, students with a GPA of 5.0 should be assigned to a prestigious hospital with the same probability, and students with a GPA of 4.0 can be assigned to that prestigious hospital with a lower probability.
However, sometimes similar individuals have different preferences, e.g., students with the same GPA can prefer different hospitals for reasons such as geographic location or specialization in a particular field. 
In that case, we want to allow these similar individuals to have different outcomes. Thus, we find that \emph{preference-informed individual fairness} (PIIF) \cite{kim2020preference} is a more appropriate fairness notion. PIIF is a relaxation of two fairness requirements: (1) individual fairness; and (2) envy freeness. Envy freeness (EF) requires that no individual prefers the outcome of another individual over their own outcome. In PIIF, we allow deviations from EF, i.e., we allow for an individual $i$ to envy the outcome of another individual $j$, conditioned on the fact that $j$'s outcome can be changed by no more than the distance between $i$ and $j$ to an outcome that $i$ will not envy. We allow deviations from IF, so long as they are in line with individuals’ preferences.

\paragraph{Our focus: proto-metrics.} In many of our results, we focus on the special case where the metric is a {\em proto-metric}: the distances between all the doctors are either 0 or 1. 
In this setting, the doctors are divided into clusters of similar doctors. PIIF means that we require envy-freeness between the doctors in each cluster. 
However, there are no constraints on the allocations of doctors from different clusters. 
We focus on this setting throughout the introduction, {\em except} in Section \ref{sec:intro_neg_results}.

To compare the preferences of the doctors within each cluster, we use \emph{stochastic domination}.
We say one distribution stochastically dominates another if, for every outcome, the probability of getting this outcome or a better one is lower-bounded by the corresponding probability in the other distribution. 
See Definition \ref{def:stochastic_domination} for a formal definition of stochastic domination.
In the presence of a proto-metric, a PIIF allocation implies that each doctor's allocation stochastically dominates the allocation of any other doctor in its cluster.

We find that there are natural settings where the restriction to proto-metrics is reasonable. For example, consider a setting where hospitals are only allowed to distinguish between medical students based on a specialization in their medical studies, e.g., neurology, cardiology, etc.\footnote{For example, in the current Israeli system, hospitals are not allowed to express preferences over the medical students \cite{bronfman2015assigning}. Allowing hospitals to have fair preferences, even in a limited way, could improve  outcomes for the hospitals.} Different hospitals may prefer different specializations. A proto-metric can partition the candidates based on their specialization (one could also add  merit-based sub-categories within each specialization).

\subsubsection{Fair Hospital Preferences}
\label{par:intro:fair_prefs}

We distinguish between two scenarios: (1) Hospitals have fair preferences over the doctors.
For instance, we allow a hospital to prefer some doctors over others because they have a higher GPA or good grades in a particular topic but not because they do not have children or belong to a specific ethnic group.
The fairness requirement for the hospitals' preferences can be formalized in different ways (see below). For now, we emphasize that, even if we impose strict fairness requirements on the individual hospitals' preferences, our results show that obtaining a fair and stable solution can be far from trivial.
(2) Hospital preferences might be discriminatory. This case is interesting since decision-makers can be discriminatory, e.g., because of biased data or prejudice. Finding a fair and stable solution in the presence of unfair hospital preferences is even more challenging (in particular, it is not clear how to define stability), see Section \ref{sec:intro_neg_results}.

Requiring fair hospital preferences is an assumption or restriction on the input to the algorithm that attempts to find a fair and stable allocation (the algorithm's input is the preferences of the doctors and the hospitals). We find the restriction to be natural and well motivated: if our goal is finding a matching that is fair to the doctors, it makes sense to ask the hospitals to indicate fair preferences. As remarked above, even under strong fairness restrictions on the hospitals' preferences, finding a fair and stable matching is challenging.

\paragraph{Our focus: strictly fair hospital preferences.} Focusing on the proto-metric setting (see above), we formalize a strong notion of {\em strictly individually fair} hospital preferences: each hospital can have arbitrary (deterministic) preferences over the clusters, but must be completely indifferent between every two doctors that are in the same cluster.

To reason about the (metric-based) fairness of a hospital's preferences, we view them as probabilistic, i.e., a distribution over ordinal preferences. Strictly IF preferences induce such a distribution, where the ``external'' ordering of the clusters is deterministic, and the ``internal'' ordering within each cluster is a uniformly random permutation of the doctors (the random internal ordering captures  indifference between doctors in the same cluster).\footnote{We note that while strict IF could also be captured with deterministic preferences that allow ties, as was considered in \cite{irving1994stable, kiselgof2014matchings}, reasoning about preferences as distributions over ordinal preference lists allows extensions to other notions of fair preferences, as well general metrics.}

\paragraph{Beyond strictly IF preferences.} Strict IF is a strong restriction or assumption on the hospitals' preferences. This makes our negative results (the failure of the classic Gale Shapley algorithms) stronger: the algorithms fail even under the strong restriction on hospital preferences. For positive results, showing an algorithm that works given a more relaxed notion of fair preferences (even in the proto-metric setting) is an interesting question for future work. 

More generally, we could hope to design algorithms that work with fair hospital preferences under general metrics, or with completely {\em unfair} hospital preferences. This raises subtle difficulties in the definition of stability and encounters natural barriers, see Section \ref{sec:intro_neg_results}.

\subsubsection{Stability Requirements}
\label{subsec:intro:stability}

The uniform distribution over all outcomes is always trivially fair. However, in addition to fairness for the doctors, we want some guarantee for the hospitals. Stability is one such guarantee, which has been studied extensively in the classical setting, with deterministic allocations and without fairness constraints. The classical stability guarantee ensures that there are no ``blocking pairs''. That is, there are no pairs of a hospital and a doctor that prefer each other over their match and are not matched. We extend this notion to the setting where there are fairness constraints, and the allocations are probabilistic. 

\paragraph{Our focus: stability under proto metrics, strictly IF hospital preferences.} For the case of a proto-metric and strict IF hospital preferences, there is a natural extension to the classical notion of stability. We say that a (probabilistic) allocation is unstable if there is a matching with non-zero probability under the probabilistic allocation, where there exists a blocking pair $(d,h)$ that (strongly) prefer each other over their match. We remark that, since we assume the hospital preferences are strictly IF, this can only happen if the doctor $d$ and the doctor who is matched to the hospital $h$ (denote them $d'$), are in different clusters (otherwise $h$, whose preferences are strictly IF, will not prefer $d$ to $d'$).

For instance, suppose there two clusters $i = \set{i_1, i_2}$ and $j = \set{j_1, j_2}$, such that hospital $h$ prefers cluster $i$ over cluster $j$, and doctor $i_1$ prefers hospital $h$ over any other hospital. Assume a probabilistic allocation where, for some matching in the support, hospital $h$ is matched to $j_1$. The pair $h$ and $i_1$ form a blocking pair, since they strongly prefer each other over their match. On the other hand, even if hospital $h$ is matched to doctor $i_2$ in every matching in the support, hospital $h$ and doctor $i_1$ do not form a blocking pair, since hospital $h$ does not have a strong preference between doctors $i_2$ and $i_1$ (since they are in the same cluster).

\paragraph{Approximate stability.} For a probabilistic allocation, we can relax the stability requirement by allowing a blocking pair to occur with small probability. The probability is over the choice of a matching drawn from the probabilistic allocation. We say that an allocation is $\tau$-approximately 
stable if the probability that {\em no} blocking pair occurs is at least $(1 - \tau)$ . 

\paragraph{Beyond proto-metrics and fair preferences.} Defining an appropriate notion of stability under unfair hospital preferences or general metrics is considerably more challenging. We provide a definition for unfair hospital preferences in the presence of a proto-metric. We also formalize a minimal weak stability requirement for general metrics (we use this requirement to show negative results). See Sections \ref{subsec:intro:stability_unfair_prefs} and \ref{subsec:intro:stability_gen_metrics}.





\subsection{Fair Preferences Do Not Guarantee a Fair Allocation}
\label{subsec:intro:composition}
Focusing on the setting of a proto-metric and strictly IF hospital preferences,
a natural way to achieve stability and fairness is to use the probabilistic form of strict IF preferences (see Section \ref{par:intro:fair_prefs}),
sample the hospitals' preferences and run the Gale-Shapley algorithm over these sampled preferences.
The probabilistic allocation is the random variable defined by this procedure.
In Section \ref{sec:failure_in_composition}, we show that while this produces a stable probabilistic allocation, it can lead to unfair outcomes, even when the hospital preferences are themselves strictly IF.

To explain this negative result, we first present the Gale-Shapley algorithm. The algorithm is not symmetric: one of the sets is the proposing set, and the other is the accepting set. Here we present the variant where the doctors are the proposing set: At the initialization, no doctor is matched. The algorithm terminates when each doctor is matched to a hospital. Until then, at each round, an unmatched doctor $d$ is chosen arbitrarily. This doctor $d$ proposes to a hospital $h$, where $d$ chooses $h$ as the most preferred hospital that did not reject it yet. Then, hospital $h$ has to decide whether to accept doctor $d$ or reject it. If hospital $h$ is unmatched too, it will always accept. If hospital $h$ is already matched to a doctor $d'$, it will accept only if it prefers doctor $d$ over doctor $d'$. Otherwise, it will reject. If hospital $h$ accepts doctor $d$, it rejects doctor $d'$.

\begin{theorem}[Running Gale-Shapley over fair preferences Gale-Shapley over fair preferences does not guarantee fairness (informal)]
\label{intro:thm:composition}
The algorithm that generates a probabilistic allocation by sampling hospital preferences from a strict individually fair distribution and running the Gale-Shapley algorithm over the sampled preferences is not fair.
\end{theorem}

\begin{proof}[Proof sketch.]
Consider the example of three doctors $i_1, i_2, j$, three hospitals $A, B, C$, and the proto-metric $d$, where $d(i_1,i_2) = 0$ and $d(i_1, j) = d(i_2,j) =1$. The doctor preferences are
\begin{equation*}
    A \succ_{i_1} B \succ_{i_1} C, \quad
    A \succ_{i_2} C \succ_{i_2} B, \quad
    C \succ_j A \succ_j B.
\end{equation*}
The hospitals have strictly individually fair preferences: (note that $i_1$ and $i_2$ are interchangeable)
\begin{equation*}
    \begin{cases}
    j \succ_A i_1 \succ_A i_2, & \textit{w.p. } 1/2 \\
    j \succ_A i_2 \succ_A i_1, & \textit{w.p. } 1/2
    \end{cases}, \quad
    \begin{cases}
    j \succ_B i_1 \succ_B i_2, & \textit{w.p. } 1/2 \\
    j \succ_B i_2 \succ_B i_1, & \textit{w.p. } 1/2
    \end{cases}, \quad
    \begin{cases}
    i_1 \succ_C i_2 \succ_C j, & \textit{w.p. } 1/2 \\
    i_2 \succ_C i_1 \succ_C j, & \textit{w.p. } 1/2
    \end{cases}.
\end{equation*}

Running the algorithm described above that samples the preferences and uses the doctor-propose Gale-Shapley algorithm results in the allocation
\begin{equation*}
    \begin{cases}
    (i_1, B), (i_2, A), (j, C),& \textit{w.p. } 1/2, \\
    (i_1, B), (i_2, C), (j, A),& \textit{w.p. } 1/2.
    \end{cases}
\end{equation*}

Hospital $A$ is the most preferred hospital by doctor $i_1$, but doctor $i_2$ (who is similar to doctor $i_1$) is matched to hospital $A$ with higher probability. Thus, this allocation is unfair.
\end{proof}

For a more detailed discussion, see \Cref{thm:doctor_propose_gs_fails}.
In Section \ref{sec:failure_in_composition}, we show a negative example for the hospital-propose Gale-Shapley variant. 
In these examples, although all the participants acted fairly, the outcome was unfair; this joins existing work on fairness failures under composition \cite{dwork2018fairness, dwork2020individual}. 

\paragraph{Digest: towards fairness.}
In the counter-example presented above, when in the sampled preferences hospital $A$ prefers doctor $i_2$ over doctor $i_1$, the algorithm matches hospital $A$ to doctor $i_2$, without making sure that in the corresponding case, where hospital $A$ prefers doctor $i_1$ over doctor $i_2$, it would match hospital $A$ to doctor $i_1$.
Intuitively, we would like the algorithm to make this decision simultaneously, i.e., to either match both doctors $i_1$ and $i_2$ to hospital $A$ with a certain probability or not to match neither of them.
More generally, when matching a doctor $d$ to a hospital $h$, we must ensure that all the doctors in $d$'s cluster have the same opportunity to be matched to this hospital.

\subsection{Positive Results: Algorithms for Fair and Stable Allocations}
\label{subsec:intro:algorithms}
In Section \ref{sec:piif_algs}, we present generalizations of the Gale-Shapley algorithm that achieve both fairness and stability, up to a small error. 

\begin{theorem}[Compatibility of fairness and stability (informal)]
There exists an efficient algorithm that, given a proto-metric, strictly IF hospital preferences, arbitrary doctor preferences, and an approximation parameter $\tau \in (0,1)$, always finds a $\tau$-approximately fair and $\tau$-approximately stable allocation. The algorithm's running time is polynomial in $(1/\tau)$.
\end{theorem}

The allocation returned by the hospital-first variant of the algorithm is (perfectly) PIIF and $\tau$-approximately stable. The allocation returned by the doctor-first variant of the algorithm is $\tau$-approximately PIIF and $\tau$-approximately stable. See the definition of $\tau$-approximate stability in Section \ref{subsec:intro:stability}. $\tau$-approximate fairness means that for every pair $i_1,i_2$ of doctors in the same cluster, $i_2$'s allocation is $\tau$-close (in statistical distance) to an allocation that $i_1$ doesn't envy. 

In these algorithms, 
we allow the parties to propose probability mass to each other.
In the variant where the hospitals propose, they propose to clusters instead of individuals, so that all doctors in the cluster have an equal opportunity to accept. 
In the variant where the doctors propose, the hospitals accept an allocation that will not cause envy within clusters, e.g., in the counter example of Theorem \ref{intro:thm:composition}, hospital $A$ can either accept both doctors $i_1$ and $i_2$ or reject both of them.
If the distances between all the doctors are 1, these algorithms are identical to the Gale-Shapley algorithm.
The formal description of the algorithms and the full proofs are in Section \ref{sec:piif_algs}.

\paragraph{Overview: Fair Propose-and-Reject -- Doctors-First (Algorithm \ref{alg:GS_WA}).}
At the initialization of the algorithm, each doctor has a free probability mass of 1. 
At each round, every doctor proposes its free probability mass to the most preferred hospital that did not reject it yet. 
For every hospital $h$, after the doctors' proposals, the  probability mass proposed by each doctor is composed of the probability mass the doctor proposed in the current round, plus the probability to be matched to this doctor in the previous round.
Given this proposed probability mass, the hospital chooses its allocation for the current round in a way that will not lead to unfairness.

In particular, each hospital $h$ uses the \emph{rising tide algorithm} (Algorithm \ref{alg:water_algorithm}) for choosing its allocation in each round.
At the initialization, hospital $h$ has unallocated probability mass 1.
The algorithm goes from the most preferred cluster to the least preferred cluster by hospital $h$.
When there is no unallocated probability mass or no proposed probability mass, the algorithm terminates.
For each cluster $C$, while there is proposed probability mass from cluster $C$:
Each doctor $d$ in cluster $C$ is allocated with the minimum between: (1) the proposed probability mass from doctor $d$; (2) the unallocated probability mass divided by the number of doctors in the cluster with non-zero proposed probability mass. Every allocated probability mass is removed from the proposed probability mass and from the unallocated probability mass.

After running the algorithm described above, in the allocation of the current round, hospital $h$ is allocated to the most preferred clusters possible given the proposed probability mass. For every two doctors in the same cluster, one can be allocated with less probability mass than the other only when having less proposed probability mass. 

Any probability mass not used for the allocation is rejected and becomes free again for the next round.
The algorithm terminates when there is no more free probability mass, i.e., there is a full allocation, or the free probability mass is very small.

This algorithm is executed \emph{locally}, in the sense that for each hospital, the only necessary knowledge is its own preferences, its previous allocation, and the proposals it received. For each doctor, the only necessary knowledge is its own preferences, its free probability mass, and the hospitals' response. No party needs to be aware of the status of the other parties.

\paragraph{Overview: Fair Propose-and-Reject -- Hospitals-First (\Cref{alg:GS_PSP}).}
This algorithm is similar to \Cref{alg:GS_WA}, except that the hospitals propose their free probability mass to the clusters. At each round, every hospital proposes to the most preferred cluster that did not reject it.
For each cluster, the proposed probability mass from each hospital is composed of the probability mass the hospital proposed in the current round, plus the probability to be matched to this hospital in the previous round.

Given this proposed probability mass, each cluster chooses an envy free allocation for the current round. To achieve envy freeness, we use the \emph{probabilistic serial procedure} due to Bogomolnaia and Moulin \cite{bogomolnaia2001new}, see \Cref{alg:PSP}.
Any probability mass not used for the allocation is rejected and becomes free again for the next round. 

The algorithm terminates when there is no free probability mass, i.e., there is a full allocation, or the free probability mass is very small.

This algorithm is also executed locally: no party has to consider anything but its own interests (though there is some coordination in each round between the doctors in the same cluster, in allocating the mass proposed to that cluster). 

\subsection{Barriers for Unfair Preferences and General Metrics}
\label{sec:intro_neg_results}

In Section \ref{sec:impossibility}, we show that if we relax the strong requirement on the hospitals' preferences by allowing either unfair preferences or a general metric, no algorithm that is ``similar in spirit'' to the algorithms outlined above (and to the classical Gale-Shapley algorithm) can guarantee both fairness and stability.
Towards this, we formalize a class of \emph{local-proposing algorithms}. Here, we focus on the case where the doctors propose, although we also describe a class of hospital-proposing algorithms. The doctor-proposing class includes all algorithms consisting of sequential rounds of proposals. In each round, each doctor proposes its unallocated probability mass to a hospital. Then, each hospital has to use the probability mass proposed to it to choose an allocation. 
We assume that doctors choose the hospital according to their preferences. We also assume that once a doctor proposes some probability mass to a hospital, if the hospital rejects it, the doctor will never propose it to this hospital again. This class, and the corresponding hospital-proposing class, generalizes the algorithms outlined above (as well as the classical Gale-Shapley algorithms).

To provide this negative result, we need definitions of stability that extend beyond the setting of strictly IF hospital preferences. Remaining in the proto-metric setting, we define stability under unfair hospital preferences in Section \ref{subsec:intro:stability_unfair_prefs}. The negative results for local algorithms under unfair preferences are in Section \ref{subsec:intro:barrier_unfair_prefs}. Moving to general metrics, we formalize a a minimal stability requirement in Section \ref{subsec:intro:weak_stability} (a weaker stability requirement makes our negative results stronger). We discuss our negative result for fair hospital preferences under general metrics in Section \ref{subsec:intro:barrier_general_metric}.


\subsubsection{Stability under Unfair Preferences, Proto-Metric}
\label{subsec:intro:stability_unfair_prefs}

If the hospitals can have unfair preferences, then fairness and stability might be trivially incompatible. For example, a prestigious hospital $h$ can express a discriminatory preferences towards members of group $T$. Suppose $h$ is the most-preferred hospital of all doctors: a stable allocation must always match $h$ to a member of $T$, but this is blatantly unfair!

We want to provide a utility guarantee to the (unfair) hospitals, but the example above demonstrates that if we allow unfair hospitals to make unfair deviations from their allocation, then fairness and stability might be trivially incompatible. Since we insist on fairness for the doctors, we find that it is natural to relax stability by requiring that there are no {\em fair} deviations that the hospitals would prefer. Formally, we modify the classical notion of a blocking pair, by only allowing a hospital to form pairs with doctors that are {\em outside} the cluster of the doctor to whom it is matched. This restricts the alternative ``offers'' that a hospital with unfair preferences can make, and ensures that they do not violate the fairness constraints. 

\begin{definition}[Stability under proto-metrics (informal)]
\label{def:intro-proto-stability}
A probabilistic matching is unstable if there exists, with non-zero probability over the matching, a pair of a doctor $d$ and a hospital $h$ that prefer each other to their respective partners, as long as the partner of the hospital $h$ is at distance 1 from the doctor $d$.
The preferences of the hospital can be probabilistic. Thus, the comparison between the hospital's allocations is in terms of stochastic domination.
\end{definition}

See Section \ref{sec:intro:fairness}, for a definition of stochastic domination.

See Definition \ref{def:contract_stability_proto} for formal definition of stability. We note that if the hospitals' preferences are strictly IF, no hospital prefers one doctor over the other if they are in the same cluster. Thus, Definition \ref{def:intro-proto-stability} is equivalent to the definition described in Section \ref{subsec:intro:stability}.

\subsubsection{Barriers for Unfair Preferences, Proto-Metric}
\label{subsec:intro:barrier_unfair_prefs}

We show that no local-proposing algorithm can find a fair and stable matching when the hospital preferences can be unfair (even in the proto-metric setting).

\begin{theorem}[Failure of local algorithms for unfair hospital preferences (informal)]
There does not exist a local algorithm that, when the metric is a proto metric but the hospital preferences might be unfair, always finds an allocation that is both PIIF and stable (see Definition \ref{def:intro-proto-stability}).
\end{theorem}

\begin{proof}[Proof sketch.]
Suppose there are four doctors $i, j, k, l$, where the first and last pairs are at distance 0 (i.e., $(i,j)$ and $(k,l)$), and the rest are at distance 1.
Suppose there exists a hospital $A$, whose preferences are $j \succ_A k \succ_A i \succ_A l$. If hospital $A$ is the most preferred hospital by all the doctors, then the fair allocation that is most preferred by hospital $A$ is to be assigned to doctors $i$ and $j$ uniformly (the other fair possibilities are to be assigned to doctors $k$ and $l$ uniformly, or to some convex combination of these two allocations). Thus, since it is also doctors $i$ and $j$'s most preferred fair allocation (they both rank $A$ first), it is the only stable one (for any other fair allocation, hospital $A$ and doctor $j$ form a blocking pair).
However, if hospital $A$ is the most preferred hospital by doctors $k$ and $i$, but doctors $j$ and $l$ are matched to hospitals they prefer over hospital $A$. Then, the fair allocation that is most preferred by hospital $A$ is to be assigned to doctor $k$ with probability 1 (assuming hospital $A$ cannot be assigned to doctor $j$). Similarly, since this is also doctor $k$'s most preferred allocation (doctor $k$ rank hospital $A$ first), it is the only stable one, otherwise hospital $A$ and doctor $k$ form a blocking pair.

If only the doctors $i$ and $k$ propose to hospital $A$ probability mass 1 in the first round, there is no allocation that hospital $A$ can choose over this probability mass that will always lead to a fair and stable matching, i.e., hospital $A$ does not know whom to accept and whom to reject. 
If doctors $j$ and $l$ will propose to hospital $A$ in a later round, to have a stable and fair solution, hospital $A$ must accept at least probability mass 1/2 from doctor $i$. However, if no other doctor will propose to hospital $A$ in a later round, to have a stable and fair solution, hospital $A$ must accept probability mass 1 from $k$. Since hospital $A$ cannot distinguish these two cases in the first round, the algorithm will return an unfair or unstable output in at least one of them.
\end{proof}

Unfair preferences allow us to create a situation where hospital $A$ actually wants to accept probability mass from doctor $j$, but under the fairness requirement, in order to accept probability mass from doctor $j$, hospital $A$ must accept some probability mass of doctor $i$ (since their distance on the metric is 0). However, if doctor $j$ will never propose to hospital $A$ then hospital $A$ does not want to accept any probability mass from doctor $i$.

Now,
suppose the preferences were fair; if hospital $A$ would want to be matched to doctor $j$, it would equally want to be matched to doctor $i$. Thus, hospital $A$ would be able to decide whether to accept or reject doctor $i$'s probability mass independently of whether doctor $j$ will propose to it in a later round or not. Thus, we must allow unfair hospital preferences to achieve the example above, under a proto-metric.

\subsubsection{Stability for General Metrics}
\label{subsec:intro:stability_gen_metrics}

Intuitively, a probabilistic allocation is stable if no hospital can offer to some doctors to be matched to it in a way that will improve both the hospital's allocation and those of the doctors. 
If the hospitals' preferences are not fair, running the Gale-Shapley algorithm fails to output a fair solution for obvious reasons (see Section \ref{subsec:intro:barrier_unfair_prefs}). 
However, even if the hospitals' preferences are fair, finding a non-trivial fair allocation can be quite challenging. This happens because even a mild probabilistic preference for one doctor over another can lead to a situation where a hospital prefers to {\em always} be matched to one of the doctors and not the other.

Consider the example of two doctors $i$ and $j$ at distance 1/3, and two hospitals $A$ and $B$. Suppose both doctors prefer hospital $A$ over hospital $B$, and hospital $A$'s preferences are 
\begin{equation*}
    \begin{cases}
        i \succ_A j,& \textit{w.p. }2/3 \\
        j \succ_A i,& \textit{w.p. }1/3
    \end{cases}.
\end{equation*}
Hospital $A$'s preferences are fair. However, in any PIIF allocation, hospital $A$ prefers to be always matched to doctor $i$, over its outcome in the allocation. This is because if hospital $A$ is always matched to doctor $i$, it is matched to its first preference with probability 2/3.
However, in any PIIF allocation, hospital $A$ is matched to doctor $j$ with probability at least 1/3, which implies that it is matched to its first preference with probability no more than $\frac{2}{3}\cdot\frac{2}{3} + \frac{1}{3}\cdot\frac{1}{3} = \frac{5}{9} < \frac{2}{3}$.

In the above example, $A$ and $i$ were a blocking pair, but the issue was that this alternative is {\em unfair} (to $j$).
Motivated by this example, and to avoid trivial incompatibilities between fairness and stability, we would like to force the hospitals to make only {\em fair offers}.
However, formalizing such a stability definition that, on the one hand, is strong enough to have a meaningful guarantee for the hospitals, and on the other hand, is compatible with our fairness requirement (or at least is not trivially incompatible), presents several subtleties. We elaborate on this in Section \ref{sec:stability}.
Instead, we present a minimal (weak) stability definition for general metrics that we use in our negative results (using a weak definition makes the negative results stronger).

\paragraph{Weak Stability.}
\label{subsec:intro:weak_stability}
The weak stability definition is guided by simple scenarios. 
Suppose that the hospitals are $A, B$ and $C$ and that the doctors are $i_1, i_2$ and $j$, where doctors $i_1$ and $i_2$ are similar (at distance 0) and doctor $j$ is far from them (at distance 1). 
Consider the following two cases:
\paragraph{Case 1.} Hospital $A$ is the most preferred hospital by all the doctors, and hospital $A$ prefers doctors $i_1$ and $i_2$ over doctor $j$. It is natural to require that hospital $A$ should be matched to the uniform distribution over $i_1$ and $i_2$.
Intuitively, this is the ``best'' allocation for hospital $A$ and doctors $i_1$ and $i_2$, subject to fairness. In this case, we maintained fairness by saying that if hospital $A$ is not matched to $i_1$ and $i_2$ with probability 1, it is allowed to prefer the alternative allocation of being matched to the uniform distribution over $i_1$ and $i_2$ since it is IF.
    
\paragraph{Case 2.} 
Hospital $A$ is the most preferred hospital by doctors $i_1$ and $j$, but not by doctor $i_2$, and hospital $A$ still prefers doctors $i_1$ and $i_2$ over doctor $j$. Suppose that doctor $i_2$ prefers hospital $B$ over hospital $A$, that doctor $i_2$ is matched to hospital $B$ with probability 1, and that the allocation of hospitals $A$ and $C$ has not been determined yet. This time, 
it is natural to require that hospital $A$ should be matched to doctor $i_1$ with probability 1. However, if hospital $A$ and doctor $i_1$ are not matched with probability 1, allowing hospital $A$ to prefer this alternative allocation implies that we allow hospital $A$ to prefer an alternative allocation which does not satisfy IF. We choose to allow this since doctor $i_2$ prefers its own allocation over being matched to hospital $A$.

The following definition captures the above intuitions:

\begin{definition}[Weak stability (informal)]
\label{def:intro:weak_stability}
An allocation is (strongly) unstable if there exists a hospital $h$ and an alternative allocation $\nu$ over the doctors, such that: (1) The hospital $h$ prefers the alternative allocation $\nu$ over its own allocation.
(2) Every doctor in the support of the alternative allocation $\nu$ prefers hospital $h$ over the hospitals in its support.
(3) For every doctor that is not in the support of the alternative allocation $\nu$, either (i) the doctor is at distance 1 from any doctor in the support of the alternative allocation $\nu$ or (ii) the doctor prefers every hospital in its own support over the hospital $h$.
(4) For every two doctors in the support of the alternative allocation $\nu$, the allocation $\nu$ satisfies IF.

If there is no such hospital and alternative allocation, then we say that the allocation is weakly stable.
\end{definition}
See Definition \ref{def:selective_stability} for formal definition of weak stability.

We note that Definition \ref{def:intro-proto-stability}, which is relevant in the proto-metric setting, is a stronger stability. In particular, it implies Definition \ref{def:intro:weak_stability} (when the metric is a proto-metric). In the full proofs in section \ref{sec:impossibility}, we show the barriers for unfair preferences for Definition \ref{def:intro:weak_stability}, instead of Definition \ref{def:intro-proto-stability}.

\subsubsection{Barriers for General Metrics}
\label{subsec:intro:barrier_general_metric}
In the setting of a general metric and IF hospital preferences, we show that we can make doctors $i$ and $j$ far enough that even under the IF requirement, hospital $A$'s preferences would actually be as described above, i.e., $j \succ_A k \succ_A i \succ_A l$. However, we can make doctors $i$ and $j$ close enough that an allocation where hospital $A$ is always matched to doctor $j$, and never to doctor $i$, would be considered unfair. Then we can arrange the proposals as we did in the case of unfair preferences such that any decision that hospital $A$ makes in the first round can lead to an unfair allocation. See \Cref{sec:failure_fair_prefs} for details.

\subsection{Open Questions}
The new frontier of fairness in two-sided markets raises many fundamental questions for further study.
In this work, we present algorithms for finding fair and stable allocations under some restrictions. We show that generalizing this result presents several difficulties. A natural question for further work is either extending the negative results beyond the class of local-proposing algorithms, or finding an algorithm for a more general setting, such as general metrics or unfair preferences.

A possible direction for generalizing the results for general metrics is to have a stronger requirement over the hospital preferences. In the negative results, we use IF preferences for a general metric, and show that effectively they behave similarly to unfair preferences. This indicates that, in the case of a general metric, individual fairness might not be a strong enough fairness requirement over the hospital preferences. We elaborate on this in Section \ref{sec:conclussion}.

We present two algorithms in Section \ref{sec:piif_algs}, which have almost the same guarantees concerning fairness and stability. It is known that different variants of the Gale-Shapley algorithms have different guarantees for optimality and incentive compatibility. In \Cref{sec:optimality}, we show that doctor-proposing variant (\Cref{alg:GS_WA}) is not optimal for the doctors.
For the other variant (\Cref{alg:GS_PSP}), we leave this as an open question. We also leave the question of incentive compatibility for future work.

\paragraph{Organization.} Sections \ref{sec:defs} and \ref{sec:stability} contain the definitions of fairness and stability. 
In Section \ref{sec:failure_in_composition} we describe the negative results for running the Gale-Shapley algorithm over fair preferences. 
Section \ref{sec:piif_algs} contains the algorithms for fair and stable allocations in the setting of proto-metric and fair hospital preferences.
Section \ref{sec:impossibility} contains the impossibility results for local algorithms in general settings.
\section{Definitions and Problem Formulation}
\label{sec:defs}

Let $\D$ be a collection of $n$ doctors and $\H$ be a collection of $n$ hospitals.
Let $\Pi_{\D}, \Pi_H$ denote all permutations over $\D$ and $\H$, respectively. 
A \emph{matching} $m$ is a mapping $\D$ to $\H$, i.e., each doctor is mapped to a single hospital. 
A \emph{probabilistic allocation} $\pi$ is a mapping from $\D$ to $\Delta(\H)$, where for $i \in \D$, $\pi(i)$ represents the \emph{prospect} of doctor $i$ (the probability distribution over hospitals that they receive). Similarly, for a hospital $h \in \H$, $\pi(h)$ represents the \emph{prospect} of doctor $h$.
We focus on the setting where we assign each doctor to a single hospital, so the probabilistic allocation is a distribution over matchings.

We assume we are given a \emph{similarity metric} for the doctors $d : \D \times \D \rightarrow [0,1]$. 

\begin{definition}[Pseudometric]
A function $d : \D \times \D \to [0,1]$ is a pseudometric over the set $\D$, if $d$ is non-negative, symmetric and satisfies the triangle inequality. For every $i \in \D$, $d(i, i) = 0$ but we allow also for $j \in \D \backslash \{i\}$, $d(i,j) = 0$.
\end{definition}

We focus on the case of \emph{proto-metrics}, where the distances are either 0 or 1:

\begin{definition}[Proto-Metric]
A pseudometric $d : \D \times \D \to [0,1]$ is a proto-metric if for every $i,j \in \D$, $d(i,j)$ is either 0 or 1. From the triangle inequality,
there exists a partition of $\D$ to clusters such that $d$ assigns a distance 0 for doctors in the same cluster, and 1 for doctors in different clusters.
\end{definition}

\subsection{Preferences} We assume each doctor $i \in \D$ has an ordinal preference over hospitals: $r_{i} \in \Pi_{\H}$, where, e.g., $r_{i}(1) \in \H$ is $i'$s favourite hospital and $r_{i}(n) \in \H$ is $i'$s least-favourite hospital. For simplicity, for a hospital $h$, we denote with $r_{i}^{-1}(h) \in [n]$ the rank of this hospital in $i$'s preferences. For $h_1, h_2 \in \H$, we use the notation $h_1 \succ_i h_2$ to indicate that $r_i^{-1}(h_1) < r_i^{-1}(h_2)$, i.e., doctor $i$ prefers hospital $h_1$ over hospital $h_2$.

We assume each hospital $h \in \H$ has a probabilistic ordinal preference over doctors: $r_{h} \in \Delta(\Pi_D)$, where, $r_{h}(1)$ is a random variable for $h'$s favourite doctor and $r_{h}(n)$ is a random variable for $h$'s least-favourite doctor. For simplicity, for a doctor $i$, we denote with $r_{h}^{-1}(i)$ the random variable containing the rank of $i$ in $h$'s preferences.

Given a preference function $r$, \emph{stochastic domination} provides a natural way to compare two prospects.

\begin{definition}[Stochastic Domination]
\label{def:stochastic_domination}
Let $r$ be a deterministic or probabilistic preference function over a set of individuals (doctors or hospitals) and let $p$ and $q$ be prospects over the same set of individuals.
We say that $p$ stochastically dominates $q$, $p \succeq_r q$, if the following holds:
\end{definition}

\begin{equation}
    \forall k \in [n]: \quad \Pr_{o \sim p,r}\sbr{r^{-1}(o) \leq k} \geq \Pr_{o \sim q,r}\sbr{r^{-1}(o) \leq k} 
\end{equation}

in other words,  for every $k$, the probability for getting a top-$k$ outcome (according to $r$) under $p$ should be at least at large as the probability for getting a top-$k$ outcome under $q$. When there exists a $k' \in [n]$ for which the above holds with a \emph{strict} inequality we say that $p$ strongly-dominates $q$. We use $\succeq_r$ and $\succ_r$ to denote the partial orders that correspond to weak domination and strong domination, respectively. (Note that this is a partial ordering over prospects, since the two outcomes $p,q$ may be incomparable  -- neither one dominates the other).

\subsection{Fairness}

Our goal is achieving fairness with respect to the doctors.
One way to define this notion of fairness is \emph{individual fairness} (IF) by Dwork \emph{et al.} \cite{dwork2012fairness}, where similar individuals should be assigned similar outcomes (e.g., because they have similar merit):

\begin{definition} [Individual Fairness \cite{dwork2012fairness}]
\label{def:IF}
An allocation $\pi$ is individually-fair (IF) with respect to a divergence $D$, and a similarity metric $d$, if for all pairs of individuals $i, j \in {\D}$, the Lipschitz condition $D(\pi(i), \pi(j)) \le d(i, j)$ is satisfied.
\end{definition}

One natural choice for the divergence function $D$ is the \emph{total variation distance}:
\begin{definition}[Total Variation Distance]
The total variation distance between two prospects $p, q \in \Delta(\H)$ is
\begin{equation}
    D_{TV}(p,q) = \frac{1}{2}\sum_{h \in \H}|\Pr_{o \sim p}[o = h] - \Pr_{o \sim q}[o = h]|.
\end{equation}
\end{definition}

Definition \ref{def:IF} does not consider the preferences of the doctors. 
Intuitively, $p \succeq_r q$ appeals to a strong way in which the prospect $p$ is at least as good as $q$.
Another way to define fairness is \emph{envy freeness}, where we disregard merit and require that every two individuals should not envy each other's outcomes.

\begin{definition}[Envy Freeness] An allocation $\pi$ is envy-free (EF) with respect to individual preferences $\{\succeq_i\}$ if for all individuals $i$, for all other individuals $j$, $\pi(i) \succeq_i \pi(j)$.
\end{definition}

We can consider both merit and individual preferences by using \emph{preference-informed individual fairness} (PIIF) due to Kim \emph{et al.} \cite{kim2020preference}.

\begin{definition}[Preference-Informed Individual Fairness \cite{kim2020preference}]
\label{PIIF}
Fix doctors with preferences $r_1, \dots, r_n$.
An allocation $\pi$ is preference-informed individually fair with respect to a divergence $D$ and a similarity metric $d$, if and only if for every two doctors $i,j$, there exists an alternative allocation $p^{i;j}$ such that
\begin{equation}
\label{eq:alternative_close}
    D(p^{i;j}, \pi(j)) \le d(i,j)
\end{equation}
\begin{equation}
\label{eq:prefer_over_alternative}
    \pi(i) \succeq_i p^{i;j}.
\end{equation}
\end{definition}

Definition \ref{PIIF} allows different allocations to similar individuals when the differences are aligned with the individuals' preferences.

We say that $\pi$ is PIIF with respect to a pair of doctors $i$ and $j$, if such an alternative allocation $p^{i;j}$ that satisfies \cref{eq:alternative_close} and \cref{eq:prefer_over_alternative} exists.

\begin{corollary}[PIIF under proto-metrics]
Under proto-metrics we get the following definition:
Fix doctors with preferences $r_1, \dots, r_n$.
An allocation $\pi$ is preference-informed individually fair if for every two doctors $i,j$, if $d(i,j) = 0$, then either: (i) $\pi(i) = \pi(j)$ or (ii) $\pi(i) \succ_i \pi(j)$.
\end{corollary}

In other words, $\pi$ is \emph{unfair} if there exist two doctors $i,j$ in the same cluster and $k \in [n]$ such that the probability of $i$ receiving a top-$k$ outcome is larger under prospect $\pi(j)$ than under prospect $\pi(i)$.

We also introduce a relaxation \emph{$\tau$-preference-informed individual fairness}. Where for every two individuals, we allow the matching to be unfair up to a small constant $\tau$.

\begin{definition}[$\tau$-Preference-Informed Individual Fairness]
An allocation $\pi$ is $\tau$-PIIF with respect to a similarity metric $d$, if it is PIIF with respect to the similarity metric $d^\tau$, where $d^\tau$ is defined as follows
\begin{equation*}
\forall i, j, \in \D: d^\tau(i,j) = \min\set{d(i,j) + \tau, 1}.
\end{equation*}
\end{definition}

For the case of a proto-metric, $\tau$-PIIF means that for every two doctors $i_1, i_2$ from the same cluster, i.e., $d(i_1,i_2) = 0$, 
we allow that $\pi(i_1) \nsucceq_{i_1} \pi(i_2)$, if we can change the prospect $\pi(i_2)$ by no more than $\tau$ $\left(D(p^{i_1;i_2}, \pi(i_2)) \le \tau\right)$ to create an alternative allocation $p^{i_1;i_2}$ that satisfies 
\begin{equation*}
    \pi(i_1) \succeq_{i_1} p^{i_1;i_2}.
\end{equation*}


\subsection{Fair Preferences}

For much of this work, we study the setting where the preferences of the hospitals are fair. 
To do so, we allow the hospitals to have probabilistic preferences.
In order to determine what are fair preferences, we use the notion of \emph{individual fairness} by Dwork \emph{et al.}  \cite{dwork2012fairness}, see Definition \ref{def:IF}.

It is left to define the divergence $D$. On one hand, we would like to define the divergence $D$ to be strong enough to guarantee fairness when running ``natural'' algorithms. On the other hand, we would like to allow the hospitals to express rich preferences to the greatest extent possible, as long as the final assignment is fair.

In this section, we provide two definitions for fair preferences, where the second definition is a relaxation of the first one. 
The first definition is restricted to proto-metrics, and we use it in our positive results, in \Cref{sec:piif_algs}. The second definition generalizes to any pseudometric. Thus, we use it in our negative results, in \Cref{sec:impossibility}, wheres we discuss general metrics.

For the case of proto-metrics, we introduce \emph{strict individual fairness}.
\begin{definition}[Strict Individually Fair Preferences]
\label{def:strict_IF}
A set of preferences is strictly individually fair with respect to a proto-metric $d : \D \times \D \to \set{0,1}$ if for every hospital $h \in \H$:
\begin{itemize}
    \item For every cluster $C \subseteq \D$ and two doctors $i_1, i_2 \in C$, i.e., such that $d(i_1, i_2) = 0$:
    \begin{equation*}
        \forall r \in [n]:\Pr[r_h(r) = i_1] = \Pr[r_h(r) = i_2].
    \end{equation*}
    \item For every two clusters $C_1, C_2 \subseteq \D$, either $C_1 \succ_h C_2$ or $C_2 \succ_h C_1$. Where $C_1 \succ_h C_2$ if:
    \begin{equation*}
        \forall i \in C_1, j \in C_2: \Pr[r_h^{-1}(i) < r_h^{-1}(j)] = 1.
    \end{equation*}
\end{itemize}
\end{definition}

When discussing a proto-metric and strict individually fair preferences we sometimes use the notation of naming the clusters $i, j, ...$ and the doctors in cluster $i$, $i_1, i_2, ...$ if $|i| > 1$ or just $i$ otherwise. If hospital $h$ prefers cluster $i$ over $j$ we denote this by $i \succ_h j$ which means that $h$'s preferences are the uniform distribution over all the permutations over of doctors in $i$ and in $j$, e.g. for $i = \set{i_1, i_2}$ and $j = \set{j}$, if hospital $h$ prefers doctor $i$ over doctor $j$, instead of writing that hospital $h$'s preferences are 
$$\begin{cases}
i_1 \succ_h i_2 \succ_h j,& \textit{w.p. } 1/2\\
i_2 \succ_h i_1 \succ_h j,& \textit{w.p. } 1/2
\end{cases}
$$
we just write $i \succ_h j$.

Definition \ref{def:strict_IF} is places a strong functional requirement over the hospital  preferences. 
It is not necessarily clear why doctors in the same cluster must be consecutive in the ordered preferences.
Consider the case of ${\D} = \set{i_1, i_2, j, k}, \H = \set{A, B, C, D}$ where doctors $i_1, i_2$ belong to one cluster and $j$ and $k$ belong to other clusters. The following preferences would not be considered strictly individually fair, even though doctors $i_1$ and $i_2$ are treated identically.
$$r_h^1 = \begin{cases}
i_1 \succ_h j \succ_h i_2 \succ_h k & \text{w.p. } 1/2 \\
i_2 \succ_h j \succ_h i_1 \succ_h k & \text{w.p. } 1/2
\end{cases}.$$

There are several ways to relax strict IF if we want to allow such preferences, one of them is \emph{mutual replacement individual fairness}. 
This definition guarantees that for every two doctors $i,j$ and hospital $h$, the statistical distance between hospital $h$'s preference $r_h$ and hospital $h$'s preference after switching doctors $i$ and $j$, is bounded by the distances between doctors $i$ and $j$ in the metric.

This means that for the set $\D = \set{i_1,i_2, j, k}$, the preference $r_h^1$ is mutual replacement individually fair but the preference $r_h^2$ is not:
$$r_h^2 = \begin{cases}
i_1 \succ_h k \succ_h i_2 \succ_h j & \text{w.p. } 1/2 \\
i_2 \succ_h j \succ_h i_1 \succ_h k & \text{w.p. } 1/2.
\end{cases}$$

That is, because the distance between doctors $i_1$ and $i_2$ is 0, but swapping them in the preference $r_h^2$ creates a different distribution. 

\begin{definition}[Mutual Replacement Individual Fairness Preferences]
\label{def:mutual_replacement_individual_fairness}
Let ${\D}$ be a set of doctors, $h \in \H$ a hospital, $d : {\D} \times {\D} \to [0,1]$ a metric and let $r_h$ be $h$'s (probabilistic) preference.
For any two doctors $i, j \in \D$, let us denote by $r_h^{i;j}$ the probabilistic preference where every appearance of doctor $i$ is replaced by doctor $j$ and every appearance of doctor $j$ is replaced by doctor $i$.
The preference $r_h$ is considered mutual replacement individually fair, if for every two doctors $i, j \in {\D}$,
\begin{equation*}
    D_{TV}(r_h, r_h^{i;j}) \le d(i,j)
\end{equation*}
\end{definition}

\section{Stability}
\label{sec:stability}

Finding a PIIF allocation is trivial, since the uniform distribution over all allocations is PIIF. 
However, we want the hospitals to have an incentive to participate in the mechanism. 
In the classical version of the matching problem, without metric induced fairness constraints and with deterministic preferences, this is captured by \emph{stability}.

\begin{definition}[Stable Matching]
\label{def:classic_stability}
Given a deterministic allocation $\pi$, deterministic preferences $\{r_i\}_{i\in {\D}}$, $\{r_h\}_{h\in \H}$, a pair $(i, h) \in {\D} \times \H$ is a blocking pair if $r_i(h) < r_i(\pi(i))$ and $r_h(i) < r_h(\pi(h))$.

The allocation $\pi$ is considered a stable matching if there are no blocking pairs.
\end{definition}

If doctor $i$ and hospital $h$ are a blocking pair, they prefer being matched to each other over participating in the matching. This means that they do not have an incentive to participate in the matching.

\begin{definition}[Stable Matching Algorithm]
A stable matching algorithm is an algorithm that outputs a stable matching given a set of preferences.
\end{definition}

For the case of deterministic preferences, a stable matching always exists and can be found in polynomial time using the \emph{Gale-Shapley} algorithm due to Gale and Shapley \cite{gale1962college} (described in Section \ref{sec:gs_alg}).

We generalize this notion to probabilistic allocations. 
Our goal is to have a stability definition that preserves the hospitals' incentive to participate in the matching while being compatible with fairness.

There is a significant difference between the case where the preferences of the hospitals are fair and the case where the hospitals have arbitrary preferences. In the latter case, for any PIIF allocation, there might be alternative allocations that some hospitals and doctors want. However, we do not want to consider those alternative allocations because they are unfair. We defer the discussion of stability under unfair preferences to \cref{sec:unfair_prefs}, where we show that unfair hospital preferences present serious barriers in finding a fair and stable solution.

In the rest of this section, we assume the hospitals have fair probabilistic preferences, and for this setting, we want to define stability. 

We focus on an ex-ante (a priori) stability definition, i.e., a guarantee for the distribution of probabilistic allocations.
In ex-ante stability, we think of the preferences and the allocation as independent random variables.
We discuss ex-post (a posteriori) stability, a guarantee for every sampled allocation, in Section \ref{sec:ex_post_stability}.

\subsection{Ex-Ante Stability}
We start by introducing a strong ex-post stability requirement, \emph{contract stability}, for the setting of a proto-metric in \Cref{sec:contract_stability}. We used this definition in our positive results in \Cref{sec:piif_algs}. In  \Cref{sec:weak_ex_post_stability}, we introduce a weak ex-post stability definition for general metrics, \emph{weak ex-ante stability}, which we use for our negative results in \Cref{sec:impossibility}. In \Cref{sec:generalizing_stability_fails}, we present the subtleties and difficulties in generalizing the strong ex-post definition presented in \Cref{sec:contract_stability} to general metrics.

\subsubsection{Contract Stability}
\label{sec:contract_stability}
Our primary notion of ex-ante stability is \emph{contract stability}, where an allocation is stable if there is no active contract. A contract is an agreement between some doctors and a hospital that are not satisfied with the allocation and can act together to improve it. 
This kind of agreement, however, can lead to unfairness. On the one hand, our challenge is to allow contracts that indicate that hospitals are unhappy with the allocation in a ``fair'' way and, on the other hand, not to allow contracts caused by ``unfair'' incentives.
We emphasize that our main motivation is understanding when and how fairness and stability are compatible. We want to ``rule out'' unfair contracts to the extent that they lead to incompatibility with fairness. Achieving a strong notion of stability, which also allows the existence of unfair contracts, is only positive, as long as it does not lead to incompatibility with fairness.

Since we focus, in this work, on proto-metrics, we present a definition for contract stability for the proto-metric setting.
This definition is a generalization of the concept of blocking pairs in the classical stability definition. However, for a matching $m$, instead of considering only the blocking pair $(i, h) \in \D \times \H$ we also consider the doctor matched to $h$, $m(h)$. We allow doctor $i$ and hospital $h$ to be a blocking pair only if doctor $i$ and doctor $m(h)$ do not belong to the same cluster. 

\begin{definition}[Contract]
\label{def:contract}
Given deterministic doctor preferences $\mathcal{P}_D = \{r_i\}_{i\in \D}$, probabilistic hospitals preferences $\mathcal{P}_H = \{r_h\}_{h\in \H}$ and a probabilistic allocation $\pi$, a tuple $\mu = (h, i; h', i') \in \H \times \D \times \H \times \D$ is a contract if $d(i,i') = 1$, $\pi^\mu(h) \succ_h \pi(h)$ and $r_i(h) < r_i(h')$, where:
\begin{align}
\forall j \in \D \backslash \set{i, i'}: &\pi^\mu(j) = \pi(j)\\
&\pi^\mu(i) =
\begin{cases}
h, & \pi(i) = h' \wedge \pi(i') = h \\
\pi(i), & \text{otherwise}
\end{cases}\\
&\pi^\mu(i') =
\begin{cases}
h', & \pi(i) = h' \wedge \pi(i') = h \\
\pi(i'), & \text{otherwise}
\end{cases}
\end{align}
\end{definition}

\begin{definition}[Active Contract]
\label{def:active_contract}
Given deterministic doctor preferences $\mathcal{P}_D = \{r_i\}_{i\in \D}$, probabilistic hospitals preferences $\mathcal{P}_H = \{r_h\}_{h\in \H}$ and a probabilistic allocation $\pi$, a contract $\mu = (h, i; h', i') \in \H \times \D \times \H \times \D$ is an active contract if 
$$
\Pr[\pi(h) = i' \wedge \pi(i) = h'] > 0.
$$
\end{definition}

\begin{definition}[Contract Stability]
\label{def:contract_stability_proto}
The allocation $\pi$ is contract stable if there are no active contracts.
\end{definition}

We note that for fair hospital preferences, the requirement that $d(i,i') = 1$ is redundant since it is implied from the requirement that $\pi^\mu(h) \succ_h \pi(h)$.
Thus, we allow the hospital to prefer only alternatives that will not contradict fairness if the preferences are fair. For fair preferences, a contract will exist for any doctor in $i$'s cluster that prefers $h$ over $h'$.

Since sometimes it is not possible to achieve contract stability, we also define the following relaxation for proto-metrics, where the probability that there exists an active contract is not 0 but bounded by $\tau$:

\begin{definition}[$\tau$-Contract Stability]
Denote by $S_\pi \subseteq \D \times \H \times \D \times \H$ the set of all active contracts in an allocation $\pi$.
An allocation $\pi$ is $\tau$-singleton contract stable if
$$\Pr[\bigvee_{(h, i ;h', i') \in S_\pi}(\pi(h) =i'\wedge \pi(h') = i)] \le \tau.$$
\end{definition}

In Section \ref{sec:piif_algs}, we show that $\tau$-contract stability is compatible with fairness when the hospital preferences are fair.

\subsubsection{Weak Ex-Ante Stability}
\label{sec:weak_ex_post_stability}
In \Cref{sec:contract_stability} we presented a stability definition for the setting of a proto-metric.
Generalizing Definition \ref{def:contract_stability_proto} to general metrics presents several subtleties, we defer this discussion to Section \ref{sec:generalizing_stability_fails}. Here, we focus on formulating a weaker ex-ante stability requirement for general metrics. Looking ahead, in Section \ref{sec:impossibility}, we show barriers for finding and proving the existence of fair and stable allocations, even for this minimal requirement. 

Our weak stability definition is guided by bad scenarios that we want to avoid.
Consider the example where $\H = \set{A, B, C}$ and $\D = \set{i_1, i_2, j}$ where $d(i_1, i_2) = 0$ and $d(i_1, j) = d(i_2,j) = 1$. Hospital $A$ is the most preferred hospital by all the doctors and hospital $A$'s preferences are $i \succ j$. In this simple case, we would expect that in any fair and stable allocation $\pi$, the prospect of hospital $A$ will be:
\begin{align*}
    \pi(A) = 
    \begin{cases}
        i_1,&\text{ w.p. }1/2\\
        i_2,&\text{ w.p. }1/2
    \end{cases}.
\end{align*}
Intuitively, this is the ``best'' prospect for hospital $A$, subject to fairness: there is no other prospect hospital $A$ prefers that is PIIF with respect to doctors $i_1$ and $i_2$.
Thus, we say that an allocation $\pi$ is not even \emph{weakly stable} if there exists an alternative IF prospect $\sigma$ that hospital $A$ prefers over its current prospect, $\pi(A)$, and the doctors in the support of $\sigma$ prefer hospital $A$ over the hospitals they are matched to in $\pi$.
    
Moreover, strongly unstable allocations can exist even if the preferences of the doctors are different. Consider the same example as described above: $\H = \set{A, B, C}$ and $\D = \set{i_1, i_2, j}$ where $d(i_1, i_2) = 0$ and $d(i_1, j) = d(i_2,j) = 1$. But now, hospital $A$ is the most preferred hospital by doctors $i_1$ and $j$, but not by doctor $i_2$, and hospital $A$'s preferences are still $i \succ j$. Suppose that $B \succ_{i_2} A$, that $\pi(i_2) = \set{B, \textit{w.p. 1}}$ and that the allocation of hospitals $A$ and $C$ has not been determined yet. This time, we would expect that in any fair and stable allocation $\pi$, hospital $A$'s prospect would be:
\begin{align*}
    \pi(A) = 
    \begin{cases}
        i_1,&\text{ w.p. }1
    \end{cases}.
\end{align*}
Note that this is not an IF alternative. Thus, we need only require that the alternative allocation is IF over the doctors that prefer hospital $A$ over the hospitals in their prospect (in the previous example $\set{i_1, i_2}$, but in this example only $\set{i_1}$).

With this in mind, we formally define \emph{weak ex-ante stability}:
\begin{definition}[Weak Ex-Ante Stability]
\label{def:selective_stability}
Let $\pi$ be a probabilistic allocation, $h \in \H$ be a hospital, $\D^* \subseteq \D$ be a set of doctors and $\sigma \in \Delta(\D^*)$ be a distribution. We say that $\nu = (h, \D^*, \sigma)$ is a selectively fair alternative allocation if $\sigma$ satisfies:
\begin{itemize}
    \item For every two doctors $i,j \in \D^*$, the distribution $\sigma$ is individually fair with respect to $i$ and $j$.
    \item For every doctor $i \in \D \backslash \D^*$, if there exists a doctor $j \in \D^*$ such that $d(i,j) < 1$, then for every hospital $h' \in supp\left(\pi(i)\right)$, $h' \succ_i h$.
    \item For every $i \in \D^*$ and $h' \in supp\left(\pi(i)\right)$, $h \succeq_i h'$.
\end{itemize}
We say that $\nu$ is active if $\sigma \succ_h \pi(h)$.
We say that the allocation $\pi$ is weakly ex-ante stable if there are no active selectively fair alternative allocations.
\end{definition}

For this definition too, we provide an approximation.
\begin{definition}[$\tau$-Weak Ex-Ante Stability]
An allocation $\pi$ is $\tau$-weakly ex-ante stable if there exists a weakly ex-ante stable allocation $\pi'$ that satisfies that
\begin{equation*}
    D_{TV}(\pi, \pi') \le \tau.
\end{equation*}
\end{definition}

\begin{claim}
\label{claim:contract_stability_implies_weak_stability}
For the setting of a proto-metric and strictly IF hospital preferences, contract stability implies weak ex-ante stability.
\end{claim}
\begin{proof}
Let $\pi$ be a contract stable allocation. Assume for contradiction that the allocation $\pi$ is not weakly ex-ante stable. That is, that there exists a hospital $h \in H$, a set of doctors $\D^* \subseteq \D$ and $\sigma \in \Delta(\D^*)$ such that $\nu = (h,\D^*,\sigma) $ is an active selectively fair alternative allocation.

Let $C \in \C$ be the highest ranked cluster by hospital $h$ that satisfies that $\Pr[\sigma \in C] > \Pr[\pi(h) \in C]$, since $\sigma \succ_h \pi(h)$ and the preferences of $h$ are strictly IF, there exists such a cluster.  

If there exists a cluster $C' \in \C$ that satisfies that $\Pr[\sigma \in C'] < \Pr[\pi(h) \in C']$, we also have that $C' \succ_h C$, then we have that $\sigma \nsucc_h \pi(h)$, which contradicts the assumption that $\nu$ is an active selectively fair alternative allocation.
Let $C' \in C$, be some cluster that satisfies that $\Pr[\pi(h) \in C'] < \Pr[\pi(h) \in C']$ and $C \succ_h C'$.

Thus, there exist doctors $i \in C$ and $i' \in C'$ that satisfies that 
$\Pr[\sigma = i] > \Pr[\pi(h) = i]$ and
$\Pr[\sigma = i'] < \Pr[\pi(h) = i']$.

Let $h' \in \H$ be a hospital that satisfies that 
$$
\Pr[\pi(h) = i' \wedge \pi(h') = i] > 0.
$$
Since $\nu$ is an active selectively fair alternative allocation, and doctor $i$ is in the support of $\sigma$, $h \succ_i h'$.

Let $\mu = (h,i;h',i')$ be a contract. Since $C \succ_h C'$ and doctor $i'$ in the support of hospital, we have that $\pi^\mu(h) \succ_h \pi(h)$ and the contract $\mu$ is active. By this, we obtain a contradiction to the assumption that the allocation $\pi$ is contract stable.
\end{proof}

\begin{claim}
For the setting of a proto-metric and strictly IF hospital preferences,
$\tau$-contract stability implies $\tau$-weak ex-ante stability.
\end{claim}
\begin{proof}
Let $\pi$ be a $\tau$-contract stable allocation. Denote by $A$ the event where there is an active contract in $\pi$.
From the definition of $\tau$-contract stability, we know that $\Pr[A] \le \tau$. Thus, in the allocation $\pi|\Bar{A}$ there are no active contract, i.e., it is contract stable. From Claim \ref{claim:contract_stability_implies_weak_stability}, the allocation $\pi|\Bar{A}$ is weakly ex-ante stable.

Thus, since the allocation $\pi$ satisfies that $D_{TV}(\pi, \pi|\Bar{A}) = \tau$, it is $\tau$-weakly ex-ante stable.
\end{proof}


Definition \ref{def:selective_stability} is minimal in two senses. First,
it is very fragile. Suppose some doctor in the alternative allocation is matched to a more preferred hospital with a very small probability and less preferred hospitals with a high probability. In that case, the allocation will be weak ex-ante stable, although it might be very close to an unstable allocation. 

Second, this definition also has weak guarantees for the hospitals.
Suppose hospital $B$ is ranked second by doctors $i_1$ and $i_2$ and hospital $B$'s preferences are $i \succ_B j$, one would expect that in any stable allocation, hospital $B$ would be matched to doctor $i_1$ or doctor $i_2$ with probability 1. However, hospital $B$ has no guarantees over the outcome.

While Definition \ref{def:selective_stability} is quite minimal, it can be relaxed even further. In particular, we allow the hospital $h$ to propose an IF alternative. In Section \ref{sec:generalizing_stability_fails}, we show that requiring an IF alternative allocation from each hospital is not enough to guarantee that the resulting allocation is PIIF. 
This can happen if the accumulated distance between two similar individuals in the allocation is larger than the distance in the prospect of each hospital, and the individuals have similar preferences.
One possibility is to make the stronger requirement that after the contract is activated, PIIF is maintained, which makes the definition weaker. This presents difficulties (see Section \ref{sec:generalizing_stability_fails}), and it only guarantees a very weak notion of stability, but it may allow for algorithms that satisfy both fairness and a weak stability notion.


\subsubsection{Generalization of Contract Stability to General Metrics}
\label{sec:generalizing_stability_fails}
We presented a strong definition for ex-ante stability for proto-metrics in \Cref{sec:contract_stability}, and a weak defintion for ex-ante stability for general metrics in \Cref{sec:weak_ex_post_stability}.
In this section, we discuss the generalization of Definition \ref{def:contract_stability_proto} to general metrics. We present several examples that demonstrate the subtleties of this generalization.

Suppose a hospital has the incentive to offer a doctor (or doctors) to be matched to it and knows the doctors will accept. In that case, this hospital does not have an incentive to participate in the matching.
Thus, we would like that, in such a case, the allocation would be considered unstable. However, if there is no restriction over the hospital's ability to suggest matches, the stability guarantee might be incompatible with fairness, even if the preferences of the hospital are fair.
Consider the example where
$$
\D = \set{i,j}, \quad \H = \set{A, B}, \quad d(i,j) = 1/3
$$
with the preferences
$$
A \succ_i B, \quad A \succ_j B, \quad \begin{cases}
    i \succ_A j&\text{ w.p. } 2/3 \\ j \succ_A i&\text{ w.p. } 1/3
\end{cases}
, \quad 
\begin{cases}
    i \succ_B j&\text{ w.p. } 2/3\\j \succ_B i&\text { w.p. } 1/3
\end{cases}
.
$$
These preferences are mutual-replacement individually fair.
Consider the probabilistic allocation 
$$
\pi_1 = \begin{cases}
    (A, i), (B, j)& \textit{w.p. }p \\
    (A, j), (B, i)& \textit{w.p. } 1-p.
\end{cases}
$$

The allocation $\pi_1$ is PIIF wherever $1/3 \le p \le 2/3$. For any value of $p$ in this range, hospital $A$ prefers to propose to doctor $i$ to always be matched to it since
\begin{align*}
    \forall p \in [1/3, 2/3]: 
    \Pr_{o\sim \pi_1(A),r_A}[r_A^{-1}(o) = 1] &= 2/3 \cdot p  + 1/3 \cdot (1 - p) \\ 
    & \le 5/9 < 2/3 = \Pr[r_A^{-1}(i) = 1].
\end{align*}
If we allow hospital $A$ the alternative of always being matched to doctor $i$, a fair and stable solution for the problem above does not exist. However, this might not be a ``real'' issue, since it is an unfair alternative. 

The example above demonstrates that the case of general metrics requires more caution than the proto-metric case. To ensure the alternative does not contradict fairness, it must be an allocation over all the doctors and not only a single doctor, which satisfies some fairness conditions.

One natural generalization to the concept of contracts is to let the hospital $h$ have a set of doctors $\D'$ that $h$ is unsatisfied with. Whenever hospital $h$ is matched to a doctor in $\D'$, it samples a doctor $d$ from some alternative prospect $\sigma$ and if doctor $d$ prefers hospital $h$ over its match, doctor $d$ and hospital $h$ are matched and doctor $d$ and hospital $h$'s matches are matched. We want such contracts to be active if hospital $h$ prefers the resulting prospect over its prospect. 

Defining the fairness constraint, however, is non-trivial.
We do not necessarily mind considering contracts that lead to unfair behavior, except that they might cause incompatibility with fairness.

We present a generalization for contract stability, where we require fairness from the alternative allocation and weakly require fairness from the resulting allocation. For the setting of proto-metric and strict IF hospital preferences, this definition is equivalent to Definition \ref{def:contract_stability_proto}.

\begin{definition}[Set Contract Stability]
\label{def:contract_stability_general}
Let $h \in H$\ be a doctor, $\D' \subseteq \D$ a set of doctors, $a\in [0,1]$ a constant and $\sigma \in \Delta(\D \backslash \D')$ a distribution over $\D \backslash \D'$. 

We say that $\mu = (h,\D',a,\sigma)$ is a set contract with respect to an allocation $\pi$ if:
\begin{itemize}
    \item For every two $i, j \in supp(\sigma)$, $\sigma$ is IF with respect to $i$ and $j$.
    \item Define $\pi^\mu$ as follows: if $\pi(h) \in \D'$, then w.p. $1-a$, $\pi^\mu(h) = \pi(h)$ and with probability $a$, hospital $h$ samples a doctor $i$ from $\sigma$ and ``proposes'' to $i$, and $i$ accepts the ``proposal'' if $h \succ_i \pi(i)$. If $i$ accepts, then $\pi^\mu(h) = i$ and $\pi^\mu(\pi(i)) = \pi(h)$, otherwise $\pi^\mu(h) = \pi(h)$ and $\pi^\mu(i) = \pi(i)$.
    
    If $\pi(h) \not\in \D'$, then $\pi^\mu(h) = \pi(h)$.
    \item For every two doctors $i \in \D \backslash \D'$ and $j \in \D'$, the allocation $\pi^\mu$ is PIIF with respect to $i$ and $j$ (in both directions).
\end{itemize}

We say the set contract is active if $\pi^\mu(h) \succ_h \pi(h)$.

We say an allocation $\pi$ is set contract stable if there are no active set contracts.
\end{definition}

In the definition above, we have a PIIF requirement only from pairs where one of the doctors is in $\D'$ and the other is in $\D\backslash\D'$. We show the intuition behind this decision by example:
$\D = \set{i_1, i_2, j, k}$, $\H= \set{A, B, C, D}$, $d(i_1, i_2) = 0$ and the other distances are 1. The preferences of the doctors are $A \succ B \succ C \succ D$ and the preferences of hospital $A$ are $k \succ i \succ j$. Consider the allocation:
\begin{align*}
    \pi_2 = \begin{cases}
        (i_1, A), (i_2, B), (k, C), (l, D), &\textit{w.p. } 1/2 \\
        (i_2, A), (i_1, B), (k, D), (l, C), &\textit{w.p. } 1/2
    \end{cases}
\end{align*}
This allocation is PIIF, but seems unstable. All the doctors prefer hospital $A$, so hospital $A$ should be assigned to its most preferred doctor, $k$. However, ``applying'' the alternative allocation $\set{k \text{ w.p. }1}$ yields the allocation 
\begin{align*}
    \pi_2' = \begin{cases}
        (k, A), (i_2, B), (i_1, C), (l, D), &\textit{w.p. } 1/2 \\
        (k, A), (i_1, B), (i_2, D), (l, C), &\textit{w.p. } 1/2
    \end{cases}
\end{align*}
which is not PIIF because doctor $i_2$ envies doctor $i_1$.  However, since this behaviour is not hospital $A$'s ``fault'', we allow such contracts to be active, where it is important to keep in mind that we require that no such contract exists and do not actually use the activation of the contracts. Thus, it is only important to ensure that the stability definition does not lead to incompatibility with fairness, i.e., that there exists a PIIF allocation where there are no active contracts. 

\paragraph{On fairness of the alternative prospect.}
Another possibility for a fairness requirement is to require that the alternative prospects are IF, as in Definition \ref{def:selective_stability}. However, the fact that each hospital's prospect is IF does not imply that the entire allocation is PIIF:
For $\D = \set{i,j,k,l}$ and $\H = \set{A, B, C, D}$. The preferences of all doctors are $A \succ B \succ C \succ D$, $d(i,j) =1/5$ and the other distances are 1. The allocation

\begin{align*}
    \pi_3 = \begin{cases}
        (i,A), (j,B), (k,C), (l,D) , &\textit{w.p. } 2/5 \\
        (i,B), (j,C), (k,A), (l, D) , &\textit{w.p. } 1/5 \\
        (i,C), (j, A), (k,B), (l, D) , &\textit{w.p. } 1/5 \\
        (i,C), (j, D), (k, B), (l, A) , &\textit{w.p. } 1/5
    \end{cases}
\end{align*}
satisfies that for each hospital $h \in \H$ the prospect $\pi_3(h)$ is IF. However, $\pi_3$ is not PIIF since $\pi_3(i) \succ_j \pi_3(j)$ and $D_{TV}(\pi_3(i), \pi_3(j))  = 2/5 > 1/5 = d(i,j)$.

Thus, we considered, instead of enforcing fairness on the proposed alternative, to enforce fairness over the resulting allocation via a PIIF requirement between pairs of doctors in $\D$.
However, this requirement turned out to be too strong (i.e., it results in a weak stability definition).
Consider the example where
\begin{align*}
&\D = \set{i_1, i_2, j}, \quad \H = \set{A, B, C}, \\
&d(i_1,i_2) = 0, \quad d(i_1,j) = 1, \quad d(i_2,j) = 1,
\end{align*}
with the mutual replacement IF preferences
\begin{align*}
&A \succ_{i_1} B \succ_{i_1} C, \quad 
A \succ_{i_2} B \succ_{i_2} C, \quad 
A \succ_j B \succ_j C, \\
&i \succ_A j, 
\quad i \succ_B j,
\quad i \succ_C j.    
\end{align*}

It looks like the only fair and stable allocation in this case should be 
\begin{align*}
    \pi_4 = \begin{cases}
        (i_1, A), (i_2, B), (j, C), &\textit{ w.p. } 1/2\\
        (i_1, B), (i_2, A), (j, C), &\textit{ w.p. } 1/2.
    \end{cases}
\end{align*}

However, the allocation 
\begin{align*}
    \pi_5 = \begin{cases}
        (i_1, A), (i_2, B), (j, C), &\textit{ w.p. } 1/3\\
        (i_1, C), (i_2, A), (j, B), &\textit{ w.p. } 1/3\\
        (i_1, B), (i_2, C), (j, A), &\textit{ w.p. } 1/3
    \end{cases}
\end{align*}
is PIIF, and any alternative prospect that hospital $A$ can propose that involves doctors $i_1$ and $i_2$ will result in an unfair allocation, where doctor $i_2$ is matched to hospital $B$ with higher probability than doctor $i_1$. Thus, if we require PIIF between doctors $i_1$ and $i_2$ in the resulting allocation, the allocation $\pi_5$ would be considered stable. We find this makes the stability requirement too weak.

\subsection{Ex-Post (Local) Stability}
\label{sec:ex_post_stability}
The definitions above are ex-ante (a priori) definitions, i.e., guarantees for the probabilistic allocation before sampling. We also present an ex-post (a posteriori) definition, i.e., a guarantee for the sampled allocation, \emph{local stability}. This definition is stronger than the ex-ante definitions, we prove this on Claim \ref{claim:local_stability_implies_contract_stability}, and we will use it in some of our negative results in Section \ref{sec:failure_in_composition}.

\begin{definition}[Local Stability]
\label{def:local_stability}
Given deterministic doctor preferences $\mathcal{P}_{\D} = \{r_i\}_{i\in \D}$ and probabilistic hospital preferences $\mathcal{P}_{\H} = \{r_h\}_{h\in \H}$, a probabilistic allocation $\pi$ is locally stable if there exists a joint distribution $\nu$ over the preferences and assignment $(\mathcal{P}_{\H}, \pi)$, such that for every possible deterministic hospital preferences $\Theta \in \Pi_\D^n$ and deterministic allocation $a: \D \to \H$, if 
\begin{align*}
    \Pr_{(P,m) \sim \nu}[P = \Theta \wedge m = a] > 0
\end{align*}
then there is no blocking pair in $a$ with respect to the preferences $\Theta$.
\end{definition}

In Definition \ref{def:local_stability}, we assume that the preferences of the hospital are given, and we want to find an allocation that satisfies a condition. Thus, although the condition is over the joint distribution of the preferences and the allocation, we say that the allocation is locally stable (or unstable) and not the joint distribution.

We note that local stability is equivalent to requiring $\pi$ is generated by running a stable matching algorithm on $\mathcal{P}_{\D}$ and $\mathcal{P}_{\H}$.

\subsubsection{Local Stability Implies Contract Stability}

This section shows the connection between ex-post and ex-ante definitions, focusing on proto-metrics and strict individually fair preferences.

\begin{claim}
\label{claim:local_stability_implies_contract_stability}
Let $\pi$ be a locally stable probabilistic allocation with respect to strict IF hospital preferences (Definition \ref{def:strict_IF}), then $\pi$ is contract stable (Definition \ref{def:contract_stability_proto}).
\end{claim}
\begin{proof}
Assume for contradiction that there is an active contract, i.e., there exists  a contract $\mu = (h, i; h', i') \in \H \times \D \times \H \times \D$ such that $h \succ_i h'$, $\pi^\mu(h) \succ_h \pi(h)$ and
\begin{equation}
\label{eq:active_contract_proof_for_claim_local_contract}
    \Pr[\pi(i') = h \wedge \pi(i) = h'] > 0.
\end{equation}
Since the preferences of $h$ are strict individually fair and $\mu$ is an active contract we know that $d(i,i') = 1$. Let us denote by $C$ and $C'$ the clusters that contain doctors $i$ and $i'$ respectively. Since $\pi^\mu(h) \succ_h \pi(h)$, we know that $C \succ_h C'$.

Since the allocation $\pi$ is locally stable there exists a joint distribution over $\mathcal{P}_{\H}$ and $\pi$, $\nu$, as described in the definition. Let us denote by $\nu_h$ the joint distribution of $r_h$ and $\pi$ that is a marginal distribution of $\nu$.
From the definition of local stability and \cref{eq:active_contract_proof_for_claim_local_contract}, there exists a permutation $\theta$ of the doctors in $\D$ and deterministic allocation $a$, that satisfies:
\begin{itemize}
    \item $\Pr_{(o,m) \sim \nu_h}[o = \theta \wedge m = a] > 0$.
    \item $a(i') = h, a(i) = h'$. 
\end{itemize} 

Since $C \succ_h C'$, we have that $i \succ_h i'$, for any sampled preferences in $r_h$, and specifically for $\theta$. Since $h \succ_i h'$ and hospital $h$ and doctor $i$ are not matched in $a$, the pair $(h,i)$ is a blocking pair for the allocation $a$, which contradicts the assumption that $\pi$ is locally stable.
\end{proof}
\section{Gale-Shapley on Fair Preferences Does Not Guarantee Fairness}
\label{sec:failure_in_composition}

For the classical setting of deterministic preferences and deterministic allocations, the celebrated \emph{Gale-Shapley} algorithm \cite{gale1962college} finds a stable matching in polynomial time.
It is natural to ask whether we can find a fair and stable allocation by sampling from the (probabilistic) fair preferences and run the Gale-Shapley algorithm.

Since the Gale-Shapley algorithm guarantees classical stability in the deterministic case, the algorithm described above ensures local stability. 

In this section, we show that running this canonical algorithm over fair preferences does not always lead to a fair allocation. We show this for the case of proto-metric with the strongest fairness requirement over the hospitals' preferences, strict individual fairness. This joins existing work due to Dwork and Ilvento \cite{dwork2018fairness} of fairness under composition that shows that even if all the participants act fairly, the outcome might be unfair. The question of whether local stability and PIIF are compatible is still open.

\subsection{The Gale-Shapley Algorithm}
\label{sec:gs_alg}
The Gale-Shapley (GS) algorithm is described in Figure \ref{alg:Gale_Shapley}. We assume we are given two sets of individuals $A$ and $B$ of size $n$. Each individual has deterministic ordered preferences over the individuals in the other set. 
The algorithm returns a \emph{matching}.

\begin{definition}[Matching]
A matching is a function $m : A \cup B \rightarrow A \cup B$ where:
\begin{itemize}
    \item For every $i \in A$, if $m(i) \in B$, $m(i)$ represents $i$'s match. Otherwise, $m(i) = i$, and $i$ is not matched. $m(i) \in A \backslash \set{i}$ is not a valid value.
    \item For every $j \in B$, if $m(j) \in A$, $m(j)$ represents $j$'s. Otherwise, $m(j) = j$, and $j$ is not matched. $m(j) \in B \backslash \set{j}$ is not a valid value.
    \item The matching must be mutual so if $m(i) = j$ then $m(j) = i$. Thus, if $m(i) = j$ then for any other individual $i'$, $m(i') \ne j$.
\end{itemize}
\end{definition}

The algorithm is not symmetric; there is a proposing set and an accepting set. We will assume that $A$ is the proposing set and $B$ is the accepting set.
At each iteration of the algorithm, some unmatched individual in $A$ proposes to its most favored individual in $B$ that did non reject it yet. If there is no unmatched individual in $A$, the algorithm terminates.
This algorithm is also called the \emph{deferred acceptance} algorithm since when an individual in $B$ gets a proposal, if it is worst than its current match, they will reject it immediately. However, if it is better than its current match, they will temporarily accept the proposal and reject its current match, but they can still reject the proposal if a better one comes along. The acceptance is final once the algorithm terminates.

\makeatletter
\renewcommand{\ALG@name}{Figure}
\makeatother

\begin{algorithm}[H]
\caption{Gale-Shapley Algorithm}
\label{alg:Gale_Shapley}
\begin{algorithmic}[1]
\STATE $\forall i \in A : m(i) \leftarrow i$, $\forall j \in B : m(j) \leftarrow j$.
\WHILE{$\exists i \in A: m(i) = i$}
    \STATE Choose arbitrarily $i \in A : m(i) = i$.
    \STATE Let $j \in B$ be the first on $i$'s preferences list.
    \IF{$m(j) = j$}
        \STATE $m(j) \leftarrow i$, $m(i) \leftarrow j$.
    \ELSIF{$m(j) \succ_j i$}
        \STATE Remove $j$ from $i$'s preferences list.
    \ELSE
        \STATE Remove $j$ from $m(j)$'s preferences list.
        \STATE $m(m(j)) \leftarrow m(j)$, $m(j) \leftarrow i$, $m(i) \leftarrow j$.
    \ENDIF
\ENDWHILE
\RETURN $m$.
\end{algorithmic}
\end{algorithm}

\makeatletter
\renewcommand{\ALG@name}{Algorithm}
\makeatother

Gale and Shapley \cite{gale1962college} showed that the GS algorithm terminates in $O(n^2)$ rounds and outputs a stable matching that is optimal among all stable matchings for each individual in the proposing set.

\subsection{Doctor Proposing Gale-Shapley}

We start by showing that for the GS variant where the doctors propose to the hospitals, there exists strict individually fair preferences that lead to an unfair solution.

\begin{theorem}
\label{thm:doctor_propose_gs_fails}
There exists a set of doctors $\D$, a set of hospitals $\H$, a proto-metric $d : \D \times \D \rightarrow \set{0,1}$, deterministic doctor preferences $\mathcal{P}_{\D} = \{r_i\}_{i\in \D}$ and probabilistic and strictly individually fair hospital preferences $\mathcal{P}_\H = \{r_h\}_{h\in \H}$, such that an algorithm that outputs a probabilistic allocation $\pi$ by sampling $\mathcal{P}_{\H}$ and running the Gale-Shapley algorithm with $\D$ as the proposing set, outputs an unfair allocation.
\end{theorem}

\begin{proof}
Consider the following example:
Let $\D = \set{i_1, i_2, j}$ be the set of doctors, $\H = \set{A, B, C}$ be the set of hospitals. The doctors $i_1, i_2$ are in the same cluster $i$, and $j$ is in a different cluster $j$.

The preference of the doctors are: 
$$
A \succeq_{i_1} B \succeq_{i_1} C, \quad A \succeq_{i_2} C \succeq_{i_2} B, \quad C \succeq_j A \succeq_j B,
$$
and the preferences of the hospitals are
$$
j \succeq_A i, \quad i \succeq_B j, \quad i \succeq_C j.
$$ 

When running GS, at the beginning both doctors $i_1$ and $i_2$ propose to hospital $A$ and doctor $j$ proposes to hospital $C$, hospital $C$ always accepts doctor $j$ at this point. 
There are two cases for the preferences of hospital $A$:
\begin{enumerate}
    \item Hospital $A$'s preferences are $i_1 \succ_A i_2$.
    Hospital $A$ accepts doctor $i_1$ and rejects doctor $i_2$. Doctor $i_2$ proposes to hospital $C$, hospital $C$ accepts doctor $i_2$ and rejects doctor $j$. Doctor $j$ proposes to hospital $A$, hospital $A$ accepts doctor $j$ and rejects doctor $i_1$. Doctor $i_1$ proposes to hospital $B$ and hospital $B$ accepts. 
    All the doctors are matched so the algorithm terminates with the allocation $(i_1, B), (i_2, C), (j, A)$.
    \item Hospital $A$'s preferences are $i_2 \succ_A i_1$.
    Hospital $A$ accepts doctor $i_2$ and rejects doctor $i_1$. Doctor $i_1$ proposes to hospital $B$, and hospital $B$ accepts. 
    All the doctors are matched so the algorithm terminates with the allocation $(i_1, B), (i_2, A), (j, C)$.
\end{enumerate}

Hence the algorithm outputs the allocation
$$ \pi = 
\begin{cases}
(i_1, B), (i_2, C), (j, A),& \text{w.p. } 1/2 \\
(i_1, B), (i_2, A), (j, C),& \text{w.p. } 1/2.
\end{cases}
$$

By doctor $i_1$'s preferences
$$
\Pr_{o\sim \pi(i_1)}[r_{i_1}^{-1}(o) = 1] = 0 < \Pr_{o\sim \pi(i_2)}[r_{i_1}^{-1}(o) = 1] = 1/2,
$$
 which implies that $\pi(i) \not\succeq_i \pi(j)$ so allocation $\pi$ is not PIIF.
\end{proof}

\subsection{Hospital Proposing Gale-Shapley}

Next, we show that for the GS variant where the hospitals propose the doctors, there exist strict individually fair preferences that leads to an unfair solution.

\begin{theorem}
There exists a set of doctors $\D$, a set of hospitals $\H$, a proto-metric $d : \D \times \D \rightarrow \set{0,1}$, deterministic doctor preferences $\mathcal{P}_{\D} = \{r_i\}_{i\in \D}$, probabilistic and strictly individually fair hospital preferences $\mathcal{P}_\H = \{r_h\}_{h\in \H}$, such that an algorithm that outputs a probabilistic allocation $\pi$ by sampling $\mathcal{P}_H$ and running the Gale-Shapley algorithm with $\H$ as the proposing set, outputs an unfair allocation.
\end{theorem}

\begin{proof}
Consider the following example:
Let $\D = \set{i_1, i_2, j, k}$ be the set of doctors, $\H = \set{A, B, C, D}$ be the set of hospitals. The doctors $i_1, i_2$ are in the same cluster $i$, and doctors $j$ and $k$ are in different clusters $j$ and $k$ respectively.

The preference of the doctors are: 
$$A \succeq_{i_1} B \succeq_{i_1} C \succeq_{i_1} D, \quad B \succeq_{i_2} A \succeq_{i_2} D \succeq_{i_2} C, \quad C \succeq_j B \succeq_j A \succeq_j D, \quad D \succeq_k A \succeq_k B \succeq_k C$$
and the preferences of the hospitals are
$$i \succeq_A j \succeq_A k, \quad j \succeq_B i \succeq_B k, \quad i \succeq_C j \succeq_C k, \quad i \succeq_D k \succeq_D j.$$

Let us simulate the run of GS. At the beginning of the run, hospital $B$ proposes to doctor $j$ and doctor $j$ accepts. As for hospital $A$,
there are two cases for the preferences of hospital $A$:

\begin{enumerate}
    \item Hospital $A$'s preferences are $i_1 \succ_A i_2$.
    Hospital $A$ proposes to doctor $i_1$ and doctor $i_1$ accepts.
    If hospital $D$'s preferences are $i_1 \succ_D i_2$, hospital $D$ proposes to doctor $i_1$ and doctor $i_1$ rejects it, then hospital $D$ proposes to doctor $i_2$ and doctor $i_2$ accepts.
    If hospital $D$'s preferences are $i_2 \succ_D i_1$, hospital $D$ proposes to doctor $i_2$ and doctor $i_2$ accepts.
    Hospital $C$ proposes to doctors $i_1$ and $i_2$ in some order and they both reject it, then hospital $C$ proposes to doctor $j$ and doctor $j$ accepts.
    If hospital $B$'s preferences are $i_1 \succ_B i_2$, hospital $B$ proposes to doctor $i_1$ and doctor $i_1$ reject it, then hospital $B$ proposes to doctor $i_2$ and doctor $i_2$ accepts.
    If hospital $B$'s preferences are $i_2 \succ_B i_1$, hospital $B$ proposes to doctor $i_2$ and doctor $i_2$ accepts.
    If hospital $D$'s preferences are $i_2 \succ_D i_1$, hospital $D$ proposes to doctor $i_1$ and doctor $i_1$ rejects.
    Hospital $D$ proposes to doctor $k$ and doctor $k$ accepts.
    All the hospitals are matched so the algorithm terminates with the matching
    $
    (i_1, A), (i_2, B), (j, C), (k, D).
    $

    \item Hospital $A$'s preferences are $i_2 \succ_A i_1$.
    Hospital $A$ proposes to doctor $i_2$ and doctor $i_2$ accepts.
    If hospital $C$'s preferences are $i_2 \succ_C i_1$, hospital $C$ proposes to doctor $i_2$ and doctor $i_2$ reject it, then hospital $C$ proposes to doctor $i_1$ and doctor $i_1$ accepts.
    If hospital $C$'s preferences are $i_1 \succ_C i_2$, hospital $C$ proposes to doctor $i_1$ and doctor $i_1$ accepts.
    Hospital $D$ proposes to doctors $i_1$ and $i_2$ in some order and they both reject it, then hospital $D$ proposes to doctor $k$ and doctor $k$ accepts.
    All the hospitals are matched so the algorithm terminates with the matching
    $
    (i_1, C), (i_2, A), (j, B), (k, D).
    $
    
    Hence the algorithm outputs the allocation
    $$ \pi =
    \begin{cases}
    (i_1, A), (i_2, B), (j, C), (k, D) ,& \text{w.p. } 1/2 \\
    (i_1, C), (i_2, A), (j, B), (k, D) ,& \text{w.p. } 1/2.
    \end{cases}
    $$
    
    By doctor $i_1$'s preferences
    $$
    \Pr_{o\sim \pi(i_1)}[r_{i_1}^{-1}(o) \le 2] = 1/2 < \Pr_{o\sim \pi(i_2)}[r_{i_1}^{-1}(o) \le 2] = 1
    $$
     which implies that $\pi(i) \not\succeq_i \pi(j)$ so  the allocation $\pi$ is not PIIF.
\end{enumerate}
\end{proof}
\section{PIIF and Contract Stable Algorithms Given a Proto-Metric}
\label{sec:piif_algs}
In the previous section, we demonstrated that a natural extension of the Gale-Shapley algorithm to our setting, in which the hospitals have probabilistic preferences, fails to guarantee fairness. This leaves open the question of whether fairness and stability are inherently incompatible, of whether the Gale-Shapley algorithm is simply not the “right” approach. In this section, we resolve this question in the proto-metric setting, proving that fairness and stability are compatible. Our proof is constructive: we give two algorithms (akin to the “doctor proposing” and “hospital proposing” variants of the Gale-Shapley algorithm) that are guaranteed to output a fair and stable solution when the preferences of the hospitals are fair to begin with.

We show that for a proto-metric and strict individually fair hospital preferences, there exists an algorithm that outputs a PIIF and contract stable allocation.
We start by showing the two algorithms for PIIF and contract stable allocations for strict individually fair hospital preferences. Then, we show a reduction from rank individually fair preferences to strict individually fair preferences that preserves contract stability.

In the two algorithms, we have a meta-algorithm that is similar to the Gale-Shapley algorithm. In this meta-algorithm, each participant in the proposing set proposes probability mass to participants in the other set. Then, the participant in the other set has to decide whether to accept or reject the probability mass. The main difference from the Gale-Shapley algorithm is that participants can decide to accept part of the proposed probability mass. Moreover, the algorithms incorporate fairness considerations when deciding what proposals to make and what proposals to accept.
In the variant where the hospitals propose, they propose to the entire cluster instead of a single doctor. Then the cluster chooses an envy-free allocation.
In the variant where the doctors propose, the hospitals choose an allocation that will not lead to unfairness.

\begin{remark}
In \Cref{sec:reduction}, we show that this result can be generalized. For a weaker requirement over the fairness of the hospital preferences, rank IF, we can achieve PIIF and set contract stability (see \Cref{def:contract_stability_general}). We do that by showing a reduction from rank IF preferences to strict IF preferences that maintains stability.
\end{remark}

\subsection{Allocation Sub-Routines}
\subsubsection{Probabilistic Serial Procedure}
In the first algorithm, \Cref{alg:GS_PSP}, we make use of the \emph{probabilistic serial procedure} (PSP) due to Bogomolnaia and Moulin \cite{bogomolnaia2001new}. This procedure is used to allocate $n$ goods to a set of $n$ individuals, who are assumed to have preferences over those goods. Bogomolnaia and Moulin \cite{bogomolnaia2001new} showed that the output of the probabilistic serial procedure is envy-free.

The intuition for this procedure in the paper is as if there are $n$ individuals and $n$ loaves of bread. Each individual has ordered preferences over the loaves of bread. The duration of the algorithm is one time unit. At each time point $0 \le t \le 1$, each individual eats their favorite bread loaf, that has not been entirely consumed ("eaten\footnote{Because of the analogy above, in the paper, this algorithm is also called the eating algorithm; we also use this name.}") yet. In the original paper, the individuals can have different speeds; here, we employ a simplified version of the algorithm where each individual has speed 1 $\frac{\text{bread loaf}}{\text{time unit}}$. The portion of bread that each individual has eaten from each bread loaf is the probability of this individual to be assigned to this bread loaf in the probabilistic allocation. At the end of the run, we have a doubly stochastic matrix that, using the Birkhoff von Neumann algorithm, indeed a probabilistic allocation.

\paragraph{Example: probabilistic serial procedure.} Consider the following toy examples for the set of individuals is $\mathcal{I}=\set{a,b,c}$ and the set of bread loaves is $\mathcal{B} = \set{\ell_1, \ell_2, \ell_3}$. 
\begin{itemize}
    \item 
    If the preferences are
    $$
    \ell_1 \succ_a \ell_2 \succ_a \ell_3 ,\quad \ell_2 \succ_b \ell_1 \succ_b \ell_3, \quad \ell_3 \succ_c \ell_1 \succ_c \ell_2.
    $$
    Then at the beginning of the run, each individual eats their favorite bread loaf. Since only one individual eats each bread loaf, they finish eating at time 1 when the algorithm terminates. The output is the doubly stochastic matrix
    \begin{center}
        \begin{tabular}{c|c|c|c}
             & $\ell_1$ & $\ell_2$ & $\ell_3$ \\
             \hline
             $a$ & 1 & 0 & 0 \\
             \hline
             $b$ & 0 & 1 & 0 \\
             \hline
             $c$ & 0 & 0 & 1
        \end{tabular}
    \end{center}
    \item If the preferences are
    $$
    \ell_1 \succ_a \ell_2 \succ_a \ell_3,\quad \ell_1 \succ_b \ell_3 \succ_b \ell_2, \quad \ell_2 \succ_c \ell_1 \succ_c \ell_3.
    $$
    Then at the beginning of the run, both $a$ and $b$ start eating $\ell_1$ and $c$ starts eating $\ell_2$. At time 1/2, there will be no more bread left in $\ell_1$ and $c$ finished eating half of $\ell_2$. Then $a$ starts eating $\ell_2$ and $b$ starts eating $\ell_3$. At time 3/4, there will be no more bread left in $\ell_2$ and $3/4$ left in $\ell_3$. Then $a$ and $c$ start eating $\ell_3$, at time time 1 there is no more bread left in $\ell_3$ and the algorithm terminates with the doubly stochastic matrix
    \begin{center}
        \begin{tabular}{c|c|c|c}
             & $\ell_1$ & $\ell_2$ & $\ell_3$ \\
             \hline
             $a$ & 1/2 & 1/4 & 1/4 \\
             \hline
             $b$ & 1/2 & 0 & 1/2 \\
             \hline 
             $c$ & 0 & 3/4 & 1/4
        \end{tabular}
    \end{center}
\end{itemize}

In the case where the sizes of the sets are not equal Bogomolnaia and Moulin \cite{bogomolnaia2001new} and  Kojima and Manea \cite{kojima2010incentives} showed generalizations using dummy individuals or dummy bread loaves, where being assigned to a dummy individual or bread loaf means being unassigned.

\setcounter{algorithm}{0}
\begin{algorithm}[H]
    \caption{Probabilistic Serial Procedure}
    \hspace*{\algorithmicindent} \textbf{Input} $\D$, $\C$, $\H$, $\forall C \in \C:p_C \in [0,1]^{|\H|}$, $\forall i \in \D:r_i \in \Pi_\D$.
    
    \begin{algorithmic}[1]
    \label{alg:PSP}
        \STATE $t \leftarrow 0$, $\forall h \in H: p_h \leftarrow p_C[h]$, $\forall i \in \D: \pi(i) \in [0,1]^{|H|} \leftarrow \mathbf{0}$.
        \STATE $t$ goes from 0 to 1 with speed 1.
        \WHILE{$t < 1$ and $\exists h'\in H: p_{h'} > 0$} 
            \FORALL{$i \in \D$ (simultaneously)} 
                \STATE $h$ is the highest ranked hospital in $i$'s preferences list.
                \WHILE{$p_h > 0$}
                    \STATE Eat $p_h$ with speed 1.
                \ENDWHILE
                \STATE Remove $h$ from $i$'s preferences list.
                \STATE $\pi(i)[h] \leftarrow$ the amount $i$ ate from $h$.
            \ENDFOR
        \ENDWHILE
        \RETURN $\forall i\in \D: \pi(i)$.
    \end{algorithmic}
\end{algorithm}

\subsubsection{Rising Tide Algorithm}
Similarly, in the second algorithm, \Cref{alg:GS_WA},
each hospital has to choose its prospect given the proposed probability mass in every round in this algorithm.
We use \emph{rising tide Algorithm}, described in \Cref{alg:water_algorithm}, to choose a probabilistic prospect for a specific hospital given the doctors' proposals.  Unlike the PSP algorithm, in this case we find an prospect for a single hospital and not for an entire cluster.
This algorithm outputs the "best" output for the hospital that will not violate the fairness requirement.

It may be helpful to think about each doctor as a vessel with a capacity equal to the proposed probability mass, and the vessels of the doctors in the same cluster are connected. The hospital has one unit of water. Going from its most preferred cluster to its least preferred cluster, the hospital pours water in the vessels that belong to the cluster until the vessels are full or there is no water for the hospital to pour.

The amount of water in each doctor's vessel is the probability of this doctor being assigned to the hospital.

In \Cref{alg:water_algorithm}, we denote the capacity of the doctors' vessels (the proposed probability mass) as the vector $p_h$, where $p_h[i]$ is the capacity of the doctor $i$.

\begin{algorithm}[H]
    \caption{Rising tide}
    \hspace*{\algorithmicindent} \textbf{Input} a hospital $h \in \H$, $p_h \in [0,1]^{n}$, $\D$,$\mathcal{C}$, $r_h:[|\C|] \to \C$.
    \begin{algorithmic}[1]
    \label{alg:water_algorithm}
    \STATE $p \leftarrow 1$, $\pi \in [0,1]^{|D|} \leftarrow \mathbf{0}$.
    \WHILE{$p > 0$ and $\exists i \in \D: p_h[i] > 0$}
        \STATE $C \leftarrow \argmin_{C \in \mathcal{C}}\{r_h^{-1}(C): \exists i \in C \,s.t. \, p_h[i] > 0\}$. \hfill\COMMENT Most preferred cluster that still has probability mass proposed to $h$.
        \STATE $n_C \leftarrow |\{ i \in C: p_h[i] > 0 \}|$.
        \hfill\COMMENT Number of doctors in $C$ that has probability mass proposed to $h$.
        \STATE $x \leftarrow \min\{p / n_C, \min_{i \in C}\{p_h[i] : p_h[i] > 0\}\}$.
        \FORALL{$i \in C: p_h[i] > 0$}
            \STATE $p_h[i] \leftarrow p_h[i] - x$.
            \STATE $\pi[i] \leftarrow \pi[i] + x$.
        \ENDFOR
        \STATE $p \leftarrow p - n_C \cdot x$.
    \ENDWHILE
    \RETURN $\pi$.
    \end{algorithmic}
\end{algorithm}

\begin{remark}
In both algorithms in this section, when finding the probabilistic allocation, we first find a doubly stochastic matrix representing the probability of every two parties from the different sets to be matched. Then we use \emph{Birkhoff von Neumann algorithm} due to Birkhoff \cite{birkhoff1946tres}, which decomposes the doubly stochastic matrix into a convex combination of permutation matrices, where a permutation matrix is an allocation when the size of the two sets is equal. Thus, this convex combination of permutation matrices is equivalent to a probabilistic allocation and is the output of the algorithm.
\end{remark}
\subsection{Hospital-Proposing Algorithm}

\subsubsection{Fair Propose-and-Reject -- Hospitals-First}
Next, we present our first algorithm for PIIF and contract stable allocations given strict individually fair hospital preferences. Our algorithm is a generalization of the Gale Shapley algorithm (see Figure \ref{alg:Gale_Shapley}) to preferences over clusters instead of preferences over individuals. 
If each cluster consists of a single doctor, i.e., there is no fairness requirement, then \Cref{alg:GS_PSP} is equivalent to Gale Shapley with the hospitals as the proposing set. 
In the case of a single cluster consisting of all the doctors, i.e., there is a strong fairness requirement among all doctors, \Cref{alg:GS_PSP} is equivalent to \Cref{alg:PSP} (the "eating" algorithm).

At the beginning of the run, each hospital has a free probability mass $p_h = 1$. We sometime denote the free probability mass of hospital $h$ at round $k$ as $p_{h,k}$.

At each round $k$, each hospital $h \in \H$ proposes its entire free probability mass to the cluster with the lowest rank (most favored) that did not yet reject it. As an analogy, it might be helpful to think of proposing free probability mass as free hospitals proposing to the next doctor in their preferences list in the original algorithm.

After receiving the proposals, each cluster runs the eating algorithm over the newly proposed probability mass and the mass from the assignment in the previous round. Any non-allocated probability mass is rejected by the cluster, and becomes free probability mass in the next round. 
As a similar analogy, it might be helpful to think of this part as the accepting set in Gale Shapley choosing between its current allocation and a new proposal and rejecting the less preferred one.

The algorithm terminates when there is no more free probability mass, or the amount of free probability mass is very small. Then the remaining probability mass is distributed between the clusters that are not fully assigned.

\begin{algorithm}[H]
    \caption{Fair Propose-and-Reject -- Hospitals-First}
    \begin{algorithmic}[1]
    \label{alg:GS_PSP}
        \STATE $\forall h \in \H: p_h \leftarrow 1$, $P \in [0,1]^{n \times n} \leftarrow \textbf{0}$, $k \leftarrow 1$, $\forall C \in \C : p_C \in [0,1]^n \leftarrow \mathbf{0}$.
        \WHILE{$\sum_{ h \in \H}p_h > \tau$}
            \FORALL{$h \in \H: p_h > 0$} 
                \STATE $C \in \mathcal{C}$ is the highest ranked cluster by $h$. 
                \STATE $p_C [h] \leftarrow p_C[h] + p_h$.
                \hfill\COMMENT{$h$ proposes its free probability mass.}
                \STATE $p_h \leftarrow 0$.
            \ENDFOR
            \FORALL{$C \in \mathcal{C}$}
                \STATE Run \Cref{alg:PSP} over the probability masses $p_C$ and $C$ and get the marginal distribution  $\pi_{k}(i)$ for every $i \in C$.
                \FORALL{$i \in C, h \in \H$}
                    \STATE $P[i, h] \leftarrow \Pr[\pi_k(i) = h]$.
                    \hfill\COMMENT{Update the output probability matrix.}
                \ENDFOR
                \FORALL{$h \in \H$}
                    \STATE $s \leftarrow p_C[h] - \sum_{i\in C}\Pr[\pi_k(i) = h]$
                    \hfill\COMMENT{Calculate the uneaten probability mass.}
                    \IF{$s > 0$}
                        \STATE $p_h \leftarrow p_h + s$.
                        \hfill\COMMENT{$C$ rejects the amount $s$ of $h$.}
                        \STATE $p_C[h] \leftarrow p_c[h] - s$.
                        \STATE Remove $C$ from $h$'s preference list.
                    \ENDIF
                \ENDFOR
            \ENDFOR
            \STATE $k \leftarrow k + 1$.
        \ENDWHILE
        \label{step_ef}
        \IF{$\sum_{ h \in \H}p_h > 0$}
        \label{alg:step:gs_psp_alloc_mass}
            \FORALL{$h \in \H: p_h > 0$}
                \WHILE{$p_h > 0$}
                    \STATE Find some $C \in \mathcal{C}$ such that $\sum_{h \in \H}p_{C}[h] < |C|$.
                    \hfill\COMMENT{Allocating the remaining mass.}
                    \STATE $x \leftarrow \min\{p_h, |C| - \sum_{h \in \H}p_{C}[h]\}$.
                    \STATE $p_h \leftarrow p_h - x$.
                    \STATE $p_C[h] \leftarrow p_C[h] + x$.
                    \FORALL{$i \in C$}
                        \STATE $P[i,h] \leftarrow P[i,h] + x / |C|$.
                    \ENDFOR
                \ENDWHILE
            \ENDFOR
        \ENDIF
        \STATE Find an allocation $\pi$ by running Birkhoff von Neumann algorithm over $P$.
        \RETURN $\pi$.
    \end{algorithmic}
\end{algorithm}

\begin{remark}
\label{remark:doubly_stochastic}
The allocation of the remaining probability mass $\tau$ forms a doubly stochastic matrix.
This is because the total amount of probability mass is $n$, and the total amount we can allocate with the doctor is $n$. After distributing the mass, all the mass is allocated. By way of choosing how much mass to allocate at each step, we make sure that no doctor or hospital is allocated to more than one unit of mass. Thus, the matrix $P$ is doubly stochastic.
\end{remark}

\begin{theorem}
\label{thm:gs_psp_piif_stable}
Given a set of doctors $\D$, a set of hospitals $\H$, a proto-metric $d : \D \times \D \rightarrow \set{0,1}$, deterministic doctor preferences $\mathcal{P}_{\D} = \{r_i\}_{i\in \D}$, probabilistic and strictly individually fair hospital preferences $\mathcal{P}_\H = \{r_h\}_{h\in \H}$ and a constant $\tau \in (0,n]$, \Cref{alg:GS_PSP} outputs an allocation $\pi$ such that $\pi$ is preference-informed individually fair and $\tau$-contract stable.
\end{theorem}

In order to prove \Cref{thm:gs_psp_piif_stable}, we first prove that the output of \Cref{alg:GS_PSP} satisfies each of the properties separately. This is done in Lemmas \ref{lemma:gs_psp_contract_stability} and \ref{lemma:gs_psp_piif}. 

\paragraph{Notations:} 
For every cluster $C \in \mathcal{C}$, round $k$ and hospital $h \in \H$: $H_k \subseteq \H$ is the set of hospitals that propose to $C$ in round $k$ ($\forall h \in H_k:p_{C,k}[h] > 0$). $C_{h,k} \subseteq C$ is the set of doctors that eat $h$ in round $k$. For every $i \in C_{h,k}$, $s_{i,h,k}$ is the time $i$ started to eat $h$ in round $k$, if $i \not\in C_{h,k}$ then $s_{i,h,k} = 1$.
For every $h \in \H$, $t_{h,k}$ is the latest time in round $k$ that $h$'s probability mass is greater than 0 (this means that if $h \not\in H_k$ then $t_{h,k} = 0$). Note that for every $i \in C_{h,k}$, $t_{h,k} \ge s_{i,h,k}$. 


\subsubsection{The Output of \Cref{alg:GS_PSP} is $\tau$-Contract Stable}

In the following claims we prove that Algorithm \ref{alg:GS_PSP} is contract stable. 

\begin{lemma}
\label{lemma:gs_psp_contract_stability}
The output of Algorithm \ref{alg:GS_PSP} is $\tau$-contract stable.
\end{lemma}

\paragraph{Proof outline.} We prove that at each round, the doctors' outcomes can only improve. Thus, if hospital $h \in \H$ was rejected from cluster $C \in \mathcal{C}$ in round $k$, $C$ will not ``change its mind'' later so there will be no active contracts between $h$ and $C$.

The main claim is \ref{claim:hospital_order}, which shows that for each hospital, eating termination time only increases. In Claim \ref{claim:PS_SD} we show this implies that each doctor's distribution can only (weakly) improve from round to round. In Claim \ref{claim:no_contracts} we show no doctor prefers to be matched to a rejected hospital over their final allocation. Lemma \ref{lemma:gs_psp_contract_stability} shows that this implies contract stability.

\begin{claim}
\label{claim:hospital_order}
Let $C \in \mathcal{C}$ be a cluster of doctors, for every round $k >1$ and hospital $h \in H_k \cap H_{k-1}$, $t_{h,k} \ge t_{h,k-1}$. 
\end{claim}

\begin{proof}
Let $k > 1$ be a round. 

We prove the claim by induction over the order in which hospitals in $H_k \cap H_{k-1}$ were entirely eaten in round $k$. (Note that the claim is immediate for hospitals that were not entirely eaten).

The following observation will be useful in the induction: For some hospital $h \in H_k \cap H_{k-1}$, we can think about the effective probability mass of $h$ in round $k-1$ as the probability mass of $h$ that was eaten. Thus, $h$ will propose at least this probability mass to $C$ in round $k$. (If $p_{C,k-1}[h]$ was entirely eaten, then $h$ can propose in round $k$ a probability mass that was rejected from another cluster in round $k-1$, in addition to $p_{C,k-1}[h]$, such that $p_{C,k}[h] > p_{C,k-1}[h]$). Hence, in the event that $C_{h,k} \subseteq C_{h,k-1}$ and $\forall i \in C_{h,k}: s_{i, h,k} \ge s_{i,h,k-1}$ we know that $t_{h,k} \ge t_{h, k-1}$.
(We note that while this observation will be useful, the condition does not always hold - see below).

\textbf{Base:} Let $h \in H_k \cap H_{k-1}$ be a hospital such that $h \in \arg\min_{h'\in H_k \cap H_{k-1}}t_{h',k}$, i.e. the first hospital in $H_k \cap H_{k-1}$ that got entirely eaten in round $k$.

Let $i \in C_{h,k}$ be some doctor.
We will divide to two cases:
\begin{enumerate}
    \item $\mathbf{s_{i,h,k} = 0}$. 
    This means that doctor $i$ ranks hospital $h$ first in $H_k$. Thus, hospital $h$ must be doctor $i$'s highest ranked hospital in $H_{k-1}$ too.
    Observe that only uneaten hospitals got rejected, so $H_k \cap H_{k-1}$ is the set of hospitals that were eaten in round $k-1$ and they were all available at time 0 in round $k$. If doctor $i$ chose to eat another hospital at time 0 in round $k-1$, then doctor $i$ would have chosen to eat this hospital at round $k$ too.
    Thus, $i \in C_{h,k-1}$ and $s_{i,h,k-1} = 0$.
    
    \item $\mathbf{s_{i,h,k} > 0}$.
    There exists some ``new'' hospital $h'\in H_k \backslash H_{k-1}$, such that $t_{h',k} = s_{i,h,k}$ and $t_{h',k} < t_{h,k}$, that  doctor $i$ eats before hospital $h$ in round $k$. This is true since doctors start to eat hospitals at the beginning of the round or after another hospital was entirely eaten, and we assumed that $s_{i,h,k} > 0$ and hospital $h$ is the first hospital in $H_k \cap H_{k-1}$ to get entirely eaten.
    
    Then, in round $k-1$ $h$ was the first hospital that doctor $i$ ate.
    This is true since all the hospitals in $H_k \cap H_{k-1}$ are available at $t_{h',k}$. Thus, if doctor $i$ would have preferred another hospital in round $k-1$ then doctor $i$ would prefer this hospital in round $k$ too, hence $i \in C_{h, k-1}$ and $s_{i,h,k-1} = 0$.
\end{enumerate}
Thus, $C_{h,k} \subseteq C_{h,k-1}$ and $\forall i \in C_{h,k} : s_{i, h, k} \ge s_{i, h,k-1} = 0$, this implies $t_{h,k} \ge t_{h,k-1}$.

\textbf{Step:} Let $h$ be the $m$-th hospital in $H_k \cap H_{k-1}$ that got entirely eaten. 
Let us divide into two cases:
\begin{enumerate}
    \item There exists a doctor $i \in C_{h,k} \backslash C_{h,k-1}$. Since doctor $i \not\in C_{h,k-1}$ there exists a hospital $h' \in H_{k-1}$ such that $h' \succ_i h$, $t_{h',k-1} \ge t_{h,k-1}$ and $i \in C_{h',k-1}$ (doctor $i$ ate hospital $h'$ in round $k-1$ and by the time hospital $h'$ was entirely eaten and doctor $i$ tried to eat hospital $h$, hospital $h$ was already entirely eaten too).

    Since $i \in C_{h',k-1}$, $h'$ proposes in round $k$ so $h' \in H_k$. Since $d \in C_{h,k}$ and $h' \succ_i h$ we know that $s_{i,h,k}\ge t_{h',k}$ (doctor $i$ starts eating hospital $h$ only after all the hospitals doctor $i$ ranks higher are entirely eaten).
    
    Since $h' \in H_k \cap H_{k-1}$ and $t_{h,k} \ge s_{i,h,k}\ge t_{h',k}$, by the induction assumption $t_{h',k} \ge t_{h',k-1}$. This implies $$t_{h,k} \ge s_{i,h,k}\ge t_{h',k} \ge t_{h',k-1} \ge t_{h,k-1}.$$
    
    \item $C_{h,k} \subseteq C_{h,k-1}$. Let $i \in C_{h,k}$ be a doctor.

    If $s_{i,h,k} = 0$, then from the same arguments as in the base $s_{i,h,k} = s_{i,h,k-1} = 0$.
    
    If $s_{i,h,k-1} = 0$, then trivially $s_{i,h,k} \ge s_{i,h,k-1} = 0$.
  
    Otherwise, $s_{i,h,k} > 0$ and $s_{i,h,k-1} > 0$:
    
    Let $h' \in H_{k-1}$ be the hospital that doctor $i$ ate before hospital $h$ in round $k-1$, i.e. $i \in C_{h',k-1}$ and $s_{i,h,k-1} = t_{h',k-1}$. By this assumption, $h' \succ_i h$. Since $i \in C_{h',k-1}$, hospital $h'$ was not rejected and $h' \in H_k \cap H_{k-1}$.
    
    Since $h' \succ_i h$, we know that $s_{i,h,k} \ge t_{h',k}$. Since $i \in C_{h,k}$, we know that $t_{h,k} \ge s_{i,h,k}$. 
    
    Thus, $h'\in H_k \cap H_{k-1}$ and $t_{h,k} \ge t_{h',k}$ so by the induction assumption $t_{h',k} \ge t_{h',k-1}$, which implies that
    $$
    s_{i,h,k} \ge t_{h',k} \ge t_{h',k-1} = s_{i,h,k-1}.
    $$
    Hence, $\forall i \in C_{h, k}: s_{i,h,k} \ge s_{i,h,k-1}$ and $C_{h,k} \subseteq C_{h,k-1}$, which implies $t_{h,k} \ge t_{h,k-1}$.
\end{enumerate}
\end{proof}

\begin{claim}
\label{claim:pr_eq_max}
Let $i \in \D$ be a doctor in cluster $C \in \mathcal{C}$ and let $h_1, ..., h_n \in \H$ be the hospitals in $\H$ ordered by the preferences of $i$. Let $k$ be a round in $\Cref{alg:GS_PSP}$ and $\pi_k(i)$ be $i$'s prospect after step $k$. Then for every rank $r \in [n]$, 
$$
\Pr_{o\sim \pi_k(i)}[r_i^{-1}(o) \le r] = \max\{t_{h_1, k}, ..., t_{h_r, k}\}
$$
\end{claim}

\begin{proof}
Fix a rank $r \in [n]$.
Observe that doctor $i$ chooses to eat a hospital with rank greater than rank $r$ only if hospitals $h_1, ..., h_r$ are not available, and the first time that this event happens is at time $\max\{t_{h_1, k}, ..., t_{h_r, k}\}$. In the eating algorithm (PSP), the amount of hospitals that a doctor $i$ ate is equal to the probability these hospitals would be assigned to doctor $i$.

Note that if $\sum_{h\in \H} p_{C,k}[h] < |C|$, this equality still hold, but then for every rank $r \in [n]$, 
$$
\Pr_{o\sim \pi_k(i)}[r_i^{-1}(o) \le r] = \max\{t_{h_1, k}, ..., t_{h_r, k}\} < 1.
$$
In that case,
$$
\Pr_{o\sim \pi_k(i)}[o = unallocated]= \Pr_{o\sim \pi_k(i)}[r_i^{-1}(o) = n+1] =1 -  \max\{t_{h_1, k}, ..., t_{h_n, k}\} < 1.
$$
\end{proof}

\begin{claim}
\label{claim:PS_SD}
Let $i \in \D$ be a doctor in cluster $C \in \mathcal{C}$, denote $i$'s allocation after step $k$ by $\pi_k(i)$. Then $\forall k > 1$: $\pi_k(i) \succeq_i \pi_{k-1}(i)$.
\end{claim}

\begin{proof}
Fix a hospital $h \in \H$. If $h \not\in H_{i-1} \backslash H_i$, there are the following possible cases:
\begin{itemize}
    \item If $h \not\in H_{k-1} \cup H_k$, then $t_{h, k} = t_{h, k-1} = 0$.
    \item If $h \in H_k \backslash H_{k-1}$, then $t_{h,k-1} = 0$ so $t_{h, k} \ge t_{h,k-1}$.
    \item If $h\in H_k \cap H_{k-1}$, from Claim \ref{claim:hospital_order} we know that $t_{h, k} \ge t_{h, k-1}$. 
\end{itemize}

Fix a rank $r \in [n]$ and denote $i$'s $r$-th ranked hospital as $h_r$. 
If $\{h_1, ..., h_r\} \cap  H_{k-1} \backslash H_k = \emptyset$, then
$$
\max\{t_{h_1, k}, ..., t_{h_r, k}\} \ge \max\{t_{h_1, k-1}, ..., t_{h_r, k-1}\}
$$

Assume there exists a hospital $h \in \{h_1, ..., h_r\}$ such that $h \in H_{k-1} \backslash H_k$. Then, $t_{h, k-1} = 1$ and $t_{h, k} = 0$. By the algorithm, if $h \in H_{k-1} \backslash H_k$ then hospital $h$ was not eaten in round $k-1$. Thus, there exist a hospital $h' \in H_{k-1}$ such that $t_{h',k-1} = 1$, $h' \succ_i h$ and $i \in C_{h',k-1}$, otherwise doctor $i$ would choose to eat $h$. Since hospital $h'$ was eaten in round $k-1$, we know that $h' \in H_k \cap H_{k-1}$ and from Claim \ref{claim:hospital_order}, $t_{h',k} \ge t_{h',k-1} = 1$. Thus, 

$$
\max\{t_{h_1, k}, ..., t_{h_r, k}\} = \max\{t_{h_1, k-1}, ..., t_{h_r, k-1}\} = 1.
$$

Then, from Claim \ref{claim:pr_eq_max}:
$$
\Pr_{o\sim \pi_k(i)}[r_i^{-1}(o) \le r] = \max\{t_{h_1, k}, ..., t_{h_r, k}\} \ge \max\{t_{h_1, k-1}, ..., t_{h_r, k-1}\} = \Pr_{o\sim \pi_{k-1}(i)}[r_i^{-1}(o) \le r],
$$

which implies that $\pi_k(i) \succeq_i \pi_{k-1}(i)$.
\end{proof}

\begin{claim}
\label{claim:no_contracts}
Let $h \in \H$ be a hospital, $C \in \mathcal{C}$ be a cluster and $m$ be a round such that $\Pr[\pi_m(h) \in C] > 0$. If there exists a cluster $C' \in \mathcal{C}$ such that $C' \succ_h C$, then for every doctor $i \in C'$ and every hospital $h'$ in the support of $\pi_m(i)$, $h' \succeq_i h$, and $\Pr[\pi_m(i) \text{ allocated}]  =1$.
\end{claim}

\begin{proof}
Since $C' \succ_h C$ we know that $h$ proposed to $C'$ at some round $k \le m$ and some of its probability mass got rejected. By the algorithm we know that in round $k$, $i$ chose to eat only hospitals that $i$ ranks at least the same as $h$, since $h$ was available the entire run. Therefore $\Pr_{o\sim \pi_k(i)}[r_i^{-1}(o) \le r_i^{-1}(h)] = 1$. From Claim \ref{claim:PS_SD} $\pi_m(i) \succeq_i \pi_k(i)$. 
Thus, 
$$\Pr_{o\sim \pi_m(i)}[r_i^{-1}(o) \le r_i^{-1}(h)] \ge \Pr_{o\sim \pi_k(i)}[r_i^{-1}(o) \le r_i^{-1}(h)] = 1.$$
Hence, $\Pr_{o\sim \pi_m(i)}[r_i^{-1}(o) \le r_i^{-1}(h)] = 1$, and the claim follows.
\end{proof}

We are now set to prove Lemma \ref{lemma:gs_psp_contract_stability}.

\begin{proof}[Proof of Lemma \ref{lemma:gs_psp_contract_stability}]
We start by showing that at the end of the first part of the algorithm, there are no active contracts, ignoring unallocated mass. In other words, if we denote the allocation at the end of the first part of the algorithm as $\pi_m$, then there are no doctors $i, i' \in \D$ and hospitals $h, h' \in \H$ such that the contract $(h, i; h', i')$ is active.

Assume for contradiction there exists an active contract $\mu = (h, i; h', i')$ in $\pi_m$.
This implies that $h \succ_i h'$, $i \succ_h i'$, doctors $i$ and $i'$ are from different clusters $C, C'$ and $\Pr[\pi_m(i) = h' \wedge  \pi_m(h) = i'] > 0$. Thus, $C \succ_h C'$ and $\Pr[\pi_m(h) \in C'] >0$, from Claim \ref{claim:no_contracts} this implies that for every hospital $h''$ in the support of the prospect $\pi_m(i)$, $h'' \succeq_i h$. Thus, hospital $h'$ is not in the support of $\pi_m(i)$ which contradicts the assumption $\mu$ is an active contract.

Denote by $S_\pi \in \H \times \D \times \H \times \D$ the set of all active contracts in the allocation $\pi$.
From claim \ref{claim:no_contracts}, if there is an active contract $(h, i; h', i') \in S_{\pi}$ in the final allocation it means that 
$$\Pr[\pi_m(h) = i'] = 0.$$
Otherwise, from Claim \ref{claim:no_contracts}, doctor $i$ must be allocated with probability 1 to hospitals that doctor $i$ prefer over hospital $h$, specifically not to hospital $h'$. However, if doctor $i$ does not have ``room'' for more allocations and is not allocated to hospital $h'$ at round $m$, it would not be allocated to hospital $h'$ in the second part of the algorithm.

Since the total amount of mass that is being allocated in the second part of the algorithm is bounded by $\tau$, we have that 
$$\Pr[\bigvee_{(h, i ;h', i') \in S_\pi}\pi(h) =i'] \le \tau,$$
which implies that
$$
 \Pr[\bigvee_{(h, i ;h', i') \in S_\pi}(\pi(h) =i'\wedge \pi(h') = i)] \le \Pr[\bigvee_{(h, i ;h', i') \in S_\pi}\pi(h) =i'] \le \tau.
$$

\end{proof}

\subsubsection{The Output of \Cref{alg:GS_PSP} is Preference-Informed Individually Fair}

\begin{lemma}
\label{lemma:gs_psp_piif}
The output of Algorithm \ref{alg:GS_PSP} is PIIF.
\end{lemma}

\begin{proof}
We start by showing that at the end of the algorithm's first part, the allocation is envy-free for each cluster.
Let $i,j \in D$ be doctors in the same cluster $C \in \mathcal{C}$ and let $h_1, ..., h_n$ be the hospitals in $\H$ ordered by the preferences of doctor $i$.
Let us denote the total number of rounds in the run of the algorithm as $m$, and the allocation at the end of the $m$-th round as $\pi_m$.

Fix a rank $r \in [n]$.
From Claim \ref{claim:PS_SD}, $\Pr_{o\sim\pi_m(i)}[r_i^{-1}(o) \le r] = \max\{t_{h_1, m}, ..., t_{h_r, m}\}$.

Assume for contradiction that $$\Pr_{o\sim\pi_m(j)}[r_i^{-1}(o) \le r] > \Pr_{o\sim\pi_m(i)}[r_i^{-1}(o) \le r] = \max\{t_{h_1, m}, ..., t_{h_r, m}\}.$$

This means that doctor $j$ ate some hospital $h \in \{h_1, ..., h_r\}$ at time $t > \max\{t_{h_1, m}, ..., t_{h_r, m}\} \ge t_{h,m}$. However, by the definition of $t_{h,m}$, at any time later than $t_{h,m}$, hospital $h$ was not available, which contradicts the assumption.

Hence, 
\begin{align*}
    &\forall r \in [n]:\Pr_{o\sim\pi_m(i)}[r_i^{-1}(o) \le r] \ge \Pr_{o\sim\pi_m(j)}[r_i^{-1}(o) \le r]
\end{align*}
so $\pi(i) \succeq_{i} \pi(j)$.

We assign the same probability to each doctor in each cluster at the second part, so the allocation is still envy-free.
Since for every cluster, the allocation is envy-free, for the set of all doctors $\D$, the allocation $\pi$ is PIIF.
\end{proof}

Now we are set to prove \Cref{thm:gs_psp_piif_stable}.
\begin{proof}[Proof of \Cref{thm:gs_psp_piif_stable}]
The proof is immediate from Lemma \ref{lemma:gs_psp_piif} and Lemma \ref{lemma:gs_psp_contract_stability}.
\end{proof}

\subsubsection{\Cref{alg:GS_PSP} Converges}

\begin{theorem}
\label{thm:gs_psp_stop}
Algorithm \ref{alg:GS_PSP} terminates after at most $O(n\cdot |\C|/\tau)$ rounds.
\end{theorem}

\begin{proof}
For a doctor $h\in \H$, denote by $p_{h,k}$ hospital $h$'s free probability mass in round $k$. Let us denote by $p_k$ the total free probability at round $k$, $p_k = \sum_{h\in \H}p_{h,k}$. 
For any round $k$ before the algorithm terminates, $p_k > \tau$. The total number of rounds is denoted by $k^*$.

For hospital $h \in \H$, cluster $C \in \C$, round $k \le k^*$, we denote by $a_{h,C,k}$ the available probability mass that hospital $h$ can propose to cluster $C$ in round $k$. At round $1$, for every hospital $h \in \H$ and cluster $C \in \C$, $a_{h,C,1} = 1$.

Let $h \in \H$ be a hospital and $C \in \C$ be a cluster and fix some round $k < k^*$. We observe the following:

\begin{enumerate}
    \item Assume $p_{h,k} = 0$. Then, $a_{h, C,k+1} = a_{h,C, k}$.
    \item Assume $C$ rejected some probability mass from hospital $h$ at some round $k' < k$. Then $a_{h,C,k''} = 0$ for every $k'' \ge k'$ and specifically $a_{h,C,k} = 0$.
    \item Assume $C$ is the cluster that hospital $h$ proposed to in round $k$. If cluster $C$ accepted all the probability mass that hospital $h$ proposed then $a_{h,C,k+1} = a_{h,C,k} - p_{h,k}$, otherwise $a_{h,C,k+1} = 0$. Since hospital $h$ proposed to cluster $C$, we know that $a_{h,C,k} \ge p_{h,k}$. Thus, either way in this case $a_{h,C,k+1} \le a_{h,C,k} - p_{h,k}$.
\end{enumerate}

Thus, for every hospital $h \in \H$, cluster $C \in \C$ and round $k < k^*$ we have that $0 \le a_{h,C,k+1} \le a_{h,C,k}$, and for at least one cluster $C^*$, $0 \le a_{h,C^*,k+1} \le a_{h,C^*,k} - p_{h,k}$.

Summing over all clusters and hospitals, we get that
$$
    0 \le \sum_{c\in \C}\sum_{h \in \H} a_{h,C,k} \le
    \sum_{C\in \C}\sum_{h \in \H} a_{h,C,k-1} - \sum_{h\in \H}p_{h,k-1} =  
    \sum_{C\in \C}\sum_{h \in \H} a_{h,C,k-1} - p_{k-1}  \overset{p_{k-1} > \tau}{<}$$$$
    \sum_{C\in \C}\sum_{h \in \H} a_{h,C,k-1} - \tau < ...< n\cdot |\C| - k \cdot \tau.
$$

Which implies that $k <  n\cdot |\C| / \tau$ for every round $k$ in the run of the algorithm. This implies that $k^* = O(n\cdot |\C| / \tau)$.
\end{proof}

Next, we show that if we wish to guarantee exact contract stability (rather than $\tau$-contract stability for $\tau>0$), Algorithm $\ref{alg:GS_PSP}$ might not converge.

\begin{claim}
For $\tau = 0$, the number of iterations in Algorithm \ref{alg:GS_PSP} might be unbounded.
\end{claim}

\begin{proof}
Consider the following example: 
Let $\D = \set{i_1, i_2, j, k}$ be the set of doctors, $\H = \set{A, B, C, D}$ be the set of hospitals. The doctors $i_1, i_2$ are in the same cluster $i$, and doctors $j$ and $k$ are in a different clusters $j$ and $k$ respectively.

The preference of the doctors are: 
$$A \succeq_{i_1} B \succeq_{i_1} C \succeq_{i_1} D, \quad B \succeq_{i_2} A \succeq_{i_2} D \succeq_{i_2} C, \quad C \succeq_j B \succeq_j A \succeq_j D, \quad D \succeq_k A \succeq_k B \succeq_k C,$$
and the preferences of the hospitals are
$$i \succeq_A j \succeq_A k, \quad j \succeq_B i \succeq_B k, \quad i \succeq_C j \succeq_C k, \quad i \succeq_D k \succeq_D j.$$

See the following simulation of the algorithm run:

\textbf{Round 1:} Hospital $A$ proposes 1 to doctor $i$, hospital $B$ proposes 1 to doctor $j$, hospital $C$ proposes 1 to doctor $i$, hospital $D$ proposes 1 to doctor $i$.
\begin{align*}
 &\pi_2(i_1) = \begin{cases}
    A,&\textit{ w.p. }1/2\\
    C,&\textit{ w.p. }1/2
\end{cases}, \quad
\pi_1(i_2) = \begin{cases}
    A,&\textit{ w.p. }1/2\\ 
    D,&\textit{ w.p. }1/2
\end{cases}, \\
&\pi_1(j) = \begin{cases}
    B,&\textit{ w.p. }1
\end{cases}, \quad
\pi_1(k) = \begin{cases}
    unallocated,&\textit{ w.p. } 1
\end{cases}.
\end{align*}
Probability mass $1/2$of hospital $C$ rejected, probability mass $1/2$ of hospital $D$ rejected.

\textbf{Round 2:} Hospital $C$ proposes $1/2$ to doctor $j$, hospital $D$ proposes $1/2$ to doctor $k$.
\begin{align*}
    &\pi_2(i_1) = \begin{cases}
        A,&\textit{ w.p. }1/2 \\
        C,&\textit{ w.p. }1/2
    \end{cases}, \quad
    \pi_2(i_2) = \begin{cases}
        A,&\textit{ w.p. }1/2\\
        D,&\textit{ w.p. }1/2
    \end{cases}, \\
    &\pi_2(j) = \begin{cases}
        B,&\textit{ w.p. }1/2 \\
        C,&\textit{ w.p. }1/2
    \end{cases}, \quad
    \pi_2(k) = \begin{cases}
        D,&\textit{ w.p. }1/2 \\
        unallocated,&\textit{ w.p. }1/2
    \end{cases}.
\end{align*}
Probability mass $1/2$ of hospital $B$ rejected.

\textbf{Round 3:} Hospital $B$ proposes $1/2$ to doctor $i$.
\begin{align*}
    &\pi_3(i_1) = \begin{cases}
        A,&\textit{ w.p. }3/4 \\
        C,&\textit{ w.p. }1/4  
    \end{cases},\quad
    \pi_3(i_2) = \begin{cases}
        A,&\textit{ w.p. }1/4 \\
        B,&\textit{ w.p. }1/2 \\
        D,&\textit{ w.p. }1/4
    \end{cases}, \\
    &\pi_3(j) = \begin{cases}
        B,&\textit{ w.p. }1/2 \\
        C,&\textit{ w.p. } 1/2
    \end{cases},\quad
    \pi_3(k) = \begin{cases}
        D,&\textit{ w.p. }1/2 \\
        unallocated,& \textit{w.p. }1/2
    \end{cases}.
\end{align*}

Probability mass $1/4$ of hospital $C$ rejected, probability mass $1/4$ of hospital $D$ rejected.

Round 3 is identical to round 1 with half the probability mass proposed, so round 4 will be identical to round 2 with half the probability mass and so on. Thus, the algorithm will never stop. The ``free'' probability mass of the hospitals gets smaller by 1/2 every 2 rounds.
\end{proof}

\subsection{Doctor-Proposing Algorithm}
Next, we present our second algorithm for PIIF and contract stable allocations given strict individually fair hospital preferences. This algorithm is also a generalization of the Gale Shapley algorithm (Figure \ref{alg:Gale_Shapley}). However, in this version, the doctors are the proposing set.

\subsubsection{Fair Propose-and-Reject -- Doctors-First}
At the beginning of the algorithm, each doctor has a free probability mass of 1. We denote the free probability mass of doctor $i \in \D$ as $p_i$. Sometimes we refer to the free probability mass of doctor $i$ in round $k$ as $p_{i,k}$.

Each round, every doctor proposes its free probability mass, $p_i$, to the most preferred hospital that did not reject it yet. As an analogy, it might be helpful to think of proposing to the next hospital on the list in the Gale Shapley algorithm.

For every hospital $h \in \H$, after the doctors' proposals, the proposed probability mass vector $p_h$ (sometimes denoted as $p_{h,k}$ as the vector in round $k$) contains for each doctor $i \in \D$, the probability mass doctor $i$ proposed in the current round plus the probability to be matched to doctor $i$ in hospital $h$'s prospect in the previous round.
Hospital $h$ runs \Cref{alg:water_algorithm} over this vector, $p_h$, to find its prospect for the current round.

Any probability mass not used for the allocation is rejected and becomes free again. As an analogy, it may be helpful to think of a hospital rejecting their current match in the Gale Shapley algorithm if a more preferred doctor proposes to it.

The algorithm terminates when there is no more free probability mass, i.e., there is a full allocation, or the free probability mass is very small.

\begin{algorithm}[H]
    \caption{Fair Propose-and-Reject -- Doctors-First}
    \begin{algorithmic}[1]
    \label{alg:GS_WA}
        \STATE $\forall i \in \D: p_i \leftarrow 1$, $P \in [0,1]^{n \times n} \leftarrow \textbf{0}$, $k \leftarrow 1$, $\forall h \in \H : p_h \in [0,1]^n \leftarrow \mathbf{0}$.
        \WHILE{$\sum_{ i \in \D}p_i > \tau$}
            \FORALL{$i\in \D: p_i > 0$} 
                \STATE $h \in \H$ is the highest ranked hospital by $i$. 
                \STATE $p_h [i] \leftarrow p_h[i] + p_i$.
                \hfill\COMMENT{$i$ proposes its remaining probability mass (like GS)}
                \STATE $p_i \leftarrow 0$.
            \ENDFOR
            \FORALL{$h \in \H$}
                \STATE Find an allocation $\pi_k(h)$ by running \Cref{alg:water_algorithm} over $p_h$.
                \FORALL{$i \in \D$}
                    \STATE $P[i, h] \leftarrow \Pr[\pi_k(h) = i]$.
                    \hfill\COMMENT{Update the output probability matrix.}
                    \STATE $s \leftarrow p_h[i] - \pi_k(h)[i]$
                    \hfill\COMMENT{Calculate the unallocated probability mass.}
                    \IF{$s > 0$}
                        \STATE $p_i \leftarrow p_i + s$.
                        \hfill\COMMENT{$h$ rejects the amount $s$ of $i$. (like GS)}
                        \STATE $p_h[i] \leftarrow p_h[i] - s$.
                        \STATE Remove $h$ from $i$'s preference list.
                    \ENDIF
                \ENDFOR
            \ENDFOR
            \STATE $k \leftarrow k + 1$.
        \ENDWHILE
        \label{partial_ef}
        \STATE \label{step:alloc_prob_mass} \textbf{Allocating the remaining probability mass:} 
        
        Split the remaining probability mass of the doctors by the hospitals such that $P$ is a doubly stochastic matrix.
        \STATE Find an allocation $\pi$ by running Birkhoff von Neumann algorithm over $P$.
        \RETURN $\pi$.
    \end{algorithmic}
\end{algorithm}

\begin{remark}
In step \ref{step:alloc_prob_mass}, we say that we allocate the remaining probability mass to form a doubly stochastic matrix.
One possible way to do it is by the following algorithm, which greedily distributes the free probability mass in a way that results in a doubly stochastic matrix:
\begin{algorithm}[H]
    \caption{Allocate free probability mass}
    \label{alg:alloc_free_prob_mass}
    \hspace*{\algorithmicindent} \textbf{Input} $\forall d \in \D: p_d \in [0,1]$, $\forall h \in \H: p_h \in [0,1]^{n}$, $P \in [0,1]^{n \times n}$.
    
    \begin{algorithmic}[1]
        \FORALL{$i \in \D$}
            \FORALL{$h \in \H$}
                \STATE $p \leftarrow \min\set{p_i, 1-\sum_{j\in\D}p_h[j]}$.
                \STATE $P[i,h] \leftarrow P[i,h] + p$.
                \STATE $p_i \leftarrow p_i - p$
                \STATE $p_h[i] \leftarrow p_h[i] + p$.
            \ENDFOR
        \ENDFOR
        \RETURN $P$.
    \end{algorithmic}
\end{algorithm}
Since the total amount of probability mass is $n$, and the total amount we can allocate with the hospitals is $n$, after running \Cref{alg:alloc_free_prob_mass}, all the probability mass is allocated. By way of choosing how much mass to allocate at each step, we make sure that no doctor or hospital is allocated to more than one unit of mass. Thus, the matrix $P$ is doubly stochastic.
\end{remark}

We show that \Cref{alg:GS_WA} is $2\tau$-PIIF, $\tau$-contract stable and terminates after $O(n^2 / \tau)$ rounds.

\begin{theorem}
\label{thm:gs_wa_fair_stable}
Given a set of doctors $\D$, a set of hospitals $\H$, a proto-metric $d : \D \times \D \rightarrow \set{0,1}$, deterministic doctor preferences $\mathcal{P}_{\D} = \{r_i\}_{i\in \D}$, probabilistic and strictly individually fair hospital preferences $\mathcal{P}_\H = \{r_h\}_{h\in \H}$ and a constant $\tau \in (0,n]$, \Cref{alg:GS_WA} outputs an allocation $\pi$ such that $\pi$ is $2\tau$-preference-informed individually fair and $\tau$-contract stable.
\end{theorem}

In order to prove \Cref{thm:gs_wa_fair_stable}, we first prove that the output of \Cref{alg:GS_WA} satisfies each of the properties separately. This is done in Lemmas \ref{lemma:GS_WA_tau_piff} and \ref{lemma:gs_wa_contract_stable}. 

\subsubsection{The Output of Algorithm \ref{alg:GS_WA} is $2\tau$-PIIF}
\begin{lemma}
\label{lemma:GS_WA_tau_piff}
The output of \Cref{alg:GS_WA} is $2\tau$-PIIF.
\end{lemma}

\paragraph{Proof outline.} The intuition behind the fairness of \Cref{alg:GS_WA} is that for every cluster $C \in \C$ and hospital $h \in \H$ at each round, all the doctors in cluster $C$ are assigned with the same probability to hospital $h$ unless they propose less probability mass than other doctors in the cluster. We assume that they proposed less probability mass since the rest of their probability mass is assigned to more preferred hospitals.

The main concern is the scenario where for cluster $C \in \C$ and doctor $i \in C$, in some early round, some probability proposed by doctor $i$ was rejected by hospital $h$ and "percolated" down to less preferred hospitals. If in some later round, doctor $i$ would appear to propose less probability mass than other doctors in cluster $C$, doctor $i$ would be allocated with less probability mass. However, doctor $i$'s probability mass is allocated with less preferred hospitals. This kind of scenario will result in an unfair solution. 
In Claim \ref{claim:max_prob_for_rejected}, we show that this scenario is not possible and that once the probability mass of doctor $i$ got rejected by hospital $h$, no other doctor in the cluster would be assigned with more probability than doctor $i$.

In claim \ref{claim:doc_alloc_better_in_support}, we show that this implies that if a doctor is assigned with less probability mass by $h$ compared to other doctors in its cluster, then the rest of this doctor's assigned probability mass is assigned to more preferred hospitals.

Claim \ref{claim:pi_k_tau_piif} and 
lemma \ref{lemma:GS_WA_tau_piff} conclude the proof by showing how to construct an alternative allocation for every two similar doctors that satisfy the fairness condition.

For every $h \in \H$, denote by $\pi_k(h)$ the prospect hospital $h$ finds in round $k$.

\begin{claim}
\label{claim:max_prob_for_rejected}
Let $h \in \H$ be a hospital and $i \in \D$ be a doctor in cluster $C \in \mathcal{C}$. For every round $k$, if at round $k$, hospital $h$ rejected probability mass $p > 0$ of doctor $i$. Then for every $k' \ge k$ and for every doctor $j \in C$, $\pi_k(h)[i] \ge \pi_{k'}(h)[j]$.
\end{claim}

\begin{proof}
For a round $m$, we will use the following definitions:
\begin{itemize}
    \item The maximal probability mass assigned to a doctor in cluster $C$ in round $m$ is denoted by $b_{C,m}$, $$b_{C,m} = \max_{j \in C}\pi_m(h)[j].$$
    \item The total amount of probability mass in the prospect $\pi_m$ that is assigned to doctors in cluster $C$ in round $m$ is denoted by $M_{C,m}$, $$M_{C,m} = \sum_{j \in C}\pi_m(h)[j].$$
    \item The set of clusters that are ranked higher than cluster $C$ by hospital $h$ is denoted by $\mathcal{S}_h^{>C}$, and the set of clusters that are ranked lower than cluster $C$ by hospital $h$ is denoted by $\mathcal{S}_h^{<C}$,
    $$
    \mathcal{S}_h^{>C} = \{C' \in \mathcal{C}: C' \succ_h C\}, \quad \mathcal{S}_h^{<C} = \{C' \in \mathcal{C}: C \succ_h C'\}.
    $$
    \item The set of doctors in cluster $C$ that are assigned with probability mass $b_{C,m}$ in round $m$ is denoted by $C_{m}^*$, $$C_{m}^* = \{j \in C: \pi_m(h)[j] = b_{C,m}\}.$$
\end{itemize}

In order to prove the claim, consider some round $k$, and a doctor $i$ who had some probability mass rejected, we proceed as follows:
\begin{enumerate}
    \item 
    \label{base_claim}
    In round $k$, 
    $\Pr[\pi_k(h) \in \mathcal{S}_h^{>C}] = 1 - M_{C,k}$.
    
    Since hospital $h$ rejected some probability mass at round $k$ we know that $$\sum_{j \in \D} \pi_k(h)[j] = 1.$$
    Thus, probability mass $1 - M_{C,k}$ is assigned to clusters in $\mathcal{C} \backslash \{C\}$. By \Cref{alg:water_algorithm}, we know that a cluster in $\mathcal{S}_h^{<C}$ is assigned to hospital $h$ only if all the probability mass that cluster $C$ proposed to hospital $h$ was assigned. In this case, some probability mass that was proposed by cluster $C$ was rejected. Thus, probability mass $1 - M_{C,k}$ is assigned to clusters in $\mathcal{S}_h^{>C}$.
    
    \item 
    \label{step_claim}
    If in round $m$, 
    $\Pr[\pi_m(h) \in \mathcal{S}_h^{>C}] = 1 - M_{C,m}$, then in round $m+1$, $$\Pr[\pi_{m+1}(h) \in \mathcal{S}_h^{>C}] = 1 - M_{C,m+1} \ge 1 - M_{C,m}.$$
    
    In step $m$, probability mass $1-M_{C,m}$ was assigned to clusters in $\mathcal{S}_h^{>C}$. Thus, this probability was proposed by the same clusters in round $m+1$ too. If for every cluster $C' \in \mathcal{S}_h^>C$, $\Pr[\pi_{m+1}(h) \in C'] \ge \Pr[\pi_{m}(h) \in C']$, then we are done. Otherwise, for some cluster $C' \in \mathcal{S}_h^{>C}$, $\Pr[\pi_{m+1}(h) \in C'] < \Pr[\pi_{m}(h) \in C']$. By \Cref{alg:water_algorithm}, this probability mass was assigned to clusters in $\mathcal{S}_h^{>C'}$ and since $C' \in \mathcal{S}_h^{>C}$, $\mathcal{S}_h^{>C'} \subset \mathcal{S}_h^{>C}$.
    Thus, at least $1 - M_{C,m}$ probability mass is assigned to clusters in $\mathcal{S}_h^{>C}$, i.e. 
    $$\Pr[\pi_{m+1}(h) \in \mathcal{S}_h^{>C}] \ge 1 - M_{C,m}.$$
    
    Since $M_{C,m}$ probability mass of cluster $C$ is available in round $m+1$, no probability mass will be assigned to clusters in $\mathcal{S}_h^{<C}$ so
    $$\Pr[\pi_{m+1}(h) \in \mathcal{S}_h^{>C}] = 1 - M_{C,m+1}.$$
    
    \item 
    \label{m_c_i_plus_one_ge_m_c_i}
    For every $m \ge k$, $M_{C,m+1} \le M_{C,m}$. This is implied from \ref{base_claim} and \ref{step_claim}.
    
    \item For every $m \ge k$, $b_{C,m+1} \le b_{C,m}$, which implies that $b_{C,m} \le b_{C,k}$.
    
    Assume for contradiction that for some round $m$, $b_{C,m+1} > b_{C,m}$.
    By \Cref{alg:GS_WA}, for every doctor $j \in \D$, $p_{h,m+1}[j] \ge \pi_m(h)[j]$, i.e., every doctor will propose to hospital $h$ in round $m+1$ at least the probability mass that was assigned to hospital $h$ in round $m$. 
    For every doctor $j \in C_{m}^*$, $\pi_m(h)[j] = b_{C,m}$, so $p_{h,m+1}[j] \ge b_{C,m}$. Since $b_{C,m} < b_{C,m+1}$, doctor $j$ will be assigned with at least $b_{C,m}$ in round $m+1$, i.e. $$\forall j \in C_{m}^*: \pi_{m+1}(h)[j]\ge b_{C,m} = \pi_{m}(h)[j].$$
    For every doctor $j \in C \backslash C_{m}^*$, no probability mass proposed by this doctor was rejected in round $m$. Thus, doctor $j$ will propose at least the same as in the previous round, i.e., $p_{h,m+1}[j] \ge p_{h,m}[j] = \pi_{m}(h)[j]$. Since $p_{h,m}[j] < b_{C,m} < b_{C,m+1}$ we know that $$\forall j \in \C \backslash \C^*\pi_{m+1}(h)[j]\ge p_{h,m}[j] = \pi_{m}(h)[j].$$
    We also know that there exists some doctor $j^* \in C$, such that $$\pi_m(h)[j^*] = b_{C,m+1} > b_{C,m} \ge \pi_{m}(h)[j^*].$$
    
    Thus, for every doctor $j \in C$, $\pi_{m+1}(h)[j]\ge \pi_{m}(h)[j]$ and for doctor $j^*$, $\pi_{m+1}(h)[j] > \pi_{m}(h)[j]$. This implies that $$M_{C,m+1} = \sum_{j \in C} \pi_{m+1}(h)[j] > \sum_{j \in C} \pi_{m}(h)[j] =  M_{C,m},$$ which is a contradiction to \ref{m_c_i_plus_one_ge_m_c_i}.
    
    \item In round $k$, $\pi_k(h)[i] = b_{C,k}$. 
    
    By \Cref{alg:water_algorithm}, doctor $i$ is assigned with
    $\pi_m(h)[i] = \min\{b_{C,m}, p_{h,m}[i]\}$.
    Since some of $m$'s probability mass was rejected we know that $\pi_m(h)[i] < p_{h,m}[i]$ which implies that $\pi_m(h)[i] = b_{C,m}$.
\end{enumerate}

Thus, in every round $k' > k$, we have that for every doctor $j \in C$, $$\pi_{k'}(h)[j] \le b_{C,k'} \le b_{C,k} = \pi_{k}(h)[i].$$

\end{proof}

\begin{claim}
\label{claim:doc_alloc_better_in_support}
Let $h \in \H$ be a hospital and $i,j \in \D$ be doctors in cluster $C \in \mathcal{C}$. For every round $k$, if $\pi_k(h)[j] > \pi_k(h)[i]$, then for every $h' \in \H$ in the support of $\pi_k(i)$, $h' \succ_{i} h$.
\end{claim}

\begin{proof}
In \Cref{alg:GS_WA}, doctor $i$ will propose to hospitals less preferred than hospital $h$, only after being rejected by hospital $h$. Since $\pi_k(h)[j] > \pi_k(h)[i]$, we know from Claim \ref{claim:max_prob_for_rejected} that in every round $k' \le k$, no probability mass proposed by doctor $i$, was rejected by hospital $h$. Thus, doctor $i$ proposed only to hospitals ranked higher than hospital $h$ and the claim follows.
\end{proof}

\begin{claim}
\label{claim:pi_k_tau_piif}
Let $k$ be the number of rounds in the run of \Cref{alg:GS_WA} and $\pi_k$ be the allocation at the end of the $k$-th round, then $\pi_k$ is $\tau$-PIIF.
\end{claim}

\begin{proof}
Let $i \in \D$ be a doctor in cluster $C \in \mathcal{C}$
and let $j \in \D$ be a doctor too. If $j \not\in C$ then $d(i,j) = 1$ and $p^{i;j} = \pi_k(i)$ satisfies the conditions.

Otherwise, $j \in C$ and $d(i,j) = 0$. Let's define 
$$
p^{i; j} = \begin{cases}
\pi_k(j), & \textit{w.p. } 1-\tau \\
unallocated, & \textit{w.p. } \tau
\end{cases}.
$$

Then, $D(p^{i; j}, \pi_k(j)) \le \tau$, and it is left to show that $\pi_k(i) \succeq_{i} p^{i; j}$. We need to show that
$$
\forall r \in [n]: \Pr_{o\sim \pi_k(i)}[r_i^{-1}(o) \le r] \ge \Pr_{o\sim p^{i;j}}[r_i^{-1}(o) \le r].
$$

Denote by $h_r$ the hospital that is ranked $r$-th by doctor $i$, i.e. $r_i(r) = h_r$.
Fix a rank $r \in [n]$.

If for every rank $r' \le r$, $\Pr[\pi_k(i) = h_{r'}] \ge \Pr[p^{i;j} = h_{r'}]$ then $$
\Pr_{o\sim \pi_k(i)}[r_i^{-1}(o) \le r] = \sum_{r'\le r}\Pr[\pi_k(i) = h_{r'}] \ge \sum_{r' \le r}\Pr[p^{i;j} = h_{r'}] = \Pr_{o\sim p^{i;j}}[r_i^{-1}(o) \le r].
$$
Otherwise, there exists some rank $r' \le r$ such that 
$$\Pr[\pi_k(i) = h_{r'}] < \Pr[p^{i;j} = h_{r'}] = (1-\tau) \Pr[\pi_k(j) = h_{r'}] \le \Pr[\pi_k(j) = h_{r'}].$$

From Claim \ref{claim:doc_alloc_better_in_support}, this implies that for every hospital $h'$ in the support of $\pi_k(i)$, $h' \succ_{i} h_{r'}$.
Since at the last round the free probability mass is at most $\tau$, this means that 
$$1 - \tau \le \Pr_{o\sim \pi_k(i)}[r_i^{-1}(o) \le r']$$
and since $r' \le r$
$$1 - \tau \le \Pr_{o\sim \pi_k(i)}[r_i^{-1}(o) \le r].$$

On the other hand, $\Pr_{o\sim p^{i;j}}[r_i^{-1}(o) \le r] \le \Pr[p^{i;j}\text{ allocated}] \le 1 - \tau$.

This implies that, 
$$
\Pr_{o\sim \pi_k(i)}[r_i^{-1}(o) \le r] \ge \Pr_{o\sim p^{i;j}}[r_i^{-1}(o) \le r]. 
$$

This means that $\pi_k(i) \succ_{i} p^{i;j}$ and the claim follows.
\end{proof}

Now we are set to prove Lemma \ref{lemma:GS_WA_tau_piff}.

\begin{proof}[Proof of Lemma \ref{lemma:GS_WA_tau_piff}]
Let $k$ be the number of rounds in the run of \Cref{alg:GS_WA}. From Claim \ref{claim:pi_k_tau_piif}, $\pi_k$ is $\tau$-PIIF.
Let $i,j \in \D$ be doctors such that $d(i,j) = 0$, there exists $p^{i;j}$ such that $D(p^{i;j}, \pi_k(j))\le \tau$ and $\pi_k(i) \succ_{i} p^{i;j}$.

Note that the allocation $\pi_k$ is not the final allocation since some probability mass might not be assigned yet. We still need to allocate this unassigned probability mass, and we have no guarantee over the assignment of the free probability mass.

Let $\pi$ be the final allocation obtained by the algorithm.
Since the remaining free probability mass at the end of round $k$ is at most $\tau$ we know that $D(\pi(j), \pi_k(j)) \le \tau$. Thus, from the triangle inequality, 
$$D(p^{i;j}, \pi(j))\le D(p^{i;j}, \pi_k(j)) +D(\pi_k(j), \pi(j))  \le 2\tau.$$

Since in the second part of the algorithm, only unallocated mass is being allocated to hospitals, we know that 
$$
\forall h \in H: \Pr[\pi(i) = h] \ge \Pr[\pi_k(i) = h].
$$

Denote by $h_r$ the hospital that is ranked $r$-th by dcotor $i$. Fix a rank $r \in [n]$.
$$
\Pr_{o\sim \pi(i)}[r_i^{-1}(o) \le r] = \sum_{r' \le r}\Pr[\pi(i) = h_{r'}] \ge \sum_{r' \le r}\Pr[\pi_k(i) = h_{r'}] = \Pr_{o\sim \pi_k(i)}[r_i^{-1}(o) \le r].
$$

Thus, $\pi(i) \succ_{i} \pi_k(i) \succ_{i} p^{i;j}$ and the claim follows.
\end{proof}
\subsubsection{Contract Stability}
Next, we prove that \Cref{alg:GS_WA} is $\tau$-contract stable.

\begin{lemma}
\label{lemma:gs_wa_contract_stable}
The outcome of \Cref{alg:GS_WA} is $\tau$-contract stable.
\end{lemma}
 
 In order to prove Lemma \ref{lemma:gs_wa_contract_stable}, we first prove the following claim:
 
\begin{claim}
\label{claim:doc_first_alg_SD_hospitals}
Let $h \in \H$ be a hospital and $k > 1$ be a round, $\pi_k(h) \succ_h \pi_{k-1}(h)$.
\end{claim}

\begin{proof}
Denote by $C_r$ the cluster that is ranked $r$-th by hospital $h$.

Fix a rank $r \in [|\mathcal{C}|]$.
If for every rank $r' \le r$, $\Pr[\pi_k(h) \in C_{r'}] \ge \Pr[\pi_{k-1}(h) \in C_{r'}]$ then 
$$
\Pr_{o\sim \pi_k(h)}[r_h^{-1}(o) \le r] \ge \Pr_{o\sim \pi_{k-1}(h)}[r_h^{-1}(o) \le r].
$$

Otherwise, there exists some rank $r' \le r$ such that $\Pr[\pi_k(h) \in C_{r'}] < \Pr[\pi_{k-1}(h) \in C_{r'}]$.

By \Cref{alg:water_algorithm}, if probability mass $\Pr[\pi_{k-1}(h) \in C_{r'}]$ is assigned to doctors in $C_{r'}$ in round $k-1$, then these doctors will propose this probability mass to $h$ in round $k$. Thus, if $\Pr[\pi_k(h) \in C_{r'}] < \Pr[\pi_{k-1}(h) \in C_{r'}]$,
then $h$ rejected probability mass of $C_r'$, and this can happen only if $\Pr_{o\sim \pi_k(h)}[r_h^{-1}(o) \le r'] = 1$. Thus,
$$
\Pr_{o\sim \pi_k(h)}[r_h^{-1}(o) \le r] \ge \Pr_{o\sim \pi_k(h)}[r_h^{-1}(o) \le r'] = 1 \ge \Pr_{o\sim \pi_k(h)}[r_h^{-1}(o) \le r].
$$

Thus, for every rank $r \in [|\mathcal{C}|],$
$$
\Pr_{o\sim \pi_k(h)}[r_h^{-1}(o) \le r] \ge \Pr_{o\sim \pi_{k-1}(h)}[r_h^{-1}(o) \le r]
$$
and $\pi_k(h) \succ_h \pi_{k-1}(h)$.
\end{proof} 

\begin{proof}[Proof of Lemma \ref{lemma:gs_wa_contract_stable}]
Let $h, h' \in \H$ be hospitals, $i, i' \in \D$ be doctors and $C, C' \in \mathcal{C}$ be clusters such that $i \in C, i' \in C'$ and $i \succ_h i', h \succ_i h'$.

Fix a round $k$. There are two options:
\begin{enumerate}
    \item No probability mass from cluster $C$ was rejected in any round $k' \le k$. Then, by \Cref{alg:GS_WA}, this means that for every hospital $h'' \in \H$ in the support of the prospect $\pi_k(i)$ satisfies that $h'' \succ_i h$. Since for any hospital $x \prec_d h$, doctor $i$ would have proposed to hospital $x$ only after being rejected from hospital $h$ and we assumed that this did not happen.
    Thus,
    $$
    \Pr[\pi_k(i) = h'] = 0 \Rightarrow \Pr[\pi_k(i) = h' \wedge \pi_k(h) = i'] = 0
    $$
    \item For some doctor $j \in C$, some probability mass of doctor $j$ was rejected in round $k' \le k$.
    Thus, by \Cref{alg:water_algorithm}, in round $k'$, $\Pr_{o\sim \pi_{k'}(h)}[r_h^{-1}(o) \le r_h^{-1}(C)] = 1$. From Claim \ref{claim:doc_first_alg_SD_hospitals}, $\pi_k(h) \succ_h \pi_{k'}(h)$, so $\Pr_{o\sim \pi_{k}(h)}[r_h^{-1}(o) \le r_h^{-1}(C)] = 1$ which implies that
    $$
    \Pr[\pi_k(h) = i'] = 0 \Rightarrow \Pr[\pi_k(i) = h'\wedge \pi_k(h) = i'] = 0
    $$
\end{enumerate}
Thus, if $\Pr[\pi(i) = h'\wedge \pi(h) = i'] > 0$ and the contract $(h, i ;h', i')$ is active, it means that this probability mass was allocated to $i$ or $h$ (or both) at the final allocation step. 

Denote by $S_\pi \in \D \times \H \times \D \times \H$ the set of all active contracts in the allocation $\pi$.
Since the total amount of mass that is being allocated in the second part of the algorithm is bounded by $\tau$, we have that 
$$
\Pr[\bigvee_{(h, i ;h', i') \in S_\pi}(\pi(h) =i'\wedge \pi(h') = i)] \le \tau.
$$
\end{proof}

\subsubsection{Algorithm \ref{alg:GS_WA} Converges}

\begin{theorem}
\label{thm:gs_wa_stop}
Algorithm \ref{alg:GS_WA} terminates after $O(n^2/\tau)$ rounds.
\end{theorem}

\begin{proof}
For a doctor $i\in \D$, denote by $p_{i,k}$ doctor $i$'s free probability in round $k$. Denote by $p_k$ the total free probability at round $k$, i.e. $p_k = \sum_{i\in \D}p_{i,k}$. 
For any round $k$, before the algorithm terminates, $p_k > \tau$.

Let $k^*$ be the total number of rounds in the algorithm.
For doctor $i \in \D$, hospital $h \in \H$ and round $k \le k^*$, denote by $a_{h,i,k}$ the available probability mass that doctor $i$ can propose to hospital $h$ in round $k$. At round $1$, for every doctor $i \in \D$ and hospital $h \in \H$, $a_{i,h,1} = 1$.

An important observation for this proof is that for every hospital $h \in \H$, doctor $i \in \D$ and round $k\ge1$ we have that $0 \le a_{i,h,k+1} \le a_{i,h,k}$, and for at least one hospital $h^*$, $0 \le a_{i,h^*,k+1} \le a_{i,h^*,k} - p_{i,k}$. We show this observation by dividing into cases:
\begin{enumerate}
    \item $p_{i,k} = 0$. Then, $a_{i, h,k+1} = a_{i,h, k}$.
    \item Hospital $h$ rejected some probability mass from doctor $i$ at some round $k' \le k$. Then $a_{i,h,k''} = 0$ for every $k'' \ge k'$ and specifically $a_{i,h,k} = a_{i,h,k+1} = 0$.
    \item Hospital $h$ is the hospitals that doctor $i$ proposed to in round $k$. If hospital $h$ accepted all the probability mass that doctor $i$ proposed then $a_{i,h,k+1} = a_{i,h,k} - p_{i,k}$, otherwise $a_{i,h,k} = 0$. Since doctor $i$ proposed to hospital $h$, we know that $a_{i,h,k} \ge p_{i,k}$. Thus, either way in this case $a_{i,h,k+1} \le a_{i,h,k} - p_{i,k}$.
\end{enumerate}
Combining the above with the fact that while $p_k > \tau$, some doctor proposes to some hospital in round $k$, we get that the observation is true.
Summing over all doctors and hospitals we get that
$$
0 \le \sum_{i\in \D}\sum_{h \in \H} a_{i,h,k} \le 
\sum_{i\in \D}\sum_{h \in \H} a_{i,h,k-1} - \sum_{i\in \D}p_{i,k-1} =  
\sum_{i\in \D}\sum_{h \in \H} a_{i,h,k-1} - p_{k-1} \overset{p_{k-1} > \tau}{<}
$$ 
$$
\sum_{i\in \D}\sum_{h \in \H} a_{i,h,k-1} - \tau < ...< n^2 - k \cdot \tau
$$

Which implies that $k < n^2 / \tau$ for every round in the run of the algorithm. This implies that the number of rounds, $k^*$ is final and smaller than $n^2 / \tau$.
\end{proof}

Next, we show that for this variant too, if we require exact contract stability, Algorithm \ref{alg:GS_WA} might not converge.

\begin{claim}
For $\tau = 0$, the number of iterations in \Cref{alg:GS_WA} might be unbounded.
\end{claim}

\begin{proof}
Consider the following example:
Let $\D = \set{i_1, i_2, j}$ be the set of doctors, $\H = \set{A, B, C}$ be the set of hospitals. The doctors $i_1$ and $i_2$ are in the same cluster $i$, and doctor $j$ is in a different cluster $j$.

The preference of the doctors are: 
$$
A \succeq_{i_1} B \succeq_{i_1} C, \quad A \succeq_{i_2} C \succeq_{i_2} B, \quad C \succeq_j A \succeq_j B,
$$
and the preferences of the hospitals are
$$
j \succeq_A i, \quad i \succeq_B j, \quad i \succeq_C j.
$$

Let's simulate a run of \Cref{alg:GS_WA} over these two sets of individuals.

\textbf{Round 1:} Doctors $i_1$ and $i_2$ proposes 1 to hospital $A$, doctor $j$ proposes $1$ to hospital $C$.

$\pi_1(A) = \set{i_1: 1/2, i_2: 1/2}$, 
$\pi_1(B) = \set{unallocated: 1}$, 
$\pi_1(C) = \set{k :1 }$.

Probability mass $1/2$ of doctor $i_1$ and probability mass $1/2$ of doctor $i_2$ rejected by $A$.

\textbf{Round 2:} Doctor $i_1$ proposes $1/2$ to hospital $B$, doctor $i_2$ proposes $1/2$ to hospital $C$.

$\pi_2(A) = \set{i_1: 1/2, i_2: 1/2}$, 
$\pi_2(B) = \set{i_1: 1/2, unallocated: 1/2}$, 

$\pi_2(C) = \set{i_2:1/2, k :1/2}$.

Probability mass $1/2$ of doctor $k$ rejected by hospital $C$.

\textbf{Round 3:} Doctor $k$ proposes $1/2$ to hospital $A$.

$\pi_3(A) = \set{i_1: 1/4, i_2: 1/4, k: 1/2}$,
$\pi_3(B) = \set{i_1: 1/2, unallocated: 1/2}$, 

$\pi_3(C) = \set{i_2:1/2, k :1/2}$.

Probability mass  $1/4$ of doctor $i_1$ and probability mass $1/4$ of doctor $i_2$ rejected by hospital $A$.

Round 3 is identical to round 1 with half the probability mass. Therefore, this cycle repeats itself indefinitely and the total free probability mass is reduced by a factor of 2 every two rounds.
\end{proof}

\subsection{Comparison Between the Two Algorithms}
\label{sec:optimality}

We present two variants of algorithms with almost the same fairness and stability guarantees. A natural question that comes to mind is what are the differences between these two algorithms and whether they have other guarantees. 
In Claim \ref{claim:algs_different}, we show that these two algorithms do not always output the same allocation.  

\begin{claim}
\label{claim:algs_different}
Algorithms \ref{alg:GS_PSP} and \ref{alg:GS_WA} do not always return the same output.
\end{claim}
\begin{proof}
Consider the following example:
Let $\D = \set{i_1, i_2, j_1, j_2}$ be the set of doctors and $\H = \set{A, B, C, D}$ be the set of hospitals. 
There are two clusters, $i = \set{i_1, i_2}$ and $j = \set{j_1, j_2}$.
The preferences of the doctors are
$$
A \succ_{i_1} B \succ_{i_1} C \succ_{i_1} D, \quad 
B \succ_{i_2} A \succ_{i_2} D \succ_{i_2} C,
$$
$$
C \succ_{j_1} A \succ_{j_1} B \succ_{j_1} D, \quad 
A \succ_{j_2} D \succ_{j_2} B \succ_{j_2} C.
$$
The preferences of the hospitals are
$$
j \succ_A i, \quad
j \succ_B i, \quad
i \succ_C j, \quad
j \succ_D i.
$$
When running \Cref{alg:GS_PSP} we get the output
\begin{center}
    $\pi_1$ = 
    \begin{tabular}{c|c|c|c|c}
         & $A$ & $B$ & $C$ & $D$  \\
         \hline
         $i_1$ & 0 & 1/4 & 3/4 & 0 \\
         \hline
         $i_2$ & 0 & 1/4 & 1/4 & 1/2 \\
         \hline
         $j_1$ & 1/2 & 1/2 & 0 & 0 \\
         \hline
         $j_2$ & 1/2 & 0 & 0 & 1/2
    \end{tabular}
\end{center}

When running \Cref{alg:GS_WA} we get the output
\begin{center}
    $\pi_2$ = 
    \begin{tabular}{c|c|c|c|c}
         & $A$ & $B$ & $C$ & $D$  \\
         \hline
         $i_1$ & 0 & 1/2 & 1/2 & 0 \\
         \hline
         $i_2$ & 0 & 1/2 & 0 & 1/2 \\
         \hline
         $j_1$ & 1/2 & 0 & 1/2 & 0 \\
         \hline
         $j_2$ & 1/2 & 0 & 0 & 1/2
    \end{tabular}
\end{center}
\end{proof}

Algorithms $\ref{alg:GS_PSP}$ and $\ref{alg:GS_WA}$ are a generalization of the Gale Shapley algorithm (Figure \ref{alg:Gale_Shapley}) where in \Cref{alg:GS_PSP} the hospitals are the proposing set, and in \Cref{alg:GS_WA} the doctors are the proposing set. 
Gale and Shapley \cite{gale1962college} showed that in the Gale-Shapley algorithm, the output is optimal for each individual in the proposing set among all stable matchings. 
It is natural to ask whether an analogy statement is true in our case.
In Claim \ref{claim:gs_wa_not_opt}, we show that for \Cref{alg:GS_WA}, where the proposing set is the doctors, it is not the case. Specifically, \Cref{alg:GS_WA} does not output the optimal solution for the doctors among all fair and stable solutions.

\begin{claim}
\label{claim:gs_wa_not_opt}
Algorithm \ref{alg:GS_WA} does not return the optimal PIIF and contract stable solution for the set of the doctors.
\end{claim}

\begin{proof}
Consider the following example:
Let $\D = \set{i_1, i_2, i_3}$ be the set of doctors and $\H = \set{A, B, C}$ be the set of hospitals. All doctors are in the same cluster $i$.
The preferences of the doctors are
$$
A \succ_{i_1} B \succ_{i_1} C, \quad
A \succ_{i_2} C \succ_{i_2} B, \quad
C \succ_{i_3} A \succ_{i_3} B.
$$
Since there is only one cluster, the hospitals have no preferences over the doctors.

When running \Cref{alg:GS_PSP} we get the output
\begin{center}
    $\pi_1$ = 
    \begin{tabular}{c|c|c|c}
         & $A$ & $B$ & $C$  \\
         \hline
         $i_1$ & 1/2 & 1/2 & 0 \\
         \hline
         $i_2$ & 1/2 & 1/4 & 1/4 \\
         \hline
         $i_3$ & 0 & 1/4 & 3/4
    \end{tabular}
\end{center}

and when running \Cref{alg:GS_WA} we get the output
\begin{center}
    $\pi_2$ = 
    \begin{tabular}{c|c|c|c}
         & $A$ & $B$ & $C$  \\
         \hline
         $i_1$ & 1/3 & 2/3 & 0 \\
         \hline
         $i_2$ & 1/3 & 1/6 & 1/2 \\
         \hline
         $i_3$ & 1/3 & 1/6 & 1/2
    \end{tabular}
\end{center}
Where, $\pi_1(i_1) \succ_{i_1} \pi_2(i_1)$, $\pi_1(i_2) \succ_{i_2} \pi_2(i_2)$ and $\pi_1(i_3) \succ_{i_3} \pi_2(i_3)$.
\end{proof}

We leave open the question of whether \Cref{alg:GS_PSP} return the optimal solution for the hospitals.
\section{Impossibility Results}
\label{sec:impossibility}
In \Cref{sec:piif_algs}, we showed positive results under two non-trivial assumptions on the structure of both the hospital preferences and the metric. It is natural to ask whether those are necessary or whether there exist more general algorithms that do not rely on these assumptions.
In this section, we show that with the type of algorithms we have used so far, this is not possible. We define a class of \emph{local-proposing algorithms} and show that, for this class of algorithms, no algorithm outputs PIIF and ``stable'' allocations, even for our most minimal stability definition, weak ex-ante stability (see Definition \ref{def:selective_stability}).

In particular, we show two impossibility results for this class of algorithms:
\begin{enumerate}
    \item For unfair preferences, even if the fairness metric is a proto-metric, there is no algorithm in this class that outputs a PIIF and weakly ex-ante stable allocation.
    \item For fair preferences and general metrics, there is no algorithm in this class that outputs a PIIF and weakly ex-ante stable allocation.
\end{enumerate}

\subsection{Defining the Class of Local-Proposing Algorithms}
We start by defining two classes of algorithms, one for algorithms where the doctors propose and one for algorithms where the hospitals propose. The classes are inspired by the Gale-Shapley algorithm and by our new algorithms.

\paragraph{Doctor-proposing local algorithms.}
For the class of doctor-proposing algorithms, the algorithm proceed in rounds, where in each round, we assume each doctor proposes a probability mass to the hospitals. Since the doctors have deterministic preferences and do not have fairness constraints concerning the hospitals, we assume the doctors follow their preferences, and propose to the hospitals according to this order (more about this in Remark \ref{remark:doctor_propose_by_prefs}). The hospitals might have probabilistic preferences and do have fairness constraints concerning the doctors. Thus, we do not assume a specific behavior from the hospitals. 

\begin{definition}[Local-proposing algorithm with doctors proposing]
A {\fontfamily{lmtt}\selectfont 
local-proposing algorithm with doctors proposing
} (LADP) is an algorithm that proceeds in rounds such that:
\begin{itemize}
    \item At the beginning of the first round, each doctor has a free probability mass of 1. 

    \item At each round, each doctor chooses a hospital to propose its free probability mass to.
    The doctors must propose to the most preferred hospital that did not reject them yet.

    \item After the doctors' proposals, each hospital has the new probability mass proposed to it by each doctor at the current round, and the probability mass it did not reject in earlier rounds available. The hospital chooses some (possibly partial) probabilistic prospect using its available probability mass. The probability mass that is not used for the prospect is rejected.
\end{itemize}
\end{definition}

The algorithms in this class differ by the process of choosing the probabilistic prospect at each round by each hospital. We assume that in this process, the hospitals are not aware of the preferences of the other hospitals or doctors.

\paragraph{Hospital-proposing local algorithms.}
We also define the corresponding class where the hospitals propose to the doctors. In this case, motivated by \Cref{alg:GS_PSP}, where hospitals propose to clusters, we allow the hospitals to treat sets of doctors as equals, i.e., propose probability mass to a set and let the doctors in the set decide how to split it. For generality, we assume the hospitals can choose which sets of doctors to treat as equals for the current round at each round (this generalizes our \Cref{alg:GS_PSP}, where the clusters were fixed by the proto-metric). 
Further, we allow these sets of equal doctors to reallocate preexisting mass (from previous rounds) among them, as we allow the clusters in \Cref{alg:GS_PSP}. Thus, we allow each set treated equally in round $t$ by hospital $h$ to reallocate the probability mass hospital $h$ proposed to the set at this round and in earlier rounds. 

Although we allow this generalization, we still make assumptions about how the doctors split the probability mass inside these sets.
The assumptions we make are natural because it means the doctors are acting according to their preferences.
These assumptions are also aligned with the way each cluster splits the probability mass in \Cref{alg:GS_PSP}, using the probabilistic serial procedure (see \Cref{alg:PSP}). 
We assume that a doctor will never give up any probability mass from the highest ranked hospital that proposed to it and that it will not be assigned with less than an equal share for this hospital than any other doctor in its cluster (unless it contradicts the previous assumption). 
Also, we assume that a cluster of doctors rejects a hospital's probability mass only if all the doctors in the cluster prefer every hospital they are assigned to over this hospital.
We explain why these assumptions are important in Remark \ref{remark:doctor_propose_by_prefs}.

We allow another generalization for this class, by allowing the hospitals to propose at each round, less than their entire free probability mass. We explain more about why this generalization is essential in the proofs of Theorems \ref{thm:fair_prefs_lahp} and \ref{thm:unfair_prefs_lahp}.

\begin{definition}[Local-proposing algorithm with hospitals proposing]
A {\fontfamily{lmtt}\selectfont 
local-proposing algorithm with hospitals proposing
} ($\alpha$-LAHP) with parameter $\alpha \in (0,1]$ is an algorithm that proceeds in rounds such that:
\begin{itemize}
    \item At the beginning of the first round, each hospital $h$ has a free probability mass 1. We denote the free probability mass at round $t$ by $p_{h,t}$.
    
    \item At each round $t \ge 1$, each hospital chooses a partition to clusters $\C_{h,t}$ (clustering) of the doctors $\D$, and a function $f_{h,t} : \C_{h,t} \rightarrow [0,1]$ such that $h$ proposes to each cluster of doctors $C \in \C_{h,t}$ probability mass $f_{h,t}(C)$. For a doctor $d\in\D$ that is clustered in cluster $C \in \C_{h,t}$ in round $t$, we denote that $f_{h,t}(d) = f_{h,t}(C)$.
    The function $f_{h,t}$ must satisfy the following conditions:
    \begin{itemize}
        \item $\sum_{C \in \C_{h,t}} f_{h,t}(C) \in [\min\set{\alpha, p_{h,t}}, p_{h,t}]$ (the hospital proposes at least $\alpha$ mass, if it can). 
        \item $\forall d \in \D:\sum_{t} f_{h,t}(d) \le 1$ (the sum of all proposals to any doctor is at most 1).
    \end{itemize}

    \item After the proposals of the hospitals, for every hospital $h$ and cluster $C \in \C_h$, the proposed probability mass for cluster $C$ by hospital $h$ is 
        \begin{equation*}
            g_{h,t}(C) = \sum_{d \in C}\pi_{t-1}(d)[h] + f_{h,t}(C),
        \end{equation*}
        where $\forall d \in \D:\pi_{0}(d)[h] = 0$.
    
    \item Then, each doctor $d$ has to choose a (possibly partial) probabilistic prospect $\pi_t(d)$. Any proposed probability mass that is not used in the allocations is rejected. The allocations must satisfy the following conditions:
    \begin{itemize}
        \item For every hospital $h$ and cluster $C \in \C_{h,t}$:
        \begin{equation*}
            \sum_{d \in C} \pi_{t}(d)[h] \le g_{h,t}(C).
        \end{equation*}
        \item For every doctor $d$, denote by $h^*_d$ the highest ranked hospital by doctor $d$ that satisfies that for some cluster $C \in \C_{h^*_d,t}$ such that $d \in C$, $g_{h,t}(C) > 0$, i.e., hospital $h^*_d$ is the highest ranked hospital that proposed to a cluster that contains doctor $d$ in round $t$. Then $\pi_t(d)$ satisfies that:
        \begin{itemize}
            \item Doctor $d$ will not accept less probability mass than other doctors in the cluster, unless they carried it from previous rounds, i.e., for every $d' \in C$:
            \begin{equation*}
                \pi_t(d)[h^*_d] \ge \pi_t(d')[h^*_d],
            \end{equation*}
            or
            \begin{equation*}
                \pi_t(d')[h^*_d] = \pi_{t-1}(d')[h^*_d].
            \end{equation*}
            \item Doctor $d$ will accept at least the same probability mass as in the previous round, i.e.,
            \begin{equation*}
                \pi_t(d)[h^*_d] \ge \pi_{t-1}(d)[h^*_d].
            \end{equation*}
        \end{itemize}
        \item For every hospital $h$ and cluster $C\in \C_{h,t}$, if
        \begin{equation*}
            \sum_{d \in C} \pi_{t}(d)[h] < g_{h,t}(C),
        \end{equation*}
        i.e., the cluster $C$ rejects probability mass proposed by hospital $h$, then for every doctor $d \in C$
        \begin{equation*}
            \sum_{h' \in \H} \pi_t(d)[h'] = 1,
        \end{equation*}
        and
        \begin{equation*}
            \forall h' \in supp(\pi_t(d)): h' \succeq_d h.
        \end{equation*}
        In words, if the cluster $C$ rejects probability mass from hospital $h$, then all the doctors in cluster $C$ have full allocations and they prefer all the hospitals in their prospect over hospital $h$.
    \end{itemize}
    \item Any probability mass that is unallocated by the doctors becomes free again, i.e.,
    \begin{equation*}
        \forall h \in \H: p_{h,t+1} = \sum_{C \in \C_h}\left(g_{h,t}(C) - \sum_{d \in C}\pi_t(d)[h]\right).
    \end{equation*} 
\end{itemize}
\end{definition}

The algorithms in this class differ by the process of choosing the proposals at each round by each hospital, and the process of choosing the prospects in each round by the clusters.
We assume that the hospitals are not aware of the preferences of the other hospitals or doctors.

The class of algorithms for which we want to show impossibility is the union of the two classes defined above.

\begin{definition}[Local-proposing algorithm]
A local-proposing algorithm is either a local-proposing algorithm with doctors proposing or a local-proposing algorithm with hospitals proposing.
\end{definition}

\paragraph{Further notes.} The following remarks explain why the assumptions made in these classes definitions are reasonable, and how they can be generalized to allow a small error.
\begin{remark}
\label{remark:doctor_propose_by_prefs}
In LAHP, the assumption that doctors accept the proposals according to their true preferences is essential. Otherwise, for a small value of $\alpha$, the doctors and hospitals can communicate their preferences (see below for more about how this could work). After a few rounds of communication, each hospital can run a centralized algorithm to find a PIIF and weakly ex-ante allocation and propose its allocation to the doctors. The doctors accept at this point. Since $\alpha$ is small, the amount of probability mass ``wasted'' in these communication rounds is small too. Thus, the output is $\epsilon$-close to a PIIF and weakly ex-ante allocation.
Thus, if we allow the doctors to accept the proposals arbitrarily, there is no meaning in restricting to this algorithm class. 

The communication can be executed as follows: At the initial round, each hospital $h$ proposes $\frac{\alpha}{n}$ mass to each doctor, and the doctor accepts $t\cdot \frac{\alpha}{n}$, where $t$ is the position of hospital $h$ in its ranking. That way, the doctors can communicate its entire order of preferences to $h$. Similarly, the hospitals communicate their preferences to the doctors. Then, the doctors communicate the preferences of all the hospitals to the doctors. 
\end{remark}

\begin{remark}
The Gale-Shapley algorithm is a local proposing algorithm. Algorithm \ref{alg:GS_PSP} and Algorithm \ref{alg:GS_WA} are nearly in the class. The only difference is that these algorithms, we allow to stop if a very small probability mass $\tau$ has left unallocated, and because of this, these algorithms achieve $2\tau$-PIIF and $\tau$-weak ex-ante stability. Allowing this for the class of local-proposing algorithms will not change our results for $\tau$ small enough. 

For the classes LADP and 1-LAHP, we show that the algorithms in these classes cannot return a solution without a constant error (1/3). Thus, for $\tau$ smaller than $1/6$, allowing the algorithms to output an approximation will not "help".

For the class $\alpha$-LAHP, in our counter example we show that it is impossible to achieve $\tau < 1/\left(\ceil{\frac{1}{\alpha}} + 3\right) \approx \alpha$.
\end{remark}

\subsection{Impossibility Result for Unfair Hospital Preferences}
\label{sec:unfair_prefs}
Our first impossibility result shows that no local-proposing algorithm outputs a PIIF and weakly ex-ante stable allocation for unfair hospital preferences.

\begin{theorem}
\label{thm:unfair_prefs_fails}
There is no local-proposing algorithm that outputs a PIIF and weakly ex-ante stable allocation if the preferences of the hospitals are not individually fair.
\end{theorem}

We start by showing this for each of the classes. 

\paragraph{Doctors propose.}
We show that if we do not require fair preferences from the hospitals, no LADP satisfies both the fairness and stability requirements. We present a setting with four doctors: $i, j, k, l$ where $d(i,j) = 0, d(k, l) = 0$ and the distance between these two clusters is 1.
Hospital $A$'s preferences are $$j \succ_A k \succ_A i \succ_A l.$$
In the first round, doctors $i$ and $k$ propose to hospital $A$ probability mass 1, and hospital $A$ needs to decide what to accept and reject. By controlling the preferences of other hospitals, we can create two cases:
\begin{enumerate}
    \item No other doctor proposes to hospital $A$ in any later round, i.e., doctors $j$ and $l$ are matched to hospitals they prefer over hospital $A$. In this case, the only way to achieve fairness and stability is to accept probability 1 from doctor $k$ and $0$ from doctor $i$.
    \item The doctors $j$ and $l$ propose probability mass $1$ to hospital $A$ in the second round. This implies that hospital $A$ is the most preferred hospital by doctors $i, j, k$ and $l$, among the hospitals to which they can be possibly matched to.
    Since $d(k, l) = 0$, in any fair allocation, any prospect for hospital $A$ that involves doctor $k$ with probability $p$ must involve doctor $l$ with probability $p$ too. Thus, the prospect $$\begin{cases}
        j,&\textit{ w.p. }1/2\\
        i,&\textit{ w.p. }1/2
    \end{cases}$$
    stochastically dominates any prospect that involves doctor $k$. Hence, in the first round, when only doctors $i$ and $k$ have proposed, the only way to achieve fairness and stability is to accept at least probability 1/2 from doctor $i$. That way, later hospital $A$ would be able to accept probability mass 1/2 from $j$.
\end{enumerate}

The problem is that hospital $A$ cannot distinguish between these two cases in the first round.

\begin{theorem}
\label{thm:unfair_prefs_ladp}
There is no algorithm $a \in LADP$ that returns a PIIF and weakly ex-ante stable allocation in the setting of a proto-metric and unfair preferences.
\end{theorem}

\begin{proof}
Assume for contradiction there exists a LADP $a$ that returns a PIIF and weakly ex-ante stable allocation. Consider the following example:

Let $\D = \set{i, j, k, l, x, y}$ be the set of doctors, $\H = \set{A, B, C, D, E, F}$ be the set of hospitals and all the distances are 1 except: $d(k, l) = 0, d(i, j) = 0$.

The preferences of the doctors are: 
\begin{align*}
&A \succ_i B \succ_i C \succ_i D \succ_i E \succ_i F,\\
&F \succ_j A \succ_j B \succ_j C \succ_j D \succ_j E,\\
&A \succ_k B \succ_k C \succ_k D \succ_k E \succ_k F,\\
&E \succ_l A \succ_l B \succ_l C \succ_l D \succ_l F,\\
&F \succ_x D \succ_x A \succ_x B \succ_x C \succ_x E,\\
&E \succ_y C \succ_y A \succ_y B \succ_y D \succ_y F.
\end{align*}

The preferences of the hospitals (except $E$ and $F$) are:
\begin{align*}
    &j \succ_A k \succ_A i \succ_A l \succ_A x \succ_A y, \\
    &j \succ_B k \succ_B i \succ_B l \succ_B x \succ_B y, \\
    &y \succ_C x \succ_C j \succ_C k \succ_C i \succ_C l, \\
    &y \succ_D x \succ_D j \succ_D k \succ_D i \succ_D l.
\end{align*}

Let us divide into two cases:
\begin{enumerate}
    \item Hospital $F$'s preferences are: 
    $$j \succ_F x \succ_F k \succ_F i \succ_F l\succ_F y$$
    and hospital $E$'s preferences are: 
    $$l \succ_E y \succ_E k \succ_E i \succ_E j\succ_E x.$$

    In any PIIF allocation $\pi$, the selectively fair alternative allocation $\nu_{F,j} = (F, \set{j}, \set{j,\textit{ w.p. }1})$ exists and is active if $\Pr[\pi(F) = j] < 1$. The alternative $\nu_{F,j}$ satisfies the fairness requirement between doctors $i$ and $j$ since hospital $F$ is doctor $i$'s least preferred hospital. 
    The alternative $\nu_{F,j}$ is active since doctor $j$ and hospital $F$ are each other's most preferred match.
    
    
    Given the assumption $\Pr[\pi(F) = j] = 1$, similarly, for any PIIF allocation $\pi$ where $\Pr[\pi(E) = l] < 1$, the selectively fair alternative allocation $(E, \set{l}, \set{l, \textit{ w.p. } 1})$ is active.
    
    Given $\Pr[\pi(F) = j] = 1$ and $\Pr[\pi(E) = l] = 1$, for any PIIF allocation $\pi$ where $\Pr[\pi(D) = x] < 1$, the selectively fair alternative allocation $(D, \set{x}, \set{x, \textit{ w.p. } 1})$ is active. 
    
    Given $\Pr[\pi(F) = j] = 1$, $\Pr[\pi(D) = x] = 1$ and $\Pr[\pi(E) = l] = 1$, for any PIIF allocation $\pi$ where $\Pr[\pi(C) = y] < 1$, the selectively fair alternative allocation $(C, \set{y}, \set{y, \textit{ w.p. } 1})$ is active. 
    
    Given all the above and since $k \succ_A i$, if $\Pr[\pi(A) = k] < 1$ the selectively fair alternative allocation $(A, \set{k}, \set{k, \textit{ w.p. } 1})$ is active. Since doctor $l$ is always matched to hospital $E$ and $E \succ_l A$, the fairness requirement between doctors $k$ and $l$ is satisfied.
    
    Thus, the only PIIF and weakly ex-ante stable allocation is $$\pi^* = \set{(A, k), (B, i), (C, y), (D, x), (E, l), (F, j)\textit{ w.p. }1}.$$
    
    Since algorithm $a$ is a LADP, in the first round, doctors $j$ and $x$ propose 1 to hospital $F$, doctors $l$ and $y$ propose 1 to hospital $E$ and doctors $i$ and $k$ propose 1 to hospital $A$. By the assumptions above, hospital $F$ accepts doctor $j$ and rejects doctor $x$, hospital $E$ accepts doctor $l$ and doctor rejects $y$ and hospital $A$ accepts doctor $k$ and rejects doctor $i$. At the second round, doctor $i$ proposes 1 to hospital $B$, doctor $x$ propose 1 to hospital $D$ and doctor $y$ propose 1 to hospital $C$ and they all accept.
    
    \item Hospital $F$'s preferences are: $x \succ_F j \succ_F k \succ_F i \succ_F l\succ_F y$ and hospital $E$'s preferences are: $y \succ_E l \succ_E k \succ_E i \succ_E j\succ_E x$.
    
    By a similar analysis to the first case, for any PIIF and weakly ex-ante stable matching $\pi$: $\Pr[\pi(F) = x] = 1$ and $\Pr[\pi(E) = y] = 1$.
    
    The only possibility for a PIIF and weakly ex-ante stable allocation is if hospital $A$'s prospect satisfies that $\Pr[\pi(A) = i] = \Pr[\pi(A) = j] = 1/2$. Otherwise the selectively fair alternative allocation $$
        \nu_{A,i,j} = \left(A, \set{i,j}, \sigma_{i,j} = \begin{cases}
            i,& \textit{w.p. } 1/2\\ j,&\textit{w.p. } 1/2
        \end{cases}\right)
    $$
    is active. This is because:
    \begin{itemize}
        \item The distribution $\sigma_{i,j}$ is IF.
        \item Since doctor $j$ and hospital $F$ are never matched, hospital $A$ is at least preferred by doctors $i$ and $j$ as any hospital in the support of their prospect.
        \item Hospital $A$ is the most preferred hospital that doctors $k$ and $l$ can be matched to, and $d(k,l) = 0$. Therefore, in any PIIF allocation, $\Pr[\pi(A) = l] = \Pr[\pi(A) = k]$. Thus, the prospect $\sigma_{i,j}$ stochastically dominates any PIIF prospect by $A$, subject to PIIF.
    \end{itemize}
    
    Since algorithm $a$ is a LADP in the first round, doctors $j$ and $x$ propose 1 to hospital $F$, doctors $l$ and $y$ propose 1 to hospital $E$ and doctors $i$ and $k$ propose 1 to hospital $A$.
    
    Hospital $A$ got the same proposals as in the previous run and cannot distinguish that the preferences of the other hospitals have changed. Thus, hospital $A$ accepts doctor $k$ and rejects doctor $i$. In any stable allocation, hospital $A$ is matched to doctor $i$ with a probability of 1/2; this contradicts the assumption that algorithm $a$'s output is PIIF and weakly ex-ante stable.
\end{enumerate}
\end{proof}

\paragraph{Hospitals propose.}
Next, we show that no LAHP outputs a PIIF and weakly ex-ante stable allocation for unfair preferences.
We start by showing this for $\alpha = 1$.

The idea behind the proof is showing a scenario with four doctors, $i, j, k$ and $l$ where doctors $i$ and $j$ are similar and doctors $k$ and $l$ are similar. Hospital $A$'s preferences are $j 
\succ_A k \succ_A i \succ_A l$. The doctors $i$ and $k$ rank hospital $A$ first. If the doctors $j$ and $l$ rank hospital $A$ first too, in any PIIF and weakly ex-ante stable allocation, hospital $A$'s prospect is
$$\begin{cases}
    i,& \textit{w.p. } 1/2 \\ 
    j,& \textit{w.p. } 1/2.
\end{cases}$$ 
Since hospital $A$ is ranked first by all the doctors, the first proposal $A$ makes is accepted and is its final allocation. Thus, $A$ must propose to $$\begin{cases}
    i,& \textit{w.p. } 1/2 \\ 
    j,& \textit{w.p. } 1/2
\end{cases}$$
in the first round. If doctors $j$ and $l$ do not rank hospital $A$ first, and there are other hospitals, $D$ and $C$, that rank doctors $j$ and $l$ first and are ranked first by doctors $j$ and $l$ respectively, in any PIIF and weakly ex-ante stable allocation, hospital $A$'s prospect is $\set{k , \textit{ w.p. } 1}$. Hospital $A$ cannot distinguish these two cases (since it doesn't know the doctors' preferences). Thus, it must propose to doctors $i$ and $j$ in the second case too, and doctor $i$ must accept. Thus, the algorithm outputs an unstable allocation. 

\begin{theorem}
In the setting of a proto-metric and unfair preferences, there is no algorithm $a \in 1$-LAHP that returns a PIIF and weakly ex-ante stable allocation.
\end{theorem}

\begin{proof}
Assume for contradiction there exists a 1-LAHP $a$ that returns a PIIF and weakly ex-ante stable allocation. Consider the following example:

Let $\D = \set{i, j, k, l}$ be the set of doctors, $\H = \set{A, B, C, D}$ be the set of hospitals and all the distances are 1 except: $d(i, j) = 0, d(k, l) = 0$.

The preferences of the doctors $i$ and $k$ are: 
\begin{align*}
&A \succ_i B \succ_i C \succ_i D, \\
&A \succ_k B \succ_k C \succ_k D. 
\end{align*}
The preferences of the hospitals are:
\begin{align*}
&j \succ_A k \succ_A i \succ_A l,\\
&j \succ_B k \succ_B i \succ_B l,\\
&j \succ_D l \succ_D i \succ_D k,\\
&k \succ_C j \succ_C l \succ_C i.
\end{align*}

We divide into two cases:
\begin{enumerate}
    \item The preference of doctor $j$ are $$A \succ_j B \succ_j C \succ_j D$$
    and the preferences of doctor $l$ are
    $$A \succ_l B\succ_l D \succ_l C.$$
    First, we show that the only PIIF and weakly ex-ante stable solution for this set of preferences is
    $$
    \pi^*_1(A) = \begin{cases}
        i,&\textit{ w.p. }1/2 \\ j,&\textit{ w.p. }1/2.
    \end{cases}
    $$
    Let $\pi$ be some PIIF and weakly ex-ante stable allocation, assume that 
    $$\pi(A) \ne \begin{cases}
        i,&\textit{ w.p. }1/2 \\ j,&\textit{ w.p. }1/2
    \end{cases},$$
    we claim that the selectively fair alternative allocation $$\nu_{A,i,j} = \left(A, \set{i,j}, \sigma_{i,j}=\begin{cases}
        i,&\textit{ w.p. }1/2 \\ j,&\textit{ w.p. }1/2
    \end{cases}\right)$$
    is active. This is because:
    \begin{itemize}
        \item The distribution $\sigma_{i,j}$ is IF.
        \item Hospital $A$ is the most preferred hospital by doctors $i$ and $j$.
        \item Since $\pi$ is PIIF, $d(k,l) = 0$ and $A$ is ranked first by $k$ and $l$, we know that $\Pr[\pi(A) = k] = \Pr[\pi(A) = l]$, which implies that $\sigma_{i,j} \succ_A \pi(A)$.
    \end{itemize}

    
    Hospital $A$ is ranked first by all the doctors. Thus, no matter what clustering hospital $A$ chooses in the first round, in any cluster $C$ that hospital $A$ proposes to, every doctor in cluster $C$ must accept $1/|C|$ fraction of the proposed probability mass from hospital $A$. For the same reason, none of the clusters would reject any probability mass from hospital $A$. Thus, any probability mass proposed by hospital $A$ to any cluster would be accepted and assigned to hospital $A$ in the final allocation. 
    
    Thus, hospital $A$ has two options:
    \begin{enumerate}
        \item Choose a clustering where doctors $i$ and $j$ are a cluster $C_{i,j}$, and propose probability mass 1 to cluster $C_{i,j}$ on the first round. Both doctors $i$ and $j$ rank hospital $A$ first. Thus, they cannot reject any of the probability mass hospital $A$ proposed. From this reason and since they both do not have any probability mass from previous rounds, they must choose the allocation $\sigma_{i,j}$.
        \item Choose a clustering such that doctors $i$ and $j$ are each on their own singleton cluster $C_i$ and $C_j$ respectively and propose probability mass 1/2 to each of them. From they same reasons they both must accept.
    \end{enumerate}
    \item The preference of doctor $j$ are 
    $$D \succ_j A \succ_j B \succ_j C,$$
    and the preferences of doctor $l$ are 
    $$C \succ_l B\succ_l A \succ_l D.$$
    
    The only PIIF and weakly ex-ante stable solution for this set of preferences is 
    $$
    \pi^*_2 = \set{(k, A), (j, D), (l, C), (i, B) \textit{ w.p. } 1}.
    $$
    This is because doctor $j$ ranks hospital $D$ first and hospital $D$ ranks doctor $j$ first, and doctor $i$ ranks hospital $D$ last. Thus, doctor $j$ and hospital $D$ must be matched. Otherwise, the selectively fair alternative allocation $(D, \set{j}, \set{j, \textit{ w.p. }1})$ is active.
    
    Similarly, the pairs $(C,l)$, $(A,k)$ and $(B,i)$ rank each other first between individuals whose prospect has not been determined yet, and are at distance 1 from $\D\backslash\set{j}$. Thus, algorithm $a$ must match them with probability 1.
    
    However, we know that hospital $A$ has only two possible way to choose a clustering and proposals on the first round:
    \begin{enumerate}
        \item Choose the clustering where doctors $i$ and $j$ are in the same cluster $C_{i,j}$ and proposes 1 to this cluster.
        Doctor $i$ still ranks hospital $A$ first, therefore, none of the probability mass hospital $A$ proposed is rejected from this cluster. Since there is no allocated probability mass from previous rounds, at the end of this round $\pi_1(i)[A] \ge \pi_1(j)[A]$ which implies that $\pi_1(i)[A] \ge 1/2$.
        \item Choose the clustering where doctors $i$ and $j$ are each on their own cluster $C_i$ and $C_j$ respectively and propose probability mass 1/2 to each of them. From the same reasons, doctor $i$ cannot reject any probability mass that hospital $A$ proposed, and $\pi_1(i)[A] = 1/2$.
    \end{enumerate}
    Since doctor $i$ ranks hospital $A$ first, no matter what clustering or proposals hospital $A$ chooses in later rounds, doctor $i$ would accept from hospital $A$ at least the same probability mass as in this round.
    This means that in the final allocation $\pi(i)[A] \ge \pi_1(i)[A] = 1/2 > 0$. However, in the only stable allocation, $\pi^*_2(i)[A] = 0$. Thus, algorithm $a$ does not return the only PIIF and weakly ex-ante stable solution.
\end{enumerate}
\end{proof}

\paragraph{Extending to general $\alpha$.}
Next, we show that the negative result can be extended to any $\alpha > 0$. The idea behind this proof is similar to the case where $\alpha=1$. However, in this case, we only know that while hospital $A$ does not distinguish the cases described above, hospital $A$ proposes to $$\begin{cases}
        i,&\textit{ w.p. }1/2 \\ j,&\textit{ w.p. }1/2
\end{cases}$$
in the first $\ceil{\frac{1}{\alpha}}$ rounds. Moreover, any probability that hospital $A$ proposes to doctor $i$, will remain matched to $i$ at the final allocation. Thus, hospital $A$ might want to delay proposing to doctor $i$ as long as possible. To make sure that hospital $A$ does propose some probability mass to doctor $i$, hospital $D$ must be delayed for at least $\ceil{\frac{1/2}{\alpha}}$ rounds before proposing to doctor $j$ and letting hospital $A$ notice the difference in the preferences of doctor $j$. For simplicity, in the proof we delay hospital $D$ by $\ceil{\frac{1}{\alpha}}$ rounds.

\begin{theorem}
\label{thm:unfair_prefs_lahp}
There is no algorithm $a \in \alpha$-LAHP that returns a PIIF and weakly ex-ante stable allocation in the setting of a proto-metric and unfair preferences.
\end{theorem}

\begin{proof}
Assume for contradiction that there exists a $\alpha$-LAHP $a$, for $\alpha \in (0,1]$, that returns a PIIF and weakly ex-ante stable allocation. Consider the following example:

Define $\beta = \ceil{\frac{1}{\alpha}}$. Let $\D = \set{i, j, k, l, d_1, ..., d_{\beta}}$ be the set of doctors, $\H = \set{A, B, C, D, H_1, ..., H_{\beta}}$ be the set of hospitals and all the distances are 1 except: $d(i, j) = 0, d(k, l) = 0$.

The preferences of the doctors (except doctors $j$ and $l$) are: 
\begin{align*}
&A \succ_i B \succ_i C \succ_i D \succ_i H_1 \succ_i ... \succ_i H_{\beta},\\
&A \succ_k B \succ_k C \succ_k D \succ_k H_1 \succ_k ... \succ_k H_{\beta},\\
\forall m \in [{\beta}] :& A \succ_{d_m} H_m \succ_{d_m} H_1 \succ_{d_m}... \succ_{d_m} H_{m-1} \succ_{d_m} H_{m+1} \succ_{d_m}...\succ_{d_m} H_{\beta} \succ_{d_m} B \succ_{d_m} C \succ_{d_m} D. 
\end{align*}
The preferences of the hospitals are:
\begin{align*}
&j \succ_A k \succ_A i \succ_A l \succ_A d_1 \succ_A ... \succ_A d_{\beta},\\
&j \succ_B k \succ_B i \succ_B l \succ_B d_1 \succ_B ... \succ_B d_{\beta},\\
&d_1 \succ_D ... \succ_D d_{\beta} \succ_D j \succ_D l \succ_D i \succ_D k,\\
&k \succ_C j \succ_C l \succ_C i \succ_C d_1 \succ_C ... \succ_C d_{\beta},\\
\forall m \in [{\beta}] :& d_m \succ_{H_m} d_1  \succ_{H_m}... \succ_{H_m} d_{m-1} \succ_{H_m} d_{m+1} \succ_{H_m}...  \succ_{H_m} d_{\beta} \succ_{H_m} i \succ_{H_m} k \succ_{H_m} j \succ_{H_m} l.
\end{align*}

We divide into two cases:
\begin{enumerate}
    \item The preference of doctor $j$ are $$A \succ_j B \succ_j C \succ_j D \succ_j H_1 \succ_j ... \succ_j H_{\beta}$$
    and the preferences of doctor $l$ are
    $$A \succ_l B\succ_l D \succ_l C \succ_l H_1 \succ_l ... \succ_l H_{\beta}.$$
    First, we show that the only PIIF and weakly ex-ante stable solution for this set of preferences is where
    $$
    \pi^*_1(A) = \begin{cases}
        i,&\textit{ w.p. }1/2 \\ j,&\textit{ w.p. }1/2.
    \end{cases}
    $$
    Let $\pi$ be some PIIF and weakly ex-ante stable allocation, assume that 
    $$\pi(A) \ne \begin{cases}
        i,&\textit{ w.p. }1/2 \\ j,&\textit{ w.p. }1/2,
    \end{cases}$$
    we claim that the selectively fair alternative allocation $$\nu_{A,i,j} = \left(A, \set{i,j}, \sigma_{i,j} = \begin{cases}
        i,&\textit{ w.p. }1/2 \\ j,&\textit{ w.p. }1/2
    \end{cases}\right)$$
    is active. This is because:
    \begin{itemize}
        \item The distribution $\sigma_{i,j}$ is IF.
        \item For any doctor $d \in \set{d_1, ..., d_{\beta}}$, replacing $d$ by doctor $i$ or doctor $j$ stochastically dominates the current allocation. Since the allocation $\pi$ is PIIF, $d(k,l) = 0$ and hospital $A$ is ranked first by doctors $k$ and $l$, we know that $\Pr[\pi(A) = k] = \Pr[\pi(A) = l]$, which implies that $\sigma_{i,j} \succ_A \pi(A)$.
        \item Hospital $A$ is doctors $i$ and $j$'s most preferred hospital.
    \end{itemize}

    Hospital $A$ is the most preferred hospital by all the doctors. Thus, for any clustering option, in any round, any proposals that hospital $A$ makes are accepted with some probability by all the doctors in the proposed clusters. 
    Thus, in the first $\beta$ rounds, hospital $A$ can propose only to clusters that contain doctors $i$ and $j$.
    This implies that after $\beta$ rounds, doctor $i$ has accepted probability mass 1/2 from hospital $A$.
    \item The preference of $j$ are $$D \succ_j A \succ_j B \succ_j C \succ_j H_1 \succ_j ... \succ_j H_{\beta}$$
    and the preferences of $l$ are 
    $$C \succ_l B\succ_l A \succ_l D \succ_l H_1 \succ_l ... \succ_l H_{\beta}.$$
    
    The only PIIF and weakly ex-ante stable solution for these preferences is 
    $$
    \pi^*_2 = \set{(k, A), (j, D), (l, C), (i, B), (d_1, H_1), ...,(d_{\beta},H_{\beta}) \textit{ w.p. } 1}.
    $$
    This is because doctor $j$ ranks hospital $D$ first, hospital $D$ ranks doctor $j$ first and doctor $i$ ranks hospital $D$ last. Thus, doctor $j$ and hospital $D$ must be matched. Otherwise, the selectively fair alternative allocation $(D, \set{j}, \set{j, \textit{ w.p. }1}$ is active. 
    
    Similarly, the pairs $(C,l)$, $(A,k)$, $(B,i)$, $\forall m \in [\beta]: (H_m, d_m)$ rank each other first between individuals whose prospect has not been determined yet, and are at distance 1 from $\D\backslash\set{j}$. Thus, algorithm $a$ must match them with probability 1.
    
    Assume hospital $A$ does not see any difference from the first case in the first $\beta$ rounds. Then, in one of the first $\beta$ rounds, hospital $A$ must propose 1/2 probability mass to doctor $i$, and $i$ must accept since it ranks $A$ first. For the same reason, doctor $i$ will never reject hospital $A$, and then algorithm $a$ does not return the only PIIF and weakly ex-ante stable solution. Hospital $A$ notices the difference in the preferences when hospital $D$ proposes some probability mass to doctor $j$, and $j$ rejects some probability mass from $A$, or reassign it to doctor $i$.
    
    Thus, to finish the proof, we have to show that hospital $D$ proposes to doctor $j$ after at least $\beta$ rounds.
    We show that hospital $D$ must propose to the doctors in $\D$ by order of its preferences before proposing to doctor $j$. Otherwise, algorithm $a$ might return an unstable allocation. We prove this by induction over the doctors in $\D \backslash \set{i,j,k,l}$ by order of $D$'s preferences. 
    
    \textbf{Base:} The first doctor that hospital $D$ proposes to must be $d_1$. Otherwise, if hospital $D$ is ranked highest by all the doctors, any other doctor that hospital $D$ proposes will accept and not reject in any later round. Then the output will be unstable.
    In the first round, hospital $D$ cannot distinguish this case from any other case.
    Thus, hospital $D$ must cluster doctor $d_1$ in a singleton cluster and propose to this singleton in the first round.
    
    \textbf{Step:} The $m$-th doctor that hospital $D$ proposes to must be the one ranked $m$-th by $D$. Assume from the induction assumption that until step $t$, hospital $D$ proposed to clusters that contains doctors from the set $$S_{m-1} = \set{r_D^{-1}(d) \le m-1: d \in \D}.$$
    We claim that at step $t+1$, hospital $D$ will propose only to clusters that contain doctors from $S_{m}$. Otherwise, consider the case where for every doctor $d \in S_{m-1}$, $r_d^{-1}(D) = |\D|$ there exists some hospital $h \in \H \backslash \set{D}$ such that $r_d(1) = h$ and  $r_h(1) = d$. 
    
    For every $d \in S_{m-1}$, in any weakly ex-ante stable allocation, $d$ must be matched to the hospital they ranked first with probability 1.
    
    At step $t+1$, hospital $D$ does not have any knowledge about the preferences of the doctors in $\D\backslash S_{m-1}$. Suppose that for every doctor $d \in \D\backslash S_{m-1}$, $r_d(1) = D$. Then in any weakly ex-ante stable allocation, hospital $D$ is matched $d_m$ with probability 1.
    
    More than that, any doctor from $\D\backslash S_m$ that hospital $D$ will propose to will be matched to $D$ in the final allocation some probability mass. 
    Thus, hospital $D$ cannot propose at step $t+1$ to any cluster that contains a doctor in $\D \backslash S_m$. 
\end{enumerate}
\end{proof}

We are now set to prove Theorem \ref{thm:unfair_prefs_fails}.
\begin{proof}[Proof of Theorem \ref{thm:unfair_prefs_fails}]
The proof is immediate from Theorems \ref{thm:unfair_prefs_lahp} and \ref{thm:unfair_prefs_ladp}.
\end{proof}

\subsection{Impossibility Result for General Metrics}
\label{sec:failure_fair_prefs}

In this section, we show that even if we require mutual replacement individually fair hospital preferences (see Definition \ref{def:mutual_replacement_individual_fairness}), the strongest fairness requirement we defined for general metrics, there is no local-proposing algorithm that outputs a PIIF and weakly ex-ante stable allocation.

\begin{theorem}
\label{thm:gen_metrics_fail}
There is no local-proposing algorithm that outputs a PIIF and weakly ex-ante stable allocation if the preferences of the hospitals are mutual replacement individually fair with respect to a general metric.
\end{theorem}

For the counter examples we use later in this section, we define below a distribution $r_h$ over the preferences of a hospital $h$ over four doctors $d_1, d_2, d_3$ and $d_4$. The distances according to the metric are $d(d_1, d_3) = 1/3$, $d(d_2, d_4) = 1/3$ and the rest of the distances are 1. The distribution $r_h$ is mutual replacement individually fair. Moreover, the distribution $r_h$ satisfies that:
\begin{equation*}
    \forall i \in [3]: \set{d_i, \textit{ w.p. }1} \succeq_h \set{d_{i+1}, \textit{ w.p. }1}
\end{equation*}
where the comparison is by stochastic domination.

We denote this type of preferences as 
$$
d_1 \,\Tilde{\succ}_h\, d_2 \,\Tilde{\succ}_h\, d_3 \,\Tilde{\succ}_h\, d_4.
$$

The following doubly stochastic matrix describes the distribution $r_h$:
\begin{center}
\begin{tabular}{c|c|c|c|c}
     & $d_1$ & $d_2$ & $d_3$ & $d_4$ \\
     \hline
     Rank 1 & 1/2 & 1/3 & 1/6 & 0  \\
     Rank 2 & 1/6 & 2/9 & 5/18 & 1/3  \\
     Rank 3 & 1/6 & 2/9 & 5/18 & 1/3  \\
     Rank 4 & 1/6 & 2/9 & 5/18 & 1/3 
\end{tabular}
\end{center}

To verify mutual-replacement individual fairness, we specify the full distribution $r_h$:
\begin{align*}
    r_h = 
    \begin{cases}
        d_2 \succ_h d_4 \succ_h d_1 \succ_h d_3,&\textit{w.p. } 2/36 \\
        d_2 \succ_h d_4 \succ_h d_3 \succ_h d_1,&\textit{w.p. } 2/36 \\
        d_2 \succ_h d_1 \succ_h d_4 \succ_h d_3,&\textit{w.p. } 2/36 \\
        d_2 \succ_h d_3 \succ_h d_4 \succ_h d_1,&\textit{w.p. } 2/36 \\
        d_2 \succ_h d_1 \succ_h d_3 \succ_h d_4,&\textit{w.p. } 2/36 \\
        d_2 \succ_h d_3 \succ_h d_1 \succ_h d_4,&\textit{w.p. } 2/36 \\ \\
        d_1 \succ_h d_3 \succ_h d_2 \succ_h d_4,&\textit{w.p. } 3/36 \\
        d_1 \succ_h d_3 \succ_h d_4 \succ_h d_2,&\textit{w.p. } 3/36 \\
        d_1 \succ_h d_2 \succ_h d_3 \succ_h d_4,&\textit{w.p. } 3/36 \\
        d_1 \succ_h d_4 \succ_h d_3 \succ_h d_2,&\textit{w.p. } 3/36 \\
        d_1 \succ_h d_2 \succ_h d_4 \succ_h d_2,&\textit{w.p. } 3/36 \\
        d_1 \succ_h d_4 \succ_h d_2 \succ_h d_3,&\textit{w.p. } 3/36 \\ \\
        d_3 \succ_h d_1 \succ_h d_2 \succ_h d_4,&\textit{w.p. } 1/36 \\ 
        d_3 \succ_h d_1 \succ_h d_4 \succ_h d_2,&\textit{w.p. } 1/36 \\
        d_3 \succ_h d_2 \succ_h d_1 \succ_h d_4,&\textit{w.p. } 1/36 \\
        d_3 \succ_h d_4 \succ_h d_1 \succ_h d_2,&\textit{w.p. } 1/36 \\
        d_3 \succ_h d_2 \succ_h d_4 \succ_h d_1,&\textit{w.p. } 1/36 \\
        d_3 \succ_h d_4 \succ_h d_2 \succ_h d_1,&\textit{w.p. } 1/36
    \end{cases}
\end{align*}

We also note that for any doctor $d$ at distance 1 from $d_1, d_2, d_3, d_4$, if the preferences are
$$
d \succ_h d_1 \,\Tilde{\succ}_h\, d_2 \,\Tilde{\succ}_h\, d_3 \,\Tilde{\succ}_h\, d_4
$$
or 
$$
d_1 \,\Tilde{\succ}_h\, d_2 \,\Tilde{\succ}_h\, d_3 \,\Tilde{\succ}_h\, d_4 \succ_h d
$$
the mutual-replacement individual fairness requirement is still satisfied.

\paragraph{Doctors propose.}
We start by showing that for any LADP, there exists a set of preferences of the doctors and hospitals and a distance metric for which the algorithm will output an unstable allocation.

This proof is similar to the proof for Theorem \ref{thm:unfair_prefs_ladp}. There are four doctors $i, j, k$ and $l$ such that $d(i,j) = 1/3$ and $d(k,l) = 1/3$ and the rest of the distances are 1. The preferences of hospital $A$ are $j \,\Tilde{\succ}_h\, k \,\Tilde{\succ}_h\, i \,\Tilde{\succ}_h\, l.$
In the first round, doctors $k$ and $i$ propose probability mass 1 to hospital $A$ (they rank $A$ first), and hospital $A$ does not know which doctors will propose to it in the following rounds:
\begin{itemize}
    \item If no other doctor will propose to hospital $A$ later, i.e., all the other doctors are matched to hospitals they prefer over hospital $A$, then in any weakly ex-ante stable allocation, hospital $A$'s prospect is $\set{k,\textit{ w.p. }1}$.
    Thus, hospital $A$ must reject all the probability mass that doctor $i$ proposes.
    
    \item If doctors $j$ and $l$ propose to hospital $A$ in the second round, i.e., they rank hospital $A$ first between hospitals that did not reject them, then in any PIIF and weakly ex-ante stable allocation, hospital $A$'s prospect is 
    $$\begin{cases}
        j,&\textit{ w.p. } 2/3\\ i,&\textit{ w.p. }1/3.
    \end{cases}$$
\end{itemize}
Since hospital $A$ cannot distinguish these two cases on the first round, assume that $A$ rejected all the probability mass that doctor $i$ proposed. Since $d(i,j) =1/3$ and doctor $i$ ranks hospital $A$ first, hospital $A$ can be matched to doctor $j$ with no more than probability $1/3$. Since doctors $l$ and $k$ both prefer hospital $A$ over all their possible matches and $d(k,l) = 1/3$, the best prospect for hospital $A$ subject to PIIF is 
$$\begin{cases}
    j,&\textit{ w.p. }1/3\\ k,&\textit{ w.p. } 1/2\\ l,&\textit{ w.p. }1/6
\end{cases}$$ and this prospect is stochastically dominated by
$$\begin{cases}
    j,&\textit{ w.p. } 2/3\\ i,&\textit{ w.p. }1/3,
\end{cases}$$
which leads to instability.

\begin{theorem}
\label{claim:fair_prefs_ladp}
In the setting of general metrics and mutual replacement individually fair hospital preferences, there is no $a \in LADP$ that returns a PIIF and weakly ex-ante stable allocation.
\end{theorem}

\begin{proof}
Assume for contradiction that there exists a LADP $a$ that returns a PIIF and weakly ex-ante stable allocation for mutual replacement individually fair hospital preferences. Consider the following example:

Let $\D = \set{i, j, k, l, x, y}$ be the set of doctors, $\H = \set{A, B, C, D, E, F}$ be the set of hospitals and all the distances are 1 except: $d(k, l) = 1/3,$ $d(i, j) = 1/3$.

The preferences of the doctors are: 
\begin{align*}
    &A \succ_i B \succ_i C \succ_i D \succ_i E, \succ_i F, \\
    &F \succ_j A \succ_j B \succ_j C \succ_j D \succ_j E ,\\
    &A \succ_k B \succ_k C \succ_k D \succ_k E \succ_k F, \\
    &E \succ_l A \succ_l B \succ_l C \succ_l D \succ_l F ,\\
    &F \succ_x D \succ_x A \succ_x B \succ_x C \succ_x E, \\
    &E \succ_y C \succ_y A \succ_y B \succ_y D \succ_y F.
\end{align*}

The preferences of the hospitals (besides hospitals $E$ and $F$) are:
\begin{align*}
    &j \,\Tilde{\succ}_A\, k \,\Tilde{\succ}_A\, i \,\Tilde{\succ}_A\, l \succ_A x \succ_A y, \\
    &j \,\Tilde{\succ}_B\, k \,\Tilde{\succ}_B\, i \,\Tilde{\succ}_B\, l \succ_B x \succ_B y, \\
    &y \succ_C x \succ_C j \,\Tilde{\succ}_C\, k \,\Tilde{\succ}_C\, i \,\Tilde{\succ}_C\, l, \\
    &x \succ_D y \succ_D j \,\Tilde{\succ}_D\, k \,\Tilde{\succ}_D\, i \,\Tilde{\succ}_D\, l.
\end{align*}
Let us divide into two cases:
\begin{enumerate}
    \item Hospital $F$'s preferences are: 
    \begin{align*}
        j \,\Tilde{\succ}_F\, k \,\Tilde{\succ}_F\, i \,\Tilde{\succ}_F\, l \succ_F x \succ_F y
    \end{align*}
    and hospital $E$'s preferences are: 
    \begin{align*}
        l \,\Tilde{\succ}_E\, j \,\Tilde{\succ}_E\, k \,\Tilde{\succ}_E\, i \succ_E x \succ_E y.
    \end{align*}
    
    The preferences of the hospitals are of the structure presented above. Therefore, they are mutual replacement individually fair.
    
    Let $\pi$ be a PIIF and weakly ex-ante stable solution for this problem.
    We claim that there is a unique solution $$\pi^* = \set{(A, k), (B, i), (C, y), (D, x), (E, l), (F, j), \textit{ w.p. } 1}.$$
    
    We first show that $\Pr[\pi(F) = j] = 1$. Assume for contradiction that $\Pr[\pi(F) = j] < 1$, then the selectively fair alternative allocation $\nu_{F,j} = (F, \set{j}, \set{j, \textit{ w.p. } 1})$ is active. This is because:
    \begin{itemize}
        \item Hospital $F$ is doctor $i$'s least preferred hospital and doctor $j$ is at distance 1 from all the other doctors. Thus, $\nu_{F,j}$ satisfies the fairness requirement.
        \item Hospital $F$ is doctor $j$'s most preferred hospital, so $j$ prefers $F$ over any hospital in doctor $j$'s prospect.
        \item Being assigned to doctor $j$ stochastically dominates any other prospect for hospital $F$.
    \end{itemize}
    Similarly, $\Pr[\pi(E) = l] = 1$.
    
    Next we claim that $\Pr[\pi(D) = x] = 1$ and $\Pr[\pi(C) = y] = 1$ since otherwise the contracts $\nu_D = (D, \set{x}, \set{x , \textit{ w.p. } 1})$ and $\nu_C = (C, \set{y}, \set{y , \textit{ w.p. } 1})$ are active. This is because:
    \begin{itemize}
        \item The distributions $\set{x , \textit{ w.p. } 1}, \set{y , \textit{ w.p. } 1}$ are individually fair.
        \item Hospitals $D$ and $C$ are doctors $x$ and $y$'s most preferred hospitals that are not always assigned to other doctors respectively. Thus, $x$ and $y$ always prefer $D$ and $C$ respectively over any hospital in their prospect in the allocation $\pi$.
        \item Being assigned to doctors $x$ and $y$ stochastically dominates any other possible prospect for hospitals $D$ and $C$ respectively.
    \end{itemize}
    
    Thus, hospital $A$'s prospect must be some combination of doctors $k$ and $i$ and since $d(k,i)=1$ and $k \,\Tilde{\succ}_A\, i$, we have that the only weakly ex-ante stable solution is when $\Pr[\pi(A) = k] = 1$.
    
    Given all the above, the allocation $\pi^*$ is the unique weakly ex-ante stable solution.
    
    Since algorithm $a$ is a PIIF and weakly ex-ante stable LADP, we know the order of proposals of the doctors and know what probability mass the hospitals must accept.
    In the first round, doctors $i$ and $k$ propose 1 to hospital $A$, hospital $A$ rejects doctor $i$ and accepts doctor $k$. Doctors $j$ and $x$ propose 1 to hospital $F$, hospital $F$ rejects doctor $x$ and accepts doctor $j$. Doctors $l$ and $y$ propose $1$ to hospital $E$, hospital $E$ rejects doctor $y$ and accepts doctor $l$.
    In the second round, doctor $i$ proposes 1 to hospital $B$, and hospital $B$ accepts. Doctor $x$ propose 1 to hospital $D$, and hospital $D$ accepts. Doctor $y$ proposes $1$ to hospital $C$, and hospital $C$ accepts.
    
    \item Hospital $F$'s preferences are: 
    \begin{align*}
         x \succ_F j \,\Tilde{\succ}_F\, k \,\Tilde{\succ}_F\, i \,\Tilde{\succ}_F\, l \succ_F y,
    \end{align*}
    and hospital $E$'s preferences are: 
    \begin{align*}
        y \succ_E l \,\Tilde{\succ}_E\, j \,\Tilde{\succ}_E\, k \,\Tilde{\succ}_E\, i \succ_E x.
    \end{align*}
    Since the preferences of the hospitals are of the structure presented earlier they are mutual replacement individually fair.
    
    Let the allocation $\pi$ be a PIIF and weakly ex-ante stable solution.
    First we claim that $\Pr[\pi(F) = x] = 1, \Pr[\pi(E) = y] = 1$. Otherwise, the selectively fair alternative allocations $\nu_F = (F, \set{x}, \set{x ,\textit{ w.p. } 1})$ and $\nu_E = (E, \set{y}, \set{y ,\textit{ w.p. } 1})$ are active. This is because:
    \begin{itemize}
        \item The distributions $\set{x,\textit{ w.p. } 1}, \set{y,\textit{ w.p. }1}$ are individually fair.
        \item Doctors $x$ and $y$ are hospitals $E$ and $F$'s most preferred doctors respectively. Thus, the alternative of always being matched to doctors $x$ and $y$ stochastically dominates any prospect for hospitals $E$ and $F$ respectively.
        \item Hospitals $F$ and $E$ are doctors $x$ and $y$'s most preferred hospitals respectively. Thus, hospitals $F$ and $E$ are always more preferred by $x$ and $y$ respectively over any hospital in their prospect in the allocation $\pi$.
    \end{itemize}
    
    Next, we claim that 
    $$\pi(A) = \begin{cases}
        i,&\textit{w.p. }1/3\\
        j,&\textit{w.p. }2/3
    \end{cases}.$$
    For any other PIIF allocation the selectively fair alternative allocation 
    $$\nu = \left(A, \set{i,j},\sigma_{i,j} =  \begin{cases}
        i,&\textit{w.p. } 1/3\\
        j,&\textit{w.p. } 2/3
    \end{cases}\right)$$
    is active. This is because:
    \begin{itemize}
        \item The distribution $\sigma_{i,j}$ is IF.
        \item Between $\H \backslash \set{E,F}$, hospital $A$ is the most preferred hospital by doctors $i, j, k, l$. Thus, in any PIIF solution 
        \begin{align*}
            |\Pr[\pi(A) = j] - \Pr[\pi(A) = i]| \le 1/3, \quad |\Pr[\pi(A) = k] - \Pr[\pi(A) = l]| \le 1/3.
        \end{align*}
        
        Under these constraints, the prospect $\sigma_{i,j}$ stochastically dominates $A$'s prospect, subject to PIIF.
    \end{itemize}
     
    When running algorithm $a$, in the first round, all the doctors propose the same as in the previous described run: doctors $j$ and $x$ propose 1 to hospital $F$, doctors $l$ and $y$ propose $1$ to hospital $E$ and doctors $i$ and $k$ propose 1 to hospital $A$. Since hospital $A$ cannot distinguish this run from the previous, hospital $A$ rejects doctor $i$ and accepts doctor $k$. Thus, in algorithm $a$'s output, doctor $i$ is never matched to hospital $A$. Thus, algorithm $a$ does not return a PIIF and weakly ex-ante stable solution, and we have obtained a contradiction.
\end{enumerate}
\end{proof}

\paragraph{Hospitals propose.}
Next, we prove that no 1-LAHP outputs a PIIF and weakly ex-ante stable allocation for the setting of mutual replacement individual fairness and a general metric.
The proof is very similar to the proof of Theorem \ref{thm:unfair_prefs_lahp}. There are four hospitals $i, j, k$ and $l$, in the proof of Theorem \ref{thm:unfair_prefs_lahp} the preferences of hospital $A$ were $j \succ_A k \succ_A i \succ_A l$. In this case, since we require fair preferences, and the preferences of hospital $A$ are $j \,\Tilde{\succ}_A\, k \,\Tilde{\succ}_A\, i \,\Tilde{\succ}_A\, l$.

If all doctors rank hospital $A$ first, then any proposal that $A$ makes in the first round will be accepted and not rejected in the entire run. Moreover, we arrange the other preferences such that there is a single weakly ex-ante stable and PIIF allocation for $A$. Thus, hospital $A$ must propose in the first round to its outcome in this PIIF and weakly ex-ante stable allocation which is 
$$\begin{cases}
        i,&\textit{w.p. } 1/3\\
        j,&\textit{w.p. } 2/3
\end{cases}$$.

If doctor $j$ prefers hospital $D$ over hospital $A$ and hospital $D$ ranks doctor $j$ first, then we arrange the other preference such that in the only PIIF and weakly ex-ante stable allocation, hospital $A$'s prospect is $\set{k,\textit{w.p.}1}$, and hospital $A$ is still doctor $i$'s favorite hospital. However, hospital $A$ cannot distinguish between these two scenarios and proposes to doctor $i$ in the first round. Once doctor $i$ accepted hospital $A$'s proposal, doctor $i$ will never reject it. Thus, hospital $A$'s prospect is different from $\set{k,\textit{w.p.}1}$, although it is the only PIIF and weakly ex-ante stable allocation.

\begin{theorem}
In the setting of a general metric and mutual replacement individually fair preferences, there is no $a \in 1$-LAHP that returns a PIIF and weakly ex-ante stable allocation.
\end{theorem}

\begin{proof}
Assume for contradiction that there exists a 1-LAHP $a$ that returns a PIIF and weakly ex-ante stable allocation. Consider the following example:

Let $\D = \set{i, j, k, l}$ be the set of doctors, $\H = \set{A, B, C, D}$ be the set of hospitals and all the distances are 1 except: $d(i, j) = 1/3, d(k, l) = 1/3$.

The preferences of doctors $i$ and $k$ are: 
\begin{align*}
&A \succ_i B \succ_i C \succ_i D, \\
&A \succ_k B \succ_k C \succ_k D.
\end{align*}
The preferences of the hospitals are:
\begin{align*}
&j \,\Tilde{\succ}_A\, k \,\Tilde{\succ}_A\, i \,\Tilde{\succ}_A\, l,\\
&j \,\Tilde{\succ}_B \,k \,\Tilde{\succ}_B\, i \,\Tilde{\succ}_B \,l,\\
&j \,\Tilde{\succ}_D \,l \,\Tilde{\succ}_D\, i \,\Tilde{\succ}_D\, k,\\
&k \,\Tilde{\succ}_C \,j \,\Tilde{\succ}_C \,l \,\Tilde{\succ}_C\, i.
\end{align*}

We divide into two cases:
\begin{enumerate}
    \item The preference of doctor $j$ are 
    $$A \succ_j B \succ_j C \succ_j D$$
    and the preferences of doctor $l$ are
    $$A \succ_l B\succ_l D \succ_l C.$$
    First, we show that in any PIIF and weakly ex-ante stable solution $\pi_1^*$ for these preferences that satisfies that
    $$
    \pi^*_1(A) = \begin{cases}
        i,&\textit{w.p. } 1/3 \\
        j,&\textit{w.p. } 2/3.
    \end{cases}
    $$
    Let the allocation $\pi$ be some PIIF and weakly ex-ante stable allocation, assume that 
    $$\pi(A) \ne \begin{cases}
        i,&\textit{w.p. } 1/3 \\
        j,&\textit{w.p. } 2/3.
    \end{cases}$$
    We claim that the selectively fair alternative allocation $$\nu_{A,i,j} = \left(A, \set{i,j}, \sigma_{i,j} = \begin{cases}
        i,&\textit{w.p. } 1/3 \\
        j,&\textit{w.p. } 2/3
    \end{cases}\right)$$
    is active. This is because:
    \begin{itemize}
        \item The distribution $\sigma_{i,j}$ is IF.
        \item Hospital $A$ is the most preferred hospital by all the doctors. Thus, in any PIIF solution 
        \begin{align*}
            |\Pr[\pi(A) = j] - \Pr[\pi(A) = i]| \le 1/3, \quad |\Pr[\pi(A) = k] - \Pr[\pi(A) = l]| \le 1/3.
        \end{align*}
        
        Under these constraints, the prospect $\sigma_{i,j}$ stochastically dominates any other prospect for hospital $A$ subject to PIIF.
        \item Hospital $A$ is doctors $i$ and $j$'s most preferred hospital.
        Thus, hospital $A$ is preferred over any other hospital in their prospect.
    \end{itemize}
    
    Hospital $A$ is the most preferred hospital by all the doctors. Thus, for any clustering hospital $A$ chooses, any proposals it makes on the first round would be accepted uniformly by all the doctors in the proposed clusters. More than that, for the same reason, this accepted mass would not be rejected or reallocated in any later round. Thus, if hospital $A$ is assigned at the end of the run to the allocation $\sigma_{i,j}$, hospital $A$ must choose a clustering where doctors $i$ and $j$ are in different singleton clusters $C_i$ and $C_j$ respectively, and propose probability mass 1/3 to cluster $C_i$ and probability mass $2/3$ to cluster $C_j$ on the first round.
    
    \item The preference of doctor $j$ are $$D \succ_j A \succ_j B \succ_j C$$
    and the preferences of doctor $l$ are 
    $$C \succ_l B\succ_l A \succ_l D.$$
    
    We claim that the only PIIF and weakly ex-ante stable solution for these preferences is 
    $$
    \pi^*_2 = \set{(k, A), (j, D), (l, C), (i, B),\textit{ w.p. } 1}.
    $$
    This is because doctor $j$ ranks hospital $D$ first, hospital $D$ ranks doctor $j$ first and doctor $i$ ranks hospital $D$ last. Thus, doctor $j$ and hospital $D$ must be matched. Otherwise, the selectively fair alternative allocation $(D, \set{j}, \set{j, \textit{ w.p. }1}$ is active.
    
    Similarly, the pairs $(C,l)$, $(A,k)$ and $(B,i)$ rank each other first between individuals whose prospect has not been determined yet, and are at distance 1 from $\D\backslash\set{j}$. Thus, algorithm $a$ must match them with probability 1.
    
    However, we know that in the first round, hospital $A$ proposes probability mass $1/3$ to doctor $i$ and $i$ must accept, since $A$ is its most preferred hospital. From the same reason, doctor $i$ will never reject or reallocate any probability mass from hospital $A$. In the allocation $\pi^*_2$, doctor $i$ is matched to hospital $A$ with probability 0.
    Thus, algorithm $a$ does not return the only PIIF and weakly ex-ante stable solution. 
\end{enumerate}
\end{proof}

\paragraph{Extending to general $\alpha$.}
Next, we generalize the proof above to any value of $\alpha$ in $(0,1]$. The difference between the two cases is that now, if $\alpha \le 2/3$, $A$ would want to start by proposing $\alpha$ probability mass to doctor $j$ and only if $j$ accepts, then it proposes to doctor $i$. However, if we could ``delay'' hospital $D$ so it proposes to doctor $j$ only after $\ceil{\frac{2/3}{\alpha}}$ rounds, then hospital $A$ would have no other option than proposing some probability mass to doctor $i$, before being rejected by doctor $j$. This puts us in a similar situation to the proof of the $\alpha = 1$ case, and we show the algorithm will not reach a fair and stable solution. 
For simplicity, in the proof we delay hospital $D$ by $\ceil{\frac{1}{\alpha}}$ rounds.

\begin{theorem}
\label{thm:fair_prefs_lahp}
There is no $a \in \alpha$-LAHP that returns a PIIF and weakly ex-ante stable allocation in the setting of a general metric and mutual replacement individually fair hospital preferences.
\end{theorem}

\begin{proof}
Assume for contradiction that there exists a $\alpha$-LAHP $a$, for $\alpha \in (0,1]$, that returns a PIIF and weakly ex-ante stable allocation. Consider the following example:

Define $\beta = \ceil{\frac{1}{\alpha}}$. Let $\D = \set{i, j, k, l, d_1, ..., d_{\beta}}$ be the set of doctors, $\H = \set{A, B, C, D, H_1, ..., H_{\beta}}$ be the set of hospitals and all the distances are 1 except: $d(i, j) = 1/3, d(k, l) = 1/3$.

The preferences of the doctors (besides doctors $j$ and $l$) are: 
\begin{align*}
&A \succ_i B \succ_i C \succ_i H_1 \succ_i ... \succ_i H_{\beta} \succ_i D,\\
&A \succ_k B \succ_k C \succ_k D \succ_k H_1 \succ_k ... \succ_k H_{\beta},\\
\forall m \in [{\beta}] :& A \succ_{d_m} H_m \succ_{d_m} H_1 \succ_{d_m}... \succ_{d_m} H_{m-1} \succ_{d_m} H_{m+1} \succ_{d_m}...\succ_{d_m}  H_{\beta} \succ_{d_m} B \succ_{d_m} C \succ_{d_m} D.
\end{align*}
The preferences of the hospitals are:
\begin{align*}
&j \,\Tilde{\succ}_A\, k \,\Tilde{\succ}_A\, i \,\Tilde{\succ}_A\, l \succ_A d_1 \succ_A ... \succ_A d_{\beta},\\
&j \,\Tilde{\succ}_B\, k \,\Tilde{\succ}_B\, i \,\Tilde{\succ}_B\, l \succ_B d_1 \succ_B ... \succ_B d_{\beta},\\
&k \,\Tilde{\succ}_C\, j \,\Tilde{\succ}_C\, l \,\Tilde{\succ}_C\, i \succ_C d_1 \succ_C ... \succ_C d_{\beta},\\
&d_1 \succ_D ... \succ_D d_{\beta} \succ_D j \,\Tilde{\succ}_D\, l \,\Tilde{\succ}_D\, i \,\Tilde{\succ}_D\, k,\\
\forall m \in [{\beta}] :& d_m \succ_{H_m} d_1  \succ_{H_m}... \succ_{H_m} d_{m-1} \succ_{H_m} d_{m+1} \succ_{H_m}...  \succ_{H_m} d_{\beta} \succ_{H_m} i \,\Tilde{\succ}_{H_m}\, k \,\Tilde{\succ}_{H_m}\, j \,\Tilde{\succ}_{H_m}\, l.
\end{align*}

We divide into two cases:
\begin{enumerate}
    \item The preference of doctor $j$ are $$A \succ_j B \succ_j C \succ_j D \succ_j H_1 \succ_j ... \succ_j H_{\beta},$$
    and the preferences of doctor $l$ are
    $$A \succ_l B\succ_l D \succ_l C \succ_l H_1 \succ_l ... \succ_l H_{\beta}.$$
    First, we show that in any PIIF and weakly ex-ante stable solution $\pi^*_1$ for these preferences satisfies that
    $$
    \pi^*_1(A) = \begin{cases}
    i,&\textit{w.p. } 1/3 \\
    j,&\textit{w.p. } 2/3.
    \end{cases}
    $$
    Let the allocation $\pi$ be some PIIF and weakly ex-ante stable allocation, assume that 
    $$\pi(A) \ne \begin{cases}
    i,&\textit{w.p. } 1/3 \\
    j,&\textit{w.p. } 2/3.
    \end{cases}$$
    We claim that the selectively fair alternative allocation $$\nu_{A,i,j} = \left(A, \set{i,j}, \sigma_{i,j} =  \begin{cases}
    i,&\textit{w.p. } 1/3 \\
    j,&\textit{w.p. } 2/3
    \end{cases}\right)$$
    is active. This is because:
    \begin{itemize}
        \item The distribution $\sigma_{i,j}$ is IF.
        \item Hospital $A$ is doctors $i$ and $j$'s most preferred hospital. Thus, they prefer hospital $A$ over any hospital in their prospect in the allocation $\pi$.
        \item For any doctor $d \in \set{d_1, ..., d_{\beta}}$, replacing doctor $d$ by doctors $i$ or $j$ stochastically dominates hospital $A$'s prospect in the allocation $\pi$. Since the allocation $\pi$ is PIIF and all the doctors rank hospital $A$ first, \begin{align*}
            |\Pr[\pi(A) = j] - \Pr[\pi(A) = i]| \le 1/3, \quad |\Pr[\pi(A) = k] - \Pr[\pi(A) = l]| \le 1/3.
        \end{align*} 
        Under these constraints, the prospect $\sigma_{i,j}$ maximizes over the rank of hospital $A$'s prospect.
    \end{itemize}
    
    Hospital $A$ is the most preferred hospital by all the doctors. Thus, for any clustering hospital $A$ chooses, any proposals it makes would be accepted by all the doctors in the proposed clusters. More than that, for the same reason, this accepted mass would not be rejected or reallocated in any later round. At the end of the run, hospital $A$ is assigned to the prospect $\sigma_{i,j}$. 
    This implies that hospital $A$ must propose probability mass only to clusters that contain the doctors $i$ and $j$, and it must propose at least $\alpha$ probability mass at each round.
    Thus, after $\beta = \ceil{\frac{1}{\alpha}}$ rounds, hospital $A$'s prospect must be $\sigma_{i,j}$.
    
    \item The preference of doctor $j$ are 
    $$D \succ_j A \succ_j B \succ_j C \succ_j H_1 \succ_j ... \succ_j H_{\beta},$$
    and the preferences of doctor $l$ are 
    $$C \succ_l B\succ_l A \succ_l D \succ_l H_1 \succ_l ... \succ_l H_{\beta}.$$
    
    We show that the only PIIF and weakly ex-ante stable solution for these preferences is 
    $$
    \pi^*_2 = \set{(k, A), (j, D), (l, C), (i, B), (d_1, H_1), ...,(d_{\beta},H_{\beta}): 1}.
    $$
    This is because doctor $j$ ranks hospital $D$ first, hospital $D$ ranks doctor $j$ first and doctor $i$ ranks hospital $D$ last. Thus, doctor $j$ and hospital $D$ must be matched. Otherwise, the selectively fair alternative allocation $(D, \set{j}, \set{j, \textit{ w.p. }1})$ is active. 
    
    Similarly, the pairs $(C,l)$, $(A,k)$, $(B,i)$, $\forall m \in [\beta]:(H_m,d_m)$, rank each other first between individuals that their prospect had not been determined yet (determining by this order), and are in distance 1 from every doctor in $\D \backslash \set{j}$. Thus, algorithm $a$ must match these pairs with probability 1.
    
    Assume hospital $A$ does not see any difference from the first case in the first $\beta$ rounds. Then, in one of the first $\beta$ rounds, hospital $A$ must propose 1/3 probability mass to doctor $i$, and $i$ must accept since it ranks hospital $A$ first. For the same reason, doctor $i$ will never reject hospital $A$, and then algorithm $a$ does not return the only PIIF and weakly ex-ante stable solution. Hospital $A$ notices the difference in the preferences when hospital $D$ proposes some probability mass to doctor $j$, and doctor $j$ rejects hospital $A$ some probability mass from $A$.
    
    Thus, finish the proof, we have to show that hospital $D$ proposes to doctor $j$ after at least $\beta+1$ rounds.
    We show that hospital $D$ must propose to the doctors in $\D \backslash \set{i, j, k, l}$ by order of its preferences, and for this cluster each of them in a singleton. Otherwise, algorithm $a$ might return an unstable allocation. We show this by induction over the doctors in $\D \backslash \set{i, j, k, l}$ by order of preferences. 
    
    \textbf{Base:} The first doctor that hospital $D$ proposes must be $d_1$. Otherwise, if hospital $D$ is ranked highest by all the doctors, any other doctor that $D$ proposes will accept, and the output will be unstable. 
    In the first round, hospital $D$ cannot distinguish this case from any other case.
    Thus, hospital $D$ must cluster $d_1$ in a singleton, and propose to it in the first round.
    
    \textbf{Step:} The $m$-th doctor that hospital $D$ proposes must be ranked $m$-th by $D$. Denote by $\widetilde{r}_D$, hospital $D$'s preferences in terms of stochastic domination.
    Assume from the induction assumption that until step $t$, hospital $D$ proposed to clusters that contain doctors from the set $$S_{m-1} = \set{\widetilde{r}_D(m'): m' \le m-1}.$$
    We claim that at step $t+1$, hospital $D$ proposes to clusters that contain doctors from $S_{m}$. Otherwise, consider the case where for every doctor $d \in S_{m-1}$, $r_d^{-1}(D) = |\D|$ there exists some hospital $h \in \H\backslash \set{D}$ such that $r_d^{-1}(h) = 1$ and  $\widetilde{r}_h^{-1}(d) = 1$.
    For every doctor $d \in S_{m-1}$, in any weakly ex-ante stable allocation, doctor $d$ must be matched to the hospital they ranked first with probability 1. 
    
    In round $t+1$, hospital $D$ does not have any knowledge about the preferences of the doctors in $\D \backslash S_{m-1}$. Assume that for every doctor $d \in \D\backslash S_{m-1}$, $r_d^{-1}(D) = 1$.
    Then, since they rank hospital $D$ first, any doctor from $\D\backslash S_m$ that hospital $D$ proposes is matched to hospital $D$ in the final allocation. However, the only weakly ex-ante stable solutions are those where hospital $D$ is matched to its $m$-th preferred doctor. Thus, hospital $D$ cannot propose at step $t+1$ to any doctor in $\D \backslash S_m$. 
    Hospital $D$ has no knowledge about the preferences of the doctors in $\D \backslash S_m$ in round $t+1$, therefore, hospital $D$ must propose only to clusters that contain doctors from $S_{m}$.
\end{enumerate}
\end{proof}

We are now set to prove Theorem \ref{thm:gen_metrics_fail}.

\begin{proof}[Proof of Theorem \ref{thm:gen_metrics_fail}]
The proof is immediate from Theorem \ref{thm:fair_prefs_lahp} and Theorem \ref{claim:fair_prefs_ladp}.
\end{proof}
\section{Open Questions}
\label{sec:conclussion}

We leave several questions open. We do not resolve the question of whether fairness and stability are compatible for general metrics or unfair preferences. Specifically, fair and stable solutions exist in our negative examples, and it is only that the local-proposing algorithm fails to find it. 
Proving the general existence of fair and stable solutions and whether they can be obtained by different algorithms than the ones we considered is an interesting direction for future work.

Also, the negative results presented in Section \ref{sec:failure_fair_prefs}, where the hospital preferences are individually fair, brings to mind that this might not be the right way to enforce fairness over the preferences. The inherent problem in the examples in Section \ref{sec:impossibility} is that in some cases, a hospital has to be matched to a doctor it does not want to in order to achieve its best outcome, and in some not. The fact that the hospital does not want to be matched to the doctor but has to implies that the preferences are not really fair, although they do satisfy the criteria of individual fairness.




\section{Acknowledgements}

We thank Shahar Dobzinski and Moni Naor for helpful conversations and suggestions about this work and its presentation. We also thank the anonymous ITCS 2022 reviewers for their feedback.

\bibliographystyle{alpha}
\newcommand{\etalchar}[1]{$^{#1}$}

\newpage
\appendix
\section{Generalization of Compatibility Result}
\label{sec:reduction}

In this section, we present a generalization for the result from Section \ref{sec:piif_algs}. We show that for a more relaxed requirement for the hospital preferences, we can still guarantee fairness and stability for the setting of a proto-metric, using \Cref{alg:GS_PSP} and \Cref{alg:GS_WA}. 
We show this for the general stability definition, set contract stability, see \cref{def:contract_stability_general}.

First, we present another definition for IF preferences, \emph{rank individual fairness}, which is a relaxation of both Definition \ref{def:mutual_replacement_individual_fairness} and Definition \ref{def:strict_IF}.
This definition guarantees that for every two doctors $i$ and $j$, the statistical distance between the rank distribution induced by the preferences of any hospital $h$ for doctors $i$ and $j$ is bounded by their distance according to the metric.

\begin{definition}[Rank Individual Fairness Preferences]
\label{def:rank_IF}
Let ${\D}$ be a set of doctors, $h \in \H$ a hospital, $d : {\D} \times {\D} \to [0,1]$ a pseudometric and let $r_h$ be $h$'s (probabilistic) preference. Preference $r_h$ is considered rank individually fair if for every two doctors $i, j \in {\D}$,
$$D_{TV}(r_h^{-1}(i), r_h^{-1}(j)) \le d(i,j).
$$
\end{definition}
\begin{corollary}
If $d$ is a proto metric and $h \in \H$ is a hospital with rank individually fair preferences then for every $i, j \in {\D}$ such that $d(i,j) = 0$, i.e. $i$ and $j$ are in the same cluster,
$$\forall k \in [n] : \Pr[r_h(k) = i] = \Pr[r_h(k) = j].$$
\end{corollary}

\subsection{Reduction from Rank IF Preferences to Strict IF Preferences}
We present an algorithm that takes as input rank individually fair preferences and outputs strict individually fair preferences. We show that any allocation that is PIIF and contract stable for the strict individually fair preferences is also PIIF and set contract stable for the rank individually fair preferences.

For a fixed hospital $h \in \H$, algorithm \ref{alg:Rank2Cluster} takes as input rank individually fair preferences and returns a function $o_h : \C \times \C \to \set{-1, 1, 0}$. If for two clusters $C_1, C_2 \in \mathcal{C}$ the preferences of the hospital are exactly the same, then the function will return 0, i.e., hospital $h$ is indifferent. Otherwise, the algorithm finds the minimal value of rank $r \in [n]$ such that there exists doctors $i \in C_1$ and $j \in C_2$
$$
\Pr[r_h^{-1}(i) \le r] \ne \Pr[r_h^{-1}(j) \le r].
$$
If $$\Pr[r_h^{-1}(i) \le r] > \Pr[r_h^{-1}(j) \le r]$$ then $o_h(C_1, C_2) = 1$, otherwise, $o_h(C_1, C_2) = -1$.

From the function $o_h$, we can conclude an order over the clusters, if $o_h(C_1, C_2) = 1$ then $C_1 \succ_h C_2$, if $o_h(C_1, C_2) = -1$, then $C_2 \succ_h C_1$, and if $o_h(C_1, C_2) = 0$ the order is chosen arbitrarily.

\begin{algorithm}[H]
    \caption{Rank IF to cluster IF ($RankToCluster(P, \mathcal{C}, R, r)$)}
    \label{alg:Rank2Cluster}
        \hspace*{\algorithmicindent} \textbf{Input} Rank IF preferences $r_h:[R] \to \D$, set of clusters $\mathcal{C}$, number of ranks $R$, current rank $r$ 
    \begin{algorithmic}[1]
    \IF{$r = R + 1$ or $|\mathcal{C}| = 1$}
        \STATE Define the function $o : \mathcal{C} \times \mathcal{C} \rightarrow \{-1,0,1\}$ such that 
        $$
        \forall C_1, C_2 \in \mathcal{C}: o(C_1,C_2) = 0.
        $$
        \RETURN $o$.
    \ELSE
        \STATE Define $p_{C,r} \leftarrow \Pr_{r_h}[r_h(r) \in C] / |C|$.
        \hfill\COMMENT The probability of each doctor in $C$ to be ranked $r$-th.
        \STATE Create a partition $\mathcal{C}_1 , ..., \mathcal{C}_k$ over $\mathcal{C}$ such that $\mathcal{C}_i = \{C \in \mathcal{C}: p_{C,r} = p_i\}$ and $p_1 > p_2 > ... > p_k$.
        \FORALL{$i \in [k]$}
            \STATE $o_i \leftarrow RankToCluster(r_h, \mathcal{C}_i, R, r+1)$.
        \ENDFOR
        \STATE Define $o : \mathcal{C} \times \mathcal{C} \rightarrow \{-1, 0, 1\}$ such that $\forall i, j \in [k], C_1 \in \mathcal{C}_i, C_2 \in \mathcal{C}_j$,
        $$o(C_1, C_2) = \begin{cases}
        -1, & i > j \\
        o_{i}(C_1, C_2), & i = j \\
        1, & i > j
        \end{cases}$$
        \RETURN $o$.
    \ENDIF
    \end{algorithmic}
\end{algorithm}

\begin{claim}
Let $\H$ be a set of hospitals, $\D$ be a set of doctors, $d : \D \times \D \rightarrow \{0,1\}$ be a proto-metric, $\mathcal{C}$ be a partition to clusters over $\D$ and $\mathcal{P} = \set{r_h : [n] \to \D}_{h\in\H}$ be rank individually fair preferences of the hospitals over $\D$. 
If for every $h \in \H$, we run \Cref{alg:Rank2Cluster} over $r_h, \mathcal{C}, n, 1$, and receive a preferences over the partition $\mathcal{C}$, $\mathcal{P}_\mathcal{C} = \set{\widetilde{r}_h:[|\C|] \to \C}_{h\in\H}$ for every $h \in \H$. 
Then every allocation that is contract stable with respect to $\mathcal{P}_\mathcal{C}$ is set contract stable with respect to $\mathcal{P}$.
\end{claim}

\begin{proof}
Let $\pi$ be a contract stable allocation with respect to $\mathcal{P}_\mathcal{C}$. Assume for contradiction there exists an active set contract $\mu = (h, \D', a, \sigma) \in \H \times \mathcal{P}(\D)\times [0,1] \times \Delta(\D\backslash\D')$ with respect to $\mathcal{P}$. Denote $\pi^{\mu}$ as the allocation $\pi$ after activating $\mu$.

In this proof, we assume without loss of generality that for every two clusters $C_1, C_2 \in \C$, $o_h(C_1, C_2) \ne 0$. Otherwise, we can define $\widetilde{\C}$ to be the set of clusters where every two clusters $C_1, C_2 \in \C$ where $o_h(C_1, C_2) = 0$ are united to the same cluster and continue this proof over $\widetilde{\C}$.

Since $\mathcal{P}_\C$ is strict individually fair and $d$ is a proto-metric we use contract stability (Definition \ref{def:contract_stability_proto}) for the stability of $\mathcal{P}_\C$.

Let's denote by $\mathcal{S}_\mu \in \D \times \D \times \H$ the set of ``possible contracts'' with respect to the set contract $\mu$: 
For every $(i, i', h') \in \mathcal{S}_\mu$, doctor $i'$ is in the set $\D'$, doctor $i$ is in the support of $\sigma$, $h \succ_i h'$ and $\Pr[\pi(h) = i' \wedge \pi(i) = h'] > 0$. 
In other words, $\mathcal{S}_\mu$ is the set of doctors and hospitals such that there is an event (with positive probability) where hospital $h$ is matched to doctor $i'$ and proposes to doctor $i$, and doctor $i$ accepts since doctor $i$ prefer hospital $h$ over its current match, hospital $h'$.

Since the allocation $\pi$ is contract stable with respect to $\mathcal{P}_\C$, for every $(i, i', h') \in \mathcal{S}_\mu$, we know that $\widetilde{r}_h^{-1}(C) \ge \widetilde{r}_h^{-1}(C')$ where clusters $C$ and $C'$ are the clusters that contain doctors $i$ and $i'$ respectively.

We also know that $\mathcal{S}_\mu$ is not empty and there is at least one set $(i, i', h') \in \mathcal{S}_\mu$ where $\widetilde{r}_h^{-1}(C) > \widetilde{r}_h^{-1}(C')$ since otherwise $\pi(h) \succeq_h \pi^\mu(h)$ which contradicts the assumption that $\mu$ is active.
For any set $(i, i', h') \in \mathcal{S}_\mu$ where $\widetilde{r}_h^{-1}(C) > \widetilde{r}_h^{-1}(C')$, there exists $r_{i,i'} \in [|\C|-1]$ such that
\begin{align*}
    \forall r < r_{i,i'}: &\Pr[r_h^{-1}(i) \le r] = \Pr[r_h^{-1}(i') \le r] \\
    &\Pr[r_h^{-1}(i) \le r_{i,i'}] < \Pr[r_h^{-1}(i') \le r_{i,i'}]
\end{align*}
For every set where $\widetilde{r}_h^{-1}(C) = \widetilde{r}_h^{-1}(C')$, for convenience, we denote $r_{i,i'} = |\C|$.

Let's denote by $r^*$ the minimal rank that satisfy the above for some set  $i^*, {i^*}', h^*$, i.e. 
\begin{align*}
    &r^* = \min_{(i, i', h') \in \mathcal{S}_\mu}r_{i,i'} \\
    &i^*, {i^*}', h^* \in \argmin_{(i, i', h') \in \mathcal{S}_\mu}r_{i,i'}.
\end{align*}
Then we have that
\begin{align*}
    \forall (i, i', h') \in \mathcal{S}_\mu: &\Pr[r_h^{-1}(i) \le r^*] \le \Pr[r_h^{-1}(i') \le r^*]\\
    &\Pr[r_h^{-1}(i^*) \le r^*] < \Pr[r_h^{-1}({i^*}') \le r^*].
\end{align*}
This implies that
\begin{align*}
    &\Pr[r_h^{-1}(\pi^\mu(h)) \le r^*] - \Pr[r_h^{-1}(\pi^\mu(h)) \le r^*] = \\
    &\sum_{(i, i', h') \in \mathcal{S}_\mu}\Pr[\pi(h) = i' \wedge \pi(i) = h']\Pr[\sigma = i](\Pr[r_h^{-1}(i) \le r^*] -\Pr[r_h^{-1}(i') \le r^*]) \le\\
    &\Pr[\pi(h) = {i^*}{}' \wedge \pi(i^*) = {h^*}{}']\Pr[\sigma = i](\Pr[r_h^{-1}(i^*) \le r^*] -\Pr[r_h^{-1}({i^*}{}') \le r^*]) < 0
\end{align*}

Thus, $\pi^\mu(h) \not\succ_h \pi(h)$ with contradiction to the assumption that $\mu$ is an active set contract.
\end{proof}

\begin{corollary}
Given a set of rank individually fair preferences $\mathcal{P}$ we can run \Cref{alg:Rank2Cluster} and get a set of strict individually fair preferences $\mathcal{P}_\C$. Then we can run \Cref{alg:GS_PSP} or \Cref{alg:GS_WA} and get a fair and contract stable allocation with respect to $\mathcal{P}_\C$ and it will also be fair and set contract stable with respect to $\mathcal{P}$.
\end{corollary}

\end{document}